\documentclass[12pt]{article}
\usepackage[utf8]{inputenc}
\usepackage[letterpaper]{geometry}
\usepackage{xcolor}
\usepackage{array}
\usepackage{float}
\usepackage{mathrsfs}
\usepackage{multirow}
\usepackage{amsmath}
\usepackage{amsthm}
\usepackage{graphicx}
\usepackage{setspace}
\usepackage[authoryear,round]{natbib}
\usepackage[unicode=true,
 bookmarks=false,
 breaklinks=false,pdfborder={0 0 0},pdfborderstyle={},backref=false,colorlinks=true]
 {hyperref}
\hypersetup{
 citecolor=blue, linkcolor=blue, urlcolor=blue}

\makeatletter

\providecommand{\tabularnewline}{\\}

\numberwithin{equation}{section}
\numberwithin{figure}{section}
\theoremstyle{definition}
\newtheorem{defn}{\protect\definitionname}
\theoremstyle{remark}
\newtheorem{rem}{\protect\remarkname}
\theoremstyle{plain}
\newtheorem{prop}{\protect\propositionname}
\theoremstyle{plain}
\newtheorem{assumption}{\protect\assumptionname}
\theoremstyle{remark}
\newtheorem{claim}{\protect\claimname}
\theoremstyle{plain}
\newtheorem{cor}{\protect\corollaryname}
\theoremstyle{plain}
\newtheorem{fact}{\protect\factname}
\theoremstyle{definition}
 \newtheorem{example}{\protect\examplename}
\theoremstyle{plain}
\newtheorem{lem}{\protect\lemmaname}

\usepackage[T1]{fontenc}
\usepackage{lmodern}
\definecolor{green}{RGB}{00, 180, 00}
\definecolor{red}{RGB}{180, 00, 00}
\usepackage{amsfonts}

\makeatother

\providecommand{\assumptionname}{Assumption}
\providecommand{\claimname}{Claim}
\providecommand{\corollaryname}{Corollary}
\providecommand{\definitionname}{Definition}
\providecommand{\examplename}{Example}
\providecommand{\factname}{Fact}
\providecommand{\lemmaname}{Lemma}
\providecommand{\propositionname}{Proposition}
\providecommand{\remarkname}{Remark}

\begin{document}
\title{Naive Analytics Equilibrium{\large{} (Draft)}}
\author{Ron Berman\thanks{The Wharton School, University of Pennsylvania. email: \protect\href{mailto:ronber@wharton.upenn.edu}{ronber@wharton.upenn.edu}}
\and Yuval Heller\thanks{Department of Economics, Bar-Ilan University. email: \protect\href{mailto:yuval.heller@biu.ac.il}{yuval.heller@biu.ac.il}. }
\thanks{We would like to express our deep gratitude to Eduardo Azevedo, Jorge
Al$\acute{\textrm{e}}$-Chilet, Arthur Fishman, Yossi Spiegel, Steve
Tadelis, Kanishka Misra, Kinshuk Jerath, Aviv Nevo, and seminar participants
at Bar-Ilan University, Tel Aviv University, the Hebrew University
of Jerusalem, University of California, Berkeley, the University of
Pennsylvania, the TSEInfo Conference and the Virtual Quant Marketing
Seminar for many useful comments. Yuval Heller is grateful to the
European Research Council for its financial support (Starting Grant
\#677057). }}
\maketitle
\begin{abstract}
\noindent We study interactions with uncertainty about demand sensitivity.
In our solution concept (1) firms choose seemingly-optimal strategies
given the level of sophistication of their data analytics, and (2)
the levels of sophistication form best responses to one another. Under
the ensuing equilibrium firms underestimate price elasticities and
overestimate advertising effectiveness, as observed empirically. The
misestimates cause firms to set prices too high and to over-advertise.
In games with strategic complements (substitutes), profits Pareto
dominate (are dominated by) those of the Nash equilibrium. Applying
the model to team production games explains the prevalence of overconfidence
among entrepreneurs and salespeople.\textbf{}\\
\textbf{Keywords}: Advertising, pricing, data analytics, strategic
distortion, strategic complements, indirect evolutionary approach.\textbf{}\\
\textbf{JEL Classification}: C73, D43, M37.
\end{abstract}

\section{Introduction}

Researchers often assume that better measurement and accurate estimations
of the sensitivity of demand allow firms to improve their advertising
and pricing decisions. Arriving at such accurate estimations requires
careful experimental techniques or sophisticated econometric methods
that correct for the endogeneity of decision variables in the empirically
observed data (see, e.g., \citealt{blake2015consumer,shapiro2019generalizable,gordon2019comparison,sinkinson2019ask}
who estimate advertising effectiveness, and \citealp{berry1994estimating,nevo2001measuring,ale2019trade}
who estimate price elasticities). 

Despite the emphasis on precision and unbiasedness by researchers,
many companies have been slow to adopt these techniques \citep{rao2019firms},
often questioning the benefit of causal inference and precise measurement.
This reluctance to measure effectiveness precisely is often attributed
to implementation difficulties, lack of knowledge and cognitive limitations
by decision makers, or moral hazard \citep{berman2018beyond,fredrik2020}.
Empirically, we often observe that firm advertising budget allocations
are consistent with over-estimation of advertising effectiveness (see,
e.g., \citealp{blake2015consumer,lewis2015unfavorable,golden2020effects}),
while pricing decisions are consistent with under-estimation of price
elasticities (see, e.g., \citealt{besanko1998logit,chintagunta2005beyond,villas1999endogeneity};
see also \citealp{hansen2020algorithmic} who demonstrate that common
AI pricing techniques induce ``too-high'' prices ).

In this paper we challenge the assumption that better estimates are
always beneficial for firms. Our results show that in many cases firms
are better off with biased, less precise, measurements because of
strategic considerations in oligopolistic markets. In equilibrium
firms will converge to biased measurements because their profits are
maximized when they act on these measurements. Moreover, the directions
of the biases, as predicted in our model, fit the empirically observed
behavior of firms well.

\paragraph{Highlights of the Model}

We present a model where the payoffs of competing players (firms)
each depend on her actions and on her demand, where the demand depends
on the actions of all players. The players do not know the demand
function, but can select actions and observe the realized demand.
The game has two stages. In stage 1 each player hires a (possibly
biased) analyst to estimate the sensitivity of demand. An analyst
may under- or over-estimate the sensitivity of demand. In stage 2
each player chooses an action taking the estimate as the true value. 

Our solution concept, called a Naive Analytics Equilibrium, is a
profile of analysts' biases and a profile of actions, such that (1)
each action is a perceived best-reply to the opponents' actions, given
the biased estimation, and (2) each bias is a best-reply to the opponents'
biases in the sense that if a player deviates to another bias this
leads to a new second-stage equilibrium, in which the deviator's (real)
profit is weakly lower than the original equilibrium payoff. The first-stage
best-replying is interpreted as the result of a gradual process in
which firms hire and fire analysts from a heterogeneous pool, and
each firm is more likely to fire its analyst if its profit is low.

\paragraph{Summary of Results}

Our model is general enough to capture price competition with differentiated
goods (where the goods can be either substitutes or complements),
as well as advertising competition (where the advertising budget of
one firm has externalities on the competitors' demand). Our results
show that firms hire biased analysts in any naive analytics equilibrium,
and that the direction of the bias is consistent with the empirically
observed behavior of firms: in price competition firms hire analysts
that under-estimate price elasticities, and in advertising competition
firms hire analysts who over-estimate the effectiveness of advertising.\footnote{Our general result on the direction of the bias is related to \citeauthor{heifetz2007dynamic}'s
\citeyearpar{heifetz2007dynamic} result in the setup of evolution
of subjective preferences. We discuss this relation in Remark \ref{rem:Heifetz-et-al}
in Section \ref{subsec:Direction-of-Analytics}.} 

We also show that there is a Pareto-domination relation between the
naive analytics equilibrium and the Nash equilibrium (of the game
without biases), where its direction depends on the type of strategic
complementarity. In a game with strategic complements (i.e., price
competition with differentiated goods) the naive analytics equilibrium
Pareto dominates the Nash equilibrium, while the opposite holds in
a game with strategic substitutes (i.e., advertising competition with
negative externalities). The intuition is that in a game with strategic
complements (resp., substitutes), each player hires a naive analyst
that induces a biased best reply in the direction that benefits (resp.,
harms) the opponents. This is so because these biases have a strategic
advantage of inducing the opponents to change their strategies in
the same (resp., opposite) direction, which is beneficial to the player. 

Next, we analyze a standard functional form in each type of competition,
and we explicitly characterize the unique naive analytics equilibrium
in price competitions and in advertising competitions. We also present
implications for market structure analysis when firms compete in prices
in Section \ref{sec:Implication-for-Merges}. Finally, we demonstrate
that our model can be applied in more general settings. Specifically,
in Section \ref{sec:Team-Production} we apply the model to a game
of team production with strategic complementarity (see, e.g., \citealp{holmstrom1982moral,cooper1988coordinating}).
We present two testable predictions in this setup: (1) players are
overconfident in the sense of overestimating their ability to contribute
to the team's output, and (2) players contribute more than in the
(unbiased) Nash equilibrium. These predictions are consistent with
the observable behavior of entrepreneurs and salespeople, who often
exhibit overconfidence.

\paragraph{Related Literature and Contribution}

From a theoretical aspect, our methodology of studying a two-stage
auxiliary game where each firm is first endowed with a biased analyst
and then chooses her pricing/advertising level given results of the
analysis is closely related to the literature on delegation (e.g.,
\citealp{fershtman1987equilibrium,fershtman1997unobserved,Dufwenberg_Guth_1999_EvoandDel,fershtman2001strategic})\footnote{See also the related literature on the ``indirect evolutionary approach''
(e.g., \citealp{GuethYaari1992Explaining,heifetz2004evolutionary,DekelElyEtAl2007Evolution,heifetz2007maximize,Herold_Kuzmics_2009,heller2019biased}).}. The delegation literature shows that in price competition, firm
owners would design incentives for managers that encourage the managers
to maximize profits as if the marginal costs are higher than their
true value (see, in particular, \citealp[p. 938]{fershtman1987equilibrium}). 

Our model contributes to this literature but also differs from it
in a few important aspects. First, in our setup the incentives of
all agents are aligned and are based solely on the firm's profit.
A deviation of the firm from profit-maximizing behavior is due to
(non-intentional) naive analytics, rather than due to explicitly distorting
the compensation of the firm's manager. Our novel mechanism is qualitatively
different (as it relies on biased estimations rather than different
incentives), and it induces testable predictions and policy implications
which are different than those induced by delegation (as further discussed
in Remark \ref{rem:An-alternative-interpretation}). Second an important
merit of our model is its generalizability to a wide variety of phenomena
and its applicability to wide class of games. The concept of biased
estimation of sensitivity of demand can be applied in many seemingly-unrelated
setups (e.g., price competition, advertising competition, and team
production), while yielding sharp results about the direction of the
observed biases as well as their magnitudes.

Our research is also related to solution concepts that represent agents
with misconceptions (e.g., conjectural equilibrium (\citealp{battigalli1997conjectural};
\citealp{esponda2013rationalizable}), self-confirming equilibrium
(\citealp{fudenberg1993self}), analogy-based expectation equilibrium
(\citealp{Jehiel2005}), cursed equilibrium (\citealp{eyster2005cursed,antler2019searching}),
coarse reasoning and categorization (\citealp{azrieli2009categorizing,azrieli2010categorization,steiner2015price,heller2016rule}),
Berk-Nash equilibrium (\citealp{esponda2016berk}), rational inattention
(\citealp{steiner2017rational}), causal misconceptions (\citealp{spiegler2017data,spiegler2019behavioral}),
and noisy belief equilibrium (\citealp{friedman2018stochastic}).
These equilibrium notions have been helpful in understanding strategic
behavior in various setups, and yet they pose a conceptual challenge:
why do players not eventually learn to correct their misconceptions?
Much of the literature presenting such models points to cognitive
limitations as the source of this rigidity. Our model and analysis
offer an additional perspective to this issue by suggesting that misperceptions,
such as naive analytics, may yield a strategic advantage and are likely
to emerge in equilibrium. In this sense our approach can be viewed
as providing a tool to explain why some misconceptions persist while
others do not.

\paragraph{Structure}

Section \ref{sec:motivation} presents a motivating example. In Section
\ref{sec:Model-and-Solution} we describe our model and solution concept.
Our main results are presented in Section \ref{sec:General-Results}.
Section \ref{sec:Applications} analyzes three applications: price
competition, advertising competition, and team-production game. 
Section \ref{sec:Implication-for-Merges} demonstrates the implications
of our model on analysis of market structures in oligopolies. The
main text contains proof sketches and formal proofs are relegated
to the appendix. 

\section{\label{sec:motivation}Motivating Example}

Consider two firms $i\in\left\{ 1,2\right\} $ each selling a product
with price $x_{i}\in\mathbb{R}_{+}$. The products are substitute
goods. The demand of firm $i\in\left\{ 1,2\right\} $ at day $t$
is given by:
\[
q_{it}\left(x_{i},x_{-i}\right)=20-x_{i}+0.8x_{-i}+z_{it},\,\,\,\,\,\,\textrm{with }z_{it}\sim\begin{cases}
\epsilon & 0.5\\
-\epsilon & 0.5\text{,}
\end{cases}
\]
where $-i$ denotes the other firm. That is, the expected demand follows
Bertrand competition with differentiated goods, and the daily demand
of each firm has a random i.i.d demand shock, with value either $\epsilon$
or $-\epsilon$ with equal probability. We assume that the marginal
costs are zero, which implies that the profit of each firm is given
by its revenue $\pi_{it}\left(x_{i},q_{it}\right)=x_{i}\cdot q_{it}$. 

The firm managers do not know their demand functions, and they hire
analysts to estimate the sensitivity of demand to price, in order
to find the optimal price. The analyst at each firm asks the firm's
employees to experiment for a couple of weeks with offering a discount
of $\Delta x$ in some of the days, and uses the average change in
demand $\Delta q$ between days with and without the discount to estimate
the elasticity of demand.

Importantly, the firm's employees do not choose the days with discounts
uniformly. The employees observe in each morning a signal that reveals
the demand shock (say, the daily weather), and they implement discounts
on days of low demand, possibly due to the employees having more free
time in these days to deal with posting the discounted price. Setting
discounts in this way creates price endogeneity through a correlation
between prices $x_{i}$ and demand shocks $z_{it}$. 

There are two types of analysts: naive and sophisticated. A naive
analyst does not monitor in which days the employees choose to give
a discount, and implicitly assumes in his econometric analysis that
the environment is the same in the days with discounts as in those
without discounts. In contrast, sophisticated analysts either monitor
the discount decisions to enforce a uniform distribution of discounts,
or correct the correlation between the weather and discounts in their
econometric analysis (e.g., by controlling for the weather).

A sophisticated analyst correctly estimates the mean change in demand
\[
\Delta q_{i}=\left(20-x_{i}+0.8x_{-i}\right)-\left(20-\left(x_{i}-\Delta x\right)+0.8x_{-i}\right)=-\Delta x,
\]
 and thus he accurately estimates the elasticity of demand 
\[
\eta_{i}=-\frac{x_{i}}{q_{i}}\frac{\Delta q}{\Delta x}=-\frac{x_{i}}{q_{i}}\frac{\left(-\Delta x\right)}{\Delta x}=\frac{x_{i}}{q_{i}}.
\]
By contrast, a naive analyst underestimates the mean change in demand
to be:
\[
\Delta q_{i,\textrm{naive}}=\left(20-x_{i}+0.8x_{-i}+\epsilon\right)-\left(20-\left(x_{i}-\Delta x\right)+0.8x_{-i}-\epsilon\right)=-\Delta x+2\epsilon,
\]
 and thus underestimates the elasticity of demand to be
\[
\eta_{i,\textrm{naive}}=\frac{x_{i}}{q_{i}}\frac{\left(\Delta x-2\epsilon\right)}{\Delta x}\equiv\frac{x_{i}}{q_{i}}\alpha_{i},
\]
where we denote $\alpha_{i}=\frac{\left(\Delta x-2\epsilon\right)}{\Delta x}$.

Assume, for example, that $\Delta x$ and $\epsilon$ are such that
$\alpha_{i}=0.6$ (which is the optimal level of naivete as analyzed
in Section \ref{subsec:Price-Competition}). If each firm adjusts
prices according to the estimated elasticity (i.e., slightly increasing
the price if the estimated elasticity is more than 1, and slightly
decreasing the price if less than 1), then the prices converge to
a unique equilibrium in which the estimated elasticity of each firm
is equal to one. Table \ref{tab:Equilibrium-Prices,-Demands} presents
the prices, demands, and profits as a function of the type of analyst
hired by each firm (the calculations are a special case of the analysis
of Section \ref{subsec:Price-Competition}).

\begin{table}[h]
\caption{\label{tab:Equilibrium-Prices,-Demands}Equilibrium prices, demands
and profits as a function of the analysts' types }

\begin{tabular}{|c|c|c|}
\hline 
\multicolumn{3}{|c|}{Prices $\left(x_{1},x_{2}\right)$}\tabularnewline
\hline 
\textcolor{blue}{$\alpha_{1}$} \textbackslash\textcolor{red}{{} $\alpha_{2}$} & \textcolor{red}{1} & \textcolor{red}{0.6}\tabularnewline
\hline 
\textcolor{blue}{1} & \textcolor{blue}{17},~\textcolor{red}{17} & \textcolor{blue}{19},~\textcolor{red}{22}\tabularnewline
\hline 
\textcolor{blue}{0.6} & \textcolor{blue}{22},~\textcolor{red}{19} & \textcolor{blue}{25},\textcolor{red}{25}\tabularnewline
\hline 
\end{tabular}~~~%
\begin{tabular}{|c|c|c|}
\hline 
\multicolumn{3}{|c|}{Demands $\left(q_{1},q_{2}\right)$}\tabularnewline
\hline 
\textcolor{blue}{$\alpha_{1}$} \textbackslash\textcolor{red}{{} $\alpha_{2}$} & \textcolor{red}{1} & \textcolor{red}{0.6}\tabularnewline
\hline 
\textcolor{blue}{1} & \textcolor{blue}{17},~\textcolor{red}{17} & \textcolor{blue}{19},~\textcolor{red}{13}\tabularnewline
\hline 
\textcolor{blue}{0.6} & \textcolor{blue}{13},~\textcolor{red}{19} & \textcolor{blue}{15},\textcolor{red}{15}\tabularnewline
\hline 
\end{tabular}~~~%
\begin{tabular}{|c|c|c|}
\hline 
\multicolumn{3}{|c|}{Profits $\left(\pi_{1},\pi_{2}\right)$}\tabularnewline
\hline 
\textcolor{blue}{$\alpha_{1}$} \textbackslash\textcolor{red}{{} $\alpha_{2}$} & \textcolor{red}{1} & \textcolor{red}{0.6}\tabularnewline
\hline 
\textcolor{blue}{1} & \textcolor{blue}{277},~\textcolor{red}{277} & \textcolor{blue}{351},~\textcolor{red}{287}\tabularnewline
\hline 
\textcolor{blue}{0.6} & \textcolor{blue}{287},~\textcolor{red}{351} & \textcolor{blue}{375},~\textcolor{red}{375}\tabularnewline
\hline 
\end{tabular}
\end{table}

Observe that each firm's profit increases when the firm hires a naive
analyst, and decreases when it hires a sophisticated analyst, regardless
of the type of analyst hired by the competing firm. The intuition
is that a naive analyst induces a firm to underestimate the elasticity
of demand, and as a result, to raise prices. This has a beneficial
indirect strategic effect of inducing the competitor to increase prices
as well. It turns out that the positive indirect effect outweighs
the negative direct effect. Thus, if firms occasionally replace their
analysts based on their annual profits (i.e., they are more likely
to fire an analyst the lower the profit is), then the firms are likely
to converge to both hiring naive analysts. This would induce both
firms to choose higher prices and have higher profits relative to
the Nash equilibrium prices arising with sophisticated analysts.

Our formal results show that these insights hold in a general model.
Specifically, we show that in a large class of strategic interactions
(incorporating both price competition and advertising competition,
as well as both strategic complements and strategic substitutes) players
(i.e., firms) choose to hire naive analysts. These naive analysts
underestimate elasticity of demand in price competition, while they
overestimate the effectiveness of advertising in advertising competition.
Finally, we show that the equilibrium induced by naive analysts Pareto
dominates the Nash equilibrium with sophisticated analysts in a game
with strategic complements, while it is Pareto dominated by the Nash
equilibrium in a game with strategic substitutes.

\section{Model and Solution Concept\label{sec:Model-and-Solution}}

We introduce an analytics game in which competing firms hire analysts
to estimate the sensitivity of demand, which is then used to determine
the strategic choices of the firm. Next we present our solution concept
of a naive analytics equilibrium.

\subsection{Underlying Game\label{subsec:Underlying-Game}}

A\emph{ naive-analytics game} $\Gamma=\left(G,A,f\right)$ is a two
stage game in which each of $N=\left\{ 1,2,...,n\right\} $ players
(firms) hires an analyst who estimates the sensitivity of demand in
the first stage (with a bias from a set of feasible biases $A$, which
corresponds to biased estimation functions in $f$, as described in
Section \ref{subsec:RIDE-and--Equilibrium}) and then makes a strategic
choice that affects demand in the second stage. We first describe
the structure of the second stage, which we call the \textit{underlying
game} and denote by $G=\left(N,X,q,\pi\right)$. In the underlying
game each firm $i\in N$ makes a strategic choice $x_{i}\in X_{i}$
that affects the demands and the profits of all firms, where $X_{i}\subseteq\mathbb{R}$
is a (possibly unbounded) interval of feasible choices of firm $i$.
The interpretation of $x_{i}$ depends on the application, e.g.,:
(1) it is equal to the price set by the firm in the motivating example,
and (2) it is equal to the advertising budget in the application of
Section \ref{subsec:Advertising-Competition}. 

We define the following notation: $X=\prod_{i\in N}X_{i}$ is the
Cartesian product of these intervals. $Int\left(X_{i}\right)$ (resp.,
$Int\left(X\right))$ denotes the interior of $X_{i}$ (Resp., $X$).
$-i\equiv N\backslash\left\{ i\right\} $ denotes all firms except
firm $i$ and $-ij\equiv N\backslash\left\{ i,j\right\} $ denotes
all firms except $i$ and $j$. In two-player games $-i$ denotes
the opponent of firm $i$. $\left(x'_{i},x_{-i}\right)$ denotes the
strategy profile in which player $i$ plays strategy $x'_{i}$, while
all other players play according to the profile $x_{-i}$ (and we
apply a similar notation for $x_{-ij}$). Finally, $q_{i}\left(x\right)$
denotes the demand of firm $i$. 

The (true) payoff, or profit, of each firm $i\in N$, denoted by $\pi_{i}\left(x_{i},q_{i}\left(x\right)\right)$,
depends on the firm's demand $q_{i}\left(x\right)$ and its strategic
choice $x_{i}$. We assume that the demand functions $q_{i}\left(x\right)$
and payoff functions $\pi_{i}\left(x_{i},q_{i}\right)$ of all firms
are continuously twice differentiable in all parameters in $Int\left(X\right)$. 

A necessary condition for a strategy $x_{i}\in Int\left(X_{i}\right)$
to be a best reply to the opponents' strategy profile is that it satisfies
the following first-order condition:

\begin{equation}
\frac{d\pi_{i}}{dx_{i}}=\underbrace{\frac{\partial\pi_{i}}{\partial x_{i}}}_{(i)}+\underbrace{\frac{\partial\pi_{i}}{\partial q_{i}}}_{(ii)}\cdot\underbrace{\frac{\partial q_{i}}{\partial x_{i}}}_{(iii)}=0.\label{eq:foc}
\end{equation}
Strategy $x_{i}\in Int\left(X_{i}\right)$ is a \emph{local best reply}
to an opponents' profile $x_{-i}$, if there is an open interval around
$x_{i}$ such that $\pi_{i}\left(x_{i},x_{-i}\right)>\pi_{i}\left(x'_{i},x_{-i}\right)$
for each $x'_{i}$ in this interval. A sufficient condition for a
strategy satisfying the FOC \eqref{eq:foc} to be a local best reply
is the second-order condition $\frac{d^{2}\pi_{i}\left(x\right)}{dx_{i}}<0$.

\subsection{Choice of Analysts and $\alpha$-Equilibrium\label{subsec:RIDE-and--Equilibrium}}

In order to maximize their profits when choosing $x_{i}$, firms need
to know or estimate the impact of their actions on their profits.
Corresponding to the numbering of terms in \eqref{eq:foc}, we assume
that each firm knows (or correctly estimates): (i) the direct effect
of its strategy on its profit $\frac{\partial\pi_{i}}{\partial x_{i}}$;
and (ii) the effect of the firm's demand on its profit $\frac{\partial\pi_{i}}{\partial q_{i}}$.
By contrast, we assume it is difficult for the firm to estimate (iii)
the response of its demand to marginal changes in its strategy, i.e.,
to estimate $\frac{\partial q_{i}}{\partial x_{i}}$. For example,
during price competition firms know how their product's prices affect
their profit margins and how demand affects profit, but might not
know how sensitive consumers might be to price changes. Similarly,
in advertising competition (Section \ref{subsec:Advertising-Competition})
firms know how increasing advertising spending affects their bottom
line costs, but might not know the impact of their advertising on
demand. Each firm therefore hires an analyst in the first stage who
is tasked with estimating $\frac{\partial q_{i}}{\partial x_{i}}$.

Let $A\subseteq\mathbb{R}_{++}$ denote the interval of feasible biases
of analysts, and we assume that $A$ includes an open interval around
$1$. Each analyst is characterized by a bias $\alpha_{i}\in A\subseteq\mathbb{R}_{++}$
that causes him to estimate the marginal effect of the strategy $x_{i}$
on demand $q_{i}$ as $f\left(\frac{\partial q_{i}}{\partial x_{i}},\alpha_{i}\right)$
instead of $\frac{\partial q_{i}}{\partial x_{i}}$, where the bias
function $f$ is continuously differentiable in both parameters and
is strictly increasing in $\frac{\partial q_{i}}{\partial x_{i}}$.
We are interested in situations in which it is clear what the sign
of $\frac{\partial q_{i}}{\partial x_{i}}$ is. Thus, we assume that
$f\left(\frac{\partial q_{i}}{\partial x_{i}},\alpha_{i}\right)$
has the same sign as $\frac{\partial q_{i}}{\partial x_{i}}$, and
that $\left|f\left(\frac{\partial q_{i}}{\partial x_{i}},\alpha_{i}\right)\right|$
is increasing in $\alpha_{i}$.\footnote{The definitions and results can be adapted to a setup in which the
analyst's bias is the opposite sign of the sensitivity of demand,
i.e., a negative $\alpha_{i}$ for which the sign of $f\left(\frac{\partial q_{i}}{\partial x_{i}},\alpha_{i}\right)$
is the opposite of the sign of $\frac{\partial q_{i}}{\partial x_{i}}$.
The outcome would be that in any naive analytics equilibrium firms
hire analysts who are biased only with respect to the magnitude of
$\frac{\partial q_{i}}{\partial x_{i}}$ but not with respect to its
sign (i.e., firms hire only analysts with positive $\alpha_{i}$-s).} We normalize $\alpha_{i}=1$ to represent a calibrated (non-biased)
analyst, i.e., $f\left(\frac{\partial q_{i}}{\partial x_{i}},1\right)=\frac{\partial q_{i}}{\partial x_{i}}$.
In the applications presented later in the paper, we will assume that
the bias is multiplicative: $f\left(\frac{\partial q_{i}}{\partial x_{i}},\alpha_{i}\right)=\alpha_{i}\cdot\frac{\partial q_{i}}{\partial x_{i}}$.
We denote this multiplicative bias by $f_{m}$.

Let $\frac{\partial q_{i}^{\alpha_{i}}}{\partial x_{i}}\equiv f\left(\frac{\partial q_{i}}{\partial x_{i}},\alpha_{i}\right)$
denote the biased estimation of analyst with bias $\alpha_{i}$. We
denote the bias profile of all analysts by $\alpha=(\alpha_{1},\ldots,\alpha_{n})$.
An $\alpha_{i}$-analyst induces the firm to choose $x_{i}$ that
solves the biased first-order and second-order conditions
\begin{equation}
\frac{d\pi_{i}^{\alpha_{i}}\left(x\right)}{dx_{i}}\equiv\frac{\partial\pi_{i}}{\partial x_{i}}+\frac{\partial\pi_{i}}{\partial q_{i}}\cdot\frac{\partial q_{i}^{\alpha_{i}}}{\partial x_{i}}=0,\,\,\,\,\,\,\,\,\,\frac{d^{2}\pi_{i}^{\alpha_{i}}\left(x\right)}{dx_{i}}\equiv\frac{d}{dx_{i}}\left(\frac{d\pi_{i}^{\alpha_{i}}\left(x\right)}{dx_{i}^{2}}\right)<0,\label{eq:biased-FOC-SOC}
\end{equation}
instead of the unbiased conditions in \eqref{eq:foc}.

There are many reasons why analysts might be biased. One example is
inadvertently creating endogenous correlation between the firm's strategy
and demand without taking this correlation into account in the analysis.
If a firm sets lower ``sale'' prices on days of low demand and
higher regular prices on days of high demand, estimating price elasticities
using the resulting data will show that consumers are less price sensitive
than they actually are (as in Section \ref{sec:motivation}). Another
example is when firms set their advertising budgets differently in
specific times such as before holidays, or weekends. This would create
correlation in the levels of advertising with those of competitors.
Ignoring this correlation during analysis may lead to a biased estimate
of advertising effectiveness. We present micro-foundations for biased
analytics towards the end of Section \ref{subsec:Price-Competition}
(price competition) and Section \ref{subsec:Advertising-Competition}
(advertising).

Next we define how the firms play in the second-stage, given the analysts'
biased profile $\alpha$. A strategy profile is an $\alpha$-equilibrium
of the game $\Gamma$ if each player believes (based on her analyst's
estimated sensitivity of demand) that any small change to its strategy
would decrease her perceived payoff. Formally
\begin{defn}[\textbf{$\mathrm{\alpha}$}-Equilibrium]
\label{def:alpha-eq}Fix a biasedness profile $\alpha\in A^{n}$.
A strategy profile $x$ is an $\alpha$-\emph{equilibrium} if (1)
$\frac{d\pi_{i}^{\alpha}\left(x\right)}{dx_{i}}=0$, and (2) $\frac{d^{2}\pi_{i}^{\alpha}\left(x\right)}{dx_{i}^{2}}<0$. 
\end{defn}

\subsection{Naive Analytics Equilibrium (NAE)\label{subsec:Naive-Analytics-Equilibrium-(NAE}}

In what follows we define our main solution concept. A naive analytics
equilibrium (NAE) is a bias profile and a strategy profile, where
(1) the strategy profile is an $\alpha$-equilibrium, and (2) each
bias is a best-reply to the opponents' bias profile (i.e., any unilateral
deviation to another bias would induce a new $\alpha$-equilibrium
with a weakly lower payoff to the deviator). 
\begin{defn}
\label{def:sloppy-analytics-eq}A \emph{naive analytics equilibrium
}of $\Gamma=\left(G,A\right)$ is a pair $\left(\alpha^{*},x^{*}\right)$,
where:
\begin{enumerate}
\item $x^{*}\in X$ is an $\alpha^{*}$-equilibrium of the underlying game
$G$.
\item $\pi_{i}\left(x^{*}\right)\geq\pi_{i}\left(x'\right)$ for each $i\in N$,
each $\alpha'_{i}\in A$, and each $\left(\alpha'_{i},\alpha_{-i}^{*}\right)$-equilibrium
$x'$.
\end{enumerate}
\end{defn}
We do not interpret the equilibrium behavior in the first-stage as
the result of an explicit maximization of sophisticated firms who
know the demand function and calculate the optimal $\alpha$-s for
their analysts. Conversely, we think of the firms as having substantial
uncertainty about the demand function and its dependence on the behavior
of various competitors. Due to this uncertainty, the firms hire analysts
to estimate the sensitivity of demand. Occasionally (say, at the end
of each year) firms consider replacing the current analyst with a
new one (say, with a new random value of $\alpha_{i}$), and a firm
is more likely to do so the lower its profit is. If after several
months it turns out that the new analyst decreases the firm's profit,
then the firm is likely to re-hire the old analyst. Gradually, such
a process would induce the firms to converge to hiring most of the
time analysts with values of $\alpha$ that are best replies to each
other, and thus constitute a naive analytics equilibrium $\left(\alpha,x\right)$.
Various existing static solution concepts are motivated by assuming
an analogous dynamic convergence process, see, e.g., \citet{huck1999indirect,DekelElyEtAl2007Evolution,winter2017mental,frenkel2018endowment}.

It is important to note that the observed data does not contradict
the optimality of the strategic choices of the firms or the correctness
of the estimations of the sensitivity of demand of the analysts. Consider,
for example, a naive analytics equilibrium $\left(\alpha,x\right)$
in the price competition described in Section \ref{sec:motivation}.
A firm that wishes to confirm that its price is indeed optimal (i.e.,
that it maximizes its profit given the demand) is likely to experiment
with temporary changes in prices to see its influence on demand. Under
the arguably plausible assumption that the analysis of such an experiment
will be done with the same level of sophistication as the one leading
to $\left(\alpha,x\right)$, the firm's conclusion from the experiment
would be that the sensitivity of demand is exactly as estimated by
the firm's naive analyst, and that the firm's equilibrium strategy
is optimal (e.g., it induces elasticity of $-1$ in the motivating
example, and thus maximizes the firm's profit). 

A firm in a naive analytics equilibrium that will deviate in the direction
that increases profit, will induce a short-run increase in its expected
profit; yet this temporary increase in profits might be difficult
to measure in noisy environments. Moreover, as soon as the competitors
observe the change in their demands (due to the firm's deviation),
they are likely to adjust their strategies toward their new (perceived)
best replies. Following this adjustment, the deviating firm's profit
will decrease (which is formalized in Proposition \ref{pro:stackelberg-leader}
below).
\begin{rem}[Delegation interpretation]
\label{rem:An-alternative-interpretation}An alternative interpretation
of our model (which we do not use in the paper) is of delegation.
Let $\pi_{i}^{\alpha_{i}}:X\rightarrow\mathbb{R}$ be a subjective
payoff function such that maximizing $\pi_{i}^{\alpha_{i}}$ with
an unbiased estimation is equivalent to maximizing $\pi_{i}$ with
a biased estimation of $\alpha_{i}$, i.e., for any strategy profile
$x\in X$ 
\[
x\textrm{ is an \ensuremath{\alpha}-equilibrium }\,\,\textrm{iff}\,\,\pi_{i}^{\alpha_{i}}\left(x\right)=\textrm{argmax}_{x'_{i}\in X_{i}}\pi_{i}^{\alpha_{i}}\left(x'_{i},x_{-i}\right)\,\,\forall i\in N.
\]
 Let $\Pi_{i}^{A}=\left\{ \pi_{i}^{\alpha_{i}}|\alpha_{i}\in A\right\} $
be the set of all such subjective payoff function. One can reinterpret
a NAE as an equilibrium of a delegation game \citep{fershtman1987equilibrium}
where in the first stage each firm's owner simultaneously chooses
a payment scheme for its manager, which induces the manager with a
subjective payoff function in $\Pi_{i}^{A}$. In the second stage
the managers play a Nash equilibrium of the game induced by the subjective
payoffs. Although, both interpretations (namely, naive analytics and
delegation) yield the same prediction about the equilibrium strategy
profile, they differ in other testable predictions, as well as with
their insights and policy implications (as discussed in Section \ref{sec:Conclusion}).
For concreteness, we focus on the comparison for price competition
(as in the motivating example). The delegation interpretation predicts
firms will correctly estimate the elasticity of demand and pay managers
a payoff that increases in the firm profit and decreases in the firm's
sales (see, \citealp[p. 938]{fershtman1987equilibrium}). The naive
analytics interpretation predicts that firms will hire naive analysts
that overestimate elasticity of demand, with a manager's payoff scheme
that depends directly on the firm's profit. As mentioned elsewhere
in the paper, we believe this latter prediction is more consistent
with the empirically observed behavior of firms.
\end{rem}

\subsection{Robustness to Unilateral Sophistication\label{subsec:Robustness-to-Unilateral}}

In this section we show that a naive analytics equilibrium is robust
to any of the players unilaterally becoming patient and sophisticated,
in the sense that each player plays the same strategy she would have
played if she were an unbiased, sophisticated, and patient player,
who plays the strategy that will maximize her payoff, given the biased
responses of her opponents (i.e., her Stackelberg-leader strategy).
Such a sophisticated player might represent a scenario, for example,
where an E-commerce retailer sells goods to consumers, but also allows
other third-party sellers to sell competing goods through the retailer's
website. 

We begin by defining an $\left(x_{i},\alpha_{-i}\right)$-equilibrium
as a profile in which player $i$ plays $x_{i}$, all players except
for $i$ play a perceived best-reply given their biases, and $x_{i}$
is a perceived best reply to some feasible bias in $A$. Formally 
\begin{defn}
Fix player $i\in N$, strategy $x_{i}\in X_{i}$, and bias profile
$\alpha_{-i}\in A^{n-1}$. Strategy profile $\left(x_{i},x_{-i}\right)\in X$
is an \emph{$\left(x_{i},\alpha_{-i}\right)$-equilibrium} if (1)
$\frac{d\pi_{j}^{\alpha_{j}}\left(x\right)}{dx_{j}}=0$ and $\frac{d^{2}\pi_{j}^{\alpha_{j}}\left(x\right)}{dx_{j}^{2}}<0$
for each player $j\neq i$, and (2) there exists $\alpha_{i}\in A$
such that $\frac{d\pi_{i}^{\alpha_{i}}\left(x\right)}{dx_{i}}=0$
and $\frac{d^{2}\pi_{i}^{\alpha_{i}}\left(x\right)}{dx_{i}^{2}}<0$.
\end{defn}
Next we define $X_{i}^{SL}\left(\alpha_{-i}\right)$ as the set of
optimal strategies of an (unbiased) Stackelberg-leader player $i$
who faces opponents with bias profile $\alpha_{-i}$.
\begin{defn}
\label{def:SL-strategy}We say that $x_{i}^{*}\in X_{i}$ is a \emph{Stackelberg-leader
strategy} against bias profile $\alpha_{-i}\in A^{n-1}$ if there
exists an $\left(x_{i}^{*},\alpha_{-i}\right)$-equilibrium $x^{*}$,
such that $\pi_{i}\left(x^{*}\right)\geq\pi_{i}\left(x'\right)$ for
any strategy $x'_{i}\in X_{i}$ and any $\left(x'_{i},\alpha_{-i}\right)$-equilibrium
$x'$. Let $X_{i}^{SL}\left(\alpha_{-i}\right)\subseteq X_{i}$ be
the set of all such Stackelberg-leader strategies.
\end{defn}
Next we characterize a naive analytics equilibrium as an $\alpha$-equilibrium
in which everyone plays Stackelberg-leader strategies. 
\begin{prop}
\label{pro:stackelberg-leader}The pair $\left(\alpha^{*},x^{*}\right)$
is a naive analytics equilibrium iff (1) $x^{*}$ is an $\alpha^{*}$-equilibrium,
and (2) $x_{i}^{*}\in X_{i}^{SL}\left(\alpha_{-i}^{*}\right)$ for
each player $i\in N$.
\end{prop}
\begin{proof}[Sketch of proof; proof in Appx. \ref{subsec:Proof-of-Stackelberg}]
 If $x_{i}^{*}$ is (resp., is not) a Stackelberg-leader strategy
against bias profile $\alpha_{-i}^{*}$, then there does not (resp.,
does) exist a bias $\alpha'_{i}$ that induces an $\left(\alpha'_{i},\alpha_{-i}\right)$-equilibrium
where player $i$ plays a Stackelberg-leader strategy $x'_{i}\in X_{i}^{SL}\left(\alpha_{-i}\right)$
and gains a payoff higher than in $x^{*}$. 
\end{proof}

\section{General Results\label{sec:General-Results}}

We answer 3 main questions about firms in a naive analytics equilibrium:
(1) when are the players' perceived best replies above or below their
unbiased best replies (Section \ref{subsec:Perceived-Best-Replies}),
(2) when do they underestimate or overestimate the sensitivity of
demand through biased analytics (Section \ref{subsec:Mis-estimating-the-Sensitivity}),
and (3) when do they achieve payoffs that Pareto dominate the Nash
equilibrium (Section \ref{subsec:Pareto-Dominance}). A summary of
results is presented in Table \ref{tab:Summary-of-Results} in Section
\ref{subsec:Summary-of-Results}.

\subsection{\label{subsec:Monotone-Derivatives}Assumptions}

Our first pair of assumptions require that (1) the externalities of
any player's strategy on any opponent's payoff are all in the same
direction, and (2) the partial derivative of increasing a player's
strategy on her payoff is in the same direction for all players, i.e.,
(i) all players have the same direction for the direct effect \textbf{$\frac{\partial\pi_{i}}{\partial x_{i}}$},
and (ii) all players have the same direction for the indirect\textbf{
}effect $\frac{\partial\pi_{i}}{\partial q_{i}}\cdot\frac{\partial q_{i}^{\alpha_{i}}}{\partial x_{i}}$.
\begin{assumption}[\textbf{Monotone externalities}]
\textbf{\label{assu:Monotone-externalities}} The payoff externalities
$\frac{d\pi_{i}\left(x\right)}{dx_{j}}$ are either all positive or
all negative for every $i\neq j\in N$ and every $x\in X$. 
\end{assumption}
\begin{assumption}[\textbf{Monotone partial derivatives}]
\textbf{\label{assu:Monotone-partial-deriv}}\renewcommand{\labelenumi}{\roman{enumi})}
\begin{enumerate}
\item \textbf{$\frac{\partial\pi_{i}}{\partial x_{i}}$ }is either all positive
or all negative for every $i\in N$ and every $x\in X$. 
\item \textbf{$\frac{\partial\pi_{i}}{\partial q_{i}}\cdot\frac{\partial q_{i}}{\partial x_{i}}$}
is either all positive or all negative for every $i\in N$ and every
$x\in X$. 
\end{enumerate}
\end{assumption}
Assumptions \ref{assu:Strong-concavity}--\ref{assu:Strong-com-subs}
require the payoff function to be concave and satisfy either strategic
complements or strategic substitutes. Moreover, we require that these
properties hold with respect to any possible biased estimation of
demand sensitivity. Thus, we call these properties robust concavity
and robust strategic complementarity. 
\begin{assumption}[\textbf{Robust concavity}]
\textbf{\label{assu:Strong-concavity}} $\frac{d^{2}\pi_{i}^{\alpha_{i}}\left(x\right)}{dx_{i}^{2}}<0$
$\forall i\in N$, $x\in X$, and $\alpha_{i}\in A$. 
\end{assumption}
Robust concavity implies that each opponent's profile admits a unique
perceived best reply, which we denote by $BR_{i}^{\alpha_{i}}\left(x_{-i}\right)$,
and we omit the superscript when denoting with the unbiased best reply,
i.e., $BR_{i}\left(x_{-i}\right)\equiv BR_{i}^{1}\left(x_{-i}\right)$.
\begin{assumption}[\textbf{Robust strategic complementarity}]
\textbf{\label{assu:Strong-com-subs}} The cross derivatives $\frac{d^{2}\pi_{i}^{\alpha_{i}}\left(x\right)}{dx_{j}dx_{i}}$
are either all positive or all negative for every $i\neq j\in N$,
$x\in X$, and $\alpha_{i}\in A$. 
\end{assumption}
Assumptions \ref{assu:Strong-concavity}--\ref{assu:Strong-com-subs}
can be replaced with stronger assumptions that depend only on the
underlying game $G$ (and not on elements related to the biases: $A$
and $f$).

\renewcommand\theassumption{\arabic{assumption}'}
\addtocounter{assumption}{-2}
\begin{assumption}
\label{assu:Strong-concavity-prime}$\frac{d}{dx_{i}}\left(\frac{\partial\pi_{i}\left(x\right)}{\partial x_{i}}\right),\,\frac{d}{dx_{i}}\left(\frac{\partial\pi_{i}}{\partial q_{i}}\frac{\partial q_{i}\left(x\right)}{\partial x_{i}}\right)\leq0$
for every $i\in N$ and $x\in X$, where at least one of the two inequalities
is strict.
\end{assumption}
\begin{assumption}
\label{assu:Strong-com-subs-prime}Either $\frac{d}{dx_{j}}\left(\frac{\partial\pi_{i}\left(x\right)}{\partial x_{i}}\right),\,\frac{d}{dx_{j}}\left(\frac{\partial\pi_{i}}{\partial q_{i}}\frac{\partial q_{i}\left(x\right)}{\partial x_{i}}\right)\geq0$
or $\frac{d}{dx_{j}}\left(\frac{\partial\pi_{i}\left(x\right)}{\partial x_{i}}\right),\frac{d}{dx_{j}}\left(\frac{\partial\pi_{i}}{\partial q_{i}}\frac{\partial q_{i}\left(x\right)}{\partial x_{i}}\right)\leq0$,
for every $i\neq j\in N$ and $x\in X$, where at least one of the
two inequalities is strict.
\end{assumption}
It is immediate that Assumption $3'$ (resp., 4') implies Assumption
\ref{assu:Strong-concavity} (resp., \ref{assu:Strong-com-subs}).

\renewcommand\theassumption{\arabic{assumption}}

Next we assume that perceived best replies are bounded. Formally (where
we write that $x\leq x'\in X$ iff $x_{j}\leq x_{j}'$ for each $j\in N$)
\begin{assumption}[\textbf{Bounded Perceived Best Replies}]
\label{assu:bounded-percieved-BR}For each \textup{\emph{bias profile}}
$\alpha\in A^{n}$ there exist profiles $\underline{x}^{\alpha}\leq\overline{x}^{\alpha}\in X$
such that $BR_{i}^{\alpha_{i}}\left(\underline{x}_{-i}^{\alpha}\right),BR_{i}^{\alpha_{i}}\left(\overline{x}_{-i}^{\alpha}\right)\in\left(\underline{x}_{\alpha},\overline{x}_{\alpha}\right)$
for each $i\in N$.
\end{assumption}
Assumption \ref{assu:bounded-percieved-BR} implies that the signs
of the two partial derivatives must differ, i.e.,

\begin{equation}
\frac{\partial\pi_{i}}{\partial x_{i}}\cdot\left(\frac{\partial\pi_{i}}{\partial q_{i}}\cdot\frac{\partial q_{i}}{\partial x_{i}}\right)<0\label{eq:opposite-signs}
\end{equation}
because otherwise there could not be an interior best-reply $BR_{i}\left(\underline{x}_{-i}\right)$,
for which $\frac{d\pi_{i}}{dx_{i}}=\frac{\partial\pi_{i}}{\partial x_{i}}+\frac{\partial\pi_{i}}{\partial q_{i}}\cdot\frac{\partial q_{i}}{\partial x_{i}}=0$.
Further observe that Assumption \ref{assu:bounded-percieved-BR} implies
the existence of $\alpha$-equilibrium for each bias profile $\alpha$.
\begin{claim}
If $\Gamma$ satisfies Assumptions \ref{assu:Strong-concavity} and
\ref{assu:bounded-percieved-BR}, then there exists an $\alpha$-equilibrium
for any $\alpha\in A$.
\end{claim}
\begin{proof}
Let $X_{\alpha}=\left\{ x\in X|\underline{x}_{\alpha}\leq x\leq\overline{x}_{\alpha}\right\} $
be the set of profiles bounded between $\underline{x}_{\alpha}$ and
$\overline{x}_{\alpha}$ (as defined in Assumption \ref{assu:bounded-percieved-BR}).
The claim is implied by applying Brouwer fixed-point theorem on the
convex and compact set $X_{\alpha}$ and the continuous function $BR^{\alpha}:X_{\alpha}\rightarrow X_{\alpha}$.
\end{proof}
Our final assumption is required only in games with strategic substitutes
with more than two players. Consider a situation in which some player
$i$ deviates from an $\alpha$-equilibrium and then never further
changes her play. Player $i$'s deviation induces the remaining players
to adapt their strategies to a perceived best reply. This adaptation,
in turn, induces these players to further adapt their strategies to
a new perceived best reply. Assumption \ref{assu:consistent-second-adaptation}
requires that the players' strategies keep the same direction of deviation
with respect to the original $\alpha$-equilibrium after the secondary
adaptation. 
\begin{assumption}[\textbf{Consistent secondary adaptation}]
\label{assu:consistent-second-adaptation} Fix any bias profile $\alpha\in A^{N}$,
any $\alpha$-equilibrium $x$, any player $i\in N$, and any strategy
$\hat{x}_{i}\in X_{i}$. Then for each player $j\neq i$
\[
BR_{j}^{\alpha_{j}}\left(\hat{x}_{i},x_{-i}\right)>x_{j}\,\,\Leftrightarrow\,\,BR_{j}^{\alpha_{j}}\left(\hat{x}_{i}\text{,}\left(BR_{k}^{\alpha_{k}}\left(\hat{x}_{i},x_{-i}\right)\right)_{k\neq i}\right)>x_{j}.
\]
\end{assumption}
It is immediate that Assumption \ref{assu:consistent-second-adaptation}
is vacuous (i.e., it always holds) if either (1) the game has two
players (in which case $BR^{\alpha_{-i},\hat{x}_{i}}\left(BR^{\alpha_{-i},\hat{x}_{i}}\left(x\right)\right)\equiv BR^{\alpha_{-i},\hat{x}_{i}}\left(x\right)$),
or (2) the game has strategic complements (which implies that the
original adaptation and the secondary adaptation are in the same direction).
Thus, Assumption \ref{assu:consistent-second-adaptation} is non-trivial
only in the remaining case of a game with $n\geq3$ players and strategic
substitutes. 

As can be seen later in the paper Assumptions \ref{assu:Monotone-externalities}--\ref{assu:consistent-second-adaptation}
are satisfied in various economic applications, such as price competition
(Section \ref{subsec:Price-Competition}), advertising competition
(Section \ref{subsec:Advertising-Competition}), team production (Section
\ref{sec:Team-Production}), and Cournot competition.

\subsection{Perceived Best Replies\label{subsec:Perceived-Best-Replies}}

Our first result characterizes whether the perceived best-replies
of the agents are above or below the unbiased best reply in naive
analytics equilibria. Specifically, it shows that the perceived best-reply
of each player is above her unbiased best reply (i.e., $x_{i}^{*}>BR_{i}\left(x_{-i}^{*}\right)$)
in any naive analytics equilibrium iff the strategic complementarity
has the same direction as the payoff externalities (i.e., $\frac{d\pi_{i}^{2}}{dx_{i}dx_{j}}\cdot\frac{d\pi_{i}}{dx_{j}}>0$).
Formally (where we write $\frac{d\pi_{i}^{2}}{dx_{i}dx_{j}}\cdot\frac{d\pi_{i}}{dx_{j}}>0$
without specifying the quantifiers on the profile $x$ and players
$i\neq j$ due to Assumption \ref{assu:Monotone-externalities} (resp.,
\ref{assu:Strong-com-subs}) that imply that the sign of $\frac{d\pi_{i}^{2}}{dx_{i}dx_{j}}$
(resp., $\frac{d\pi_{i}}{dx_{j}}$) is the same for all players and
profiles).
\begin{prop}
\label{prop:percieved-BR-vs-unbiased-BR}If $\Gamma$ satisfies Assumptions
\ref{assu:Monotone-externalities}--\ref{assu:consistent-second-adaptation}
and $\left(\alpha^{*},x^{*}\right)$ is a naive analytics equilibrium:
\end{prop}
\begin{enumerate}
\item $x_{i}^{*}>BR_{i}\left(x_{-i}^{*}\right)$ for any $i\in N$ iff $\frac{d\pi_{i}^{2}}{dx_{i}dx_{j}}\cdot\frac{d\pi_{i}}{dx_{j}}>0$.
\item $x_{i}^{*}<BR_{i}\left(x_{-i}^{*}\right)$ for any $i\in N$ iff $\frac{d\pi_{i}^{2}}{dx_{i}dx_{j}}\cdot\frac{d\pi_{i}}{dx_{j}}<0$.
\end{enumerate}
\begin{proof}
Assume that $x_{i}^{*}>BR_{i}\left(x_{-i}^{*}\right)$. If opponents
did not react to player $i$ unilaterally decreasing her strategy
towards $BR_{i}\left(x_{-i}^{*}\right)$, then player $i$ would have
gained from it, but the fact that $x_{i}^{*}$ is a Stackelberg-leader
strategy means that decreasing player $i$'s strategy has a negative
impact on player $i$'s payoff due to the opponents' reaction (``$\underset{(1)}{\Leftrightarrow}$''
below). This can only happen if either the game has robust strategic
complements and positive externalities, resulting in lower prices
and payoffs if player $i$ decreases $x_{i}^{*}$, or similarly if
it has robust strategic substitutes and negative externalities (``$\underset{(2)}{\Leftrightarrow}$''
below), both of which imply that the strategic complementarity and
the externalities have the same direction. The equilibrium response
remains in this initial direction due to Assumption \ref{assu:consistent-second-adaptation}.
The argument that $x_{i}^{*}<BR_{i}\left(x_{-i}^{*}\right)$ implies
that the strategic complementarity and the externalities are in opposite
directions is analogous. 

\[
\begin{array}{c}
x_{i}^{*}>BR_{i}\left(x_{-i}^{*}\right)\\
x_{i}^{*}<BR_{i}\left(x_{-i}^{*}\right)
\end{array}\,\underset{(1)}{\Leftrightarrow}\,\begin{array}{c}
\textrm{Decreasing \ensuremath{x_{i}}}\textrm{ induces a harmful opponents' reponse}\\
\textrm{Increasing \ensuremath{x_{i}}}\textrm{ induces a harmful opponents' reponse}
\end{array}\,\underset{(2)}{\Leftrightarrow}\,
\]

\[
\begin{array}{c}
\left(\frac{d^{2}\pi_{i}}{dx_{i}dx_{j}}>0\,\&\,\left(\frac{d\pi_{i}}{dx_{j}}>0\right)\right)\textrm{OR}\left(\frac{d^{2}\pi_{i}}{dx_{i}dx_{j}}<0\,\&\,\left(\frac{d\pi_{i}}{dx_{j}}<0\right)\right)\\
\left(\frac{d^{2}\pi_{i}}{dx_{i}dx_{j}}>0\,\&\,\left(\frac{d\pi_{i}}{dx_{j}}<0\right)\right)\textrm{OR}\left(\frac{d^{2}\pi_{i}}{dx_{i}dx_{j}}<0\,\&\,\left(\frac{d\pi_{i}}{dx_{j}}>0\right)\right)
\end{array}\,\underset{(4)}{\Leftrightarrow}\,\begin{array}{c}
\frac{d\pi_{i}^{2}}{dx_{i}dx_{j}}\cdot\frac{d\pi_{i}}{dx_{j}}>0\\
\frac{d\pi_{i}^{2}}{dx_{i}dx_{j}}\cdot\frac{d\pi_{i}}{dx_{j}}<0\text{.}
\end{array}\qedhere
\]

Appendix \ref{subsec:Example-why-Assumption-6-is-nessecary} presents
an example that illustrates why Assumption \ref{assu:consistent-second-adaptation}
(consistent\textbf{ }secondary adaptation) is necessary for Proposition
\ref{prop:percieved-BR-vs-unbiased-BR} (and for the remaining results
in this section). An immediate corollary of Proposition \ref{prop:percieved-BR-vs-unbiased-BR}
is that the difference between perceived best reply and the unbiased
best reply of all players is always in the direction that benefits
(harms) the opponents in games with strategic complements (substitutes).
Formally,
\end{proof}
\begin{cor}
\label{cor:sign-strategic-compl-supp}If $\Gamma$ satisfies Assumptions
\ref{assu:Monotone-externalities}--\ref{assu:consistent-second-adaptation}
and $\left(\alpha^{*},x^{*}\right)$ is a naive analytics equilibrium,
then:
\begin{enumerate}
\item $\textrm{sgn}\left(x_{i}^{*}-BR_{i}\left(x_{-i}^{*}\right)\right)=\textrm{sgn}\left(\frac{d\pi_{i}}{dx_{j}}\right)$
iff  $\frac{d\pi_{i}^{2}}{dx_{i}dx_{j}}>0$.
\item $\textrm{sgn}\left(x_{i}^{*}-BR_{i}\left(x_{-i}^{*}\right)\right)=-\textrm{sgn}\left(\frac{d\pi_{i}}{dx_{j}}\right)$
iff $\frac{d\pi_{i}^{2}}{dx_{i}dx_{j}}>0$.
\end{enumerate}
\end{cor}

\subsection{\label{subsec:Direction-of-Analytics}Direction of Analytics Bias
\label{subsec:Mis-estimating-the-Sensitivity}}

Our next result characterizes the condition for analysts to either
overestimate or underestimate the sensitivity of demand in any naive
analytics equilibrium. Specifically, we show that all players overestimate
the sensitivity of demand (i.e., $\alpha_{i}>1$ for each player $i$)
iff the number of positive derivatives is even among the three monotone
derivative: strategic complementarity $\frac{d\pi_{i}^{2}}{dx_{i}dx_{j}}$,
payoff externalities $\frac{d\pi_{j}}{dx_{i}}$, and partial derivative
$\frac{\partial\pi_{i}}{\partial x_{i}}$.
\begin{prop}
\label{pro:alpha}If $\Gamma$ satisfies Assumptions \ref{assu:Monotone-externalities}--\ref{assu:consistent-second-adaptation}
and $\left(\alpha^{*},x^{*}\right)$ is a naive analytics equilibrium:
\begin{enumerate}
\item $\alpha_{i}^{*}>1$ for any player $i\in N$ iff $\frac{\partial\pi_{i}}{\partial x_{i}}\cdot\frac{d\pi_{j}}{dx_{i}}\cdot\frac{d\pi_{i}^{2}}{dx_{i}dx_{j}}>0$.
\item $\alpha_{i}^{*}<1$ for any player $i\in N$ iff $\frac{\partial\pi_{i}}{\partial x_{i}}\cdot\frac{d\pi_{j}}{dx_{i}}\cdot\frac{d\pi_{i}^{2}}{dx_{i}dx_{j}}<0$.
\end{enumerate}
\end{prop}
\begin{proof}
The result is derived as follows:
\[
\begin{array}{c}
\alpha_{i}>1\\
\alpha_{i}<1
\end{array}\,\underset{(1)}{\Leftrightarrow}\,\begin{array}{c}
\left|\frac{\partial q_{i}^{\alpha_{i}}}{\partial x_{i}}\right|>\left|\frac{\partial q_{i}}{\partial x_{i}}\right|\\
\left|\frac{\partial q_{i}^{\alpha_{i}}}{\partial x_{i}}\right|<\left|\frac{\partial q_{i}}{\partial x_{i}}\right|
\end{array}\,\underset{}{\Leftrightarrow}\,\begin{array}{c}
\left(\frac{\partial\pi_{i}}{\partial q_{i}}\cdot\frac{\partial q_{i}^{\alpha_{i}}}{\partial x_{i}}>\frac{\partial\pi_{i}}{\partial q_{i}}\cdot\frac{\partial q_{i}}{\partial x_{i}}>0\right)\textrm{OR}\left(\frac{\partial\pi_{i}}{\partial q_{i}}\cdot\frac{\partial q_{i}^{\alpha_{i}}}{\partial x_{i}}<\frac{\partial\pi_{i}}{\partial q_{i}}\cdot\frac{\partial q_{i}}{\partial x_{i}}<0\right)\\
\left(\frac{\partial\pi_{i}}{\partial q_{i}}\cdot\frac{\partial q_{i}}{\partial x_{i}}>\frac{\partial\pi_{i}}{\partial q_{i}}\cdot\frac{\partial q_{i}^{\alpha_{i}}}{\partial x_{i}}>0\right)\textrm{OR}\left(\frac{\partial\pi_{i}}{\partial q_{i}}\cdot\frac{\partial q_{i}}{\partial x_{i}}<\frac{\partial\pi_{i}}{\partial q_{i}}\cdot\frac{\partial q_{i}^{\alpha_{i}}}{\partial x_{i}}<0\right)
\end{array}
\]
\[
\,\underset{(2)}{\Leftrightarrow}\,\begin{array}{c}
\left(\frac{\partial\pi_{i}}{\partial x_{i}}<0\,\&\,\frac{d\pi_{i}^{\alpha_{i}}}{dx_{i}}>\frac{d\pi_{i}}{dx_{i}}\right)\textrm{OR}\left(\frac{\partial\pi_{i}}{\partial x_{i}}>0\,\&\,\frac{d\pi_{i}^{\alpha_{i}}}{dx_{i}}<\frac{d\pi_{i}}{dx_{i}}\right)\\
\left(\frac{\partial\pi_{i}}{\partial x_{i}}<0\,\&\,\frac{d\pi_{i}^{\alpha_{i}}}{dx_{i}}<\frac{d\pi_{i}}{dx_{i}}\right)\textrm{OR}\left(\frac{\partial\pi_{i}}{\partial x_{i}}>0\,\&\,\frac{d\pi_{i}^{\alpha_{i}}}{dx_{i}}>\frac{d\pi_{i}}{dx_{i}}\right)
\end{array}
\]
\[
\,\underset{(3)}{\Leftrightarrow}\,\begin{array}{c}
\left(\frac{\partial\pi_{i}}{\partial x_{i}}<0\,\&\,x_{i}^{*}<BR_{i}\left(x_{-i}^{*}\right)\right)\textrm{OR}\left(\frac{\partial\pi_{i}}{\partial x_{i}}>0\,\&\,x_{i}^{*}>BR_{i}\left(x_{-i}^{*}\right)\right)\\
\left(\frac{\partial\pi_{i}}{\partial x_{i}}<0\,\&\,x_{i}^{*}>BR_{i}\left(x_{-i}^{*}\right)\right)\textrm{OR}\left(\frac{\partial\pi_{i}}{\partial x_{i}}>0\,\&\,x_{i}^{*}<BR_{i}\left(x_{-i}^{*}\right)\right)
\end{array}
\]
\[
\,\underset{(4)}{\Leftrightarrow}\,\begin{array}{c}
\left(\frac{\partial\pi_{i}}{\partial x_{i}}<0\,\&\,\frac{d\pi_{i}^{2}}{dx_{i}dx_{j}}\cdot\frac{d\pi_{i}}{dx_{j}}<0\right)\textrm{OR}\left(\frac{\partial\pi_{i}}{\partial x_{i}}>0\,\&\,\frac{d\pi_{i}^{2}}{dx_{i}dx_{j}}\cdot\frac{d\pi_{i}}{dx_{j}}>0\right)\\
\left(\frac{\partial\pi_{i}}{\partial x_{i}}<0\,\&\,\frac{d\pi_{i}^{2}}{dx_{i}dx_{j}}\cdot\frac{d\pi_{i}}{dx_{j}}<0\right)\textrm{OR}\left(\frac{\partial\pi_{i}}{\partial x_{i}}>0\,\&\,\frac{d\pi_{i}^{2}}{dx_{i}dx_{j}}\cdot\frac{d\pi_{i}}{dx_{j}}<0\right),
\end{array}
\]

where
\begin{enumerate}
\item $\underset{(1)}{\Leftrightarrow}$ is implied by the assumption that
$\left|\frac{\partial q_{i}^{\alpha_{i}}}{\partial x_{i}}\right|$
is increasing in $\alpha_{i}$,
\item $\underset{(2)}{\Leftrightarrow}$is due to Assumption \ref{assu:Monotone-partial-deriv}
and (\ref{eq:opposite-signs}) (the observation that $\frac{\partial\pi_{i}}{\partial x_{i}}$
and $\frac{\partial\pi_{i}}{\partial q_{i}}\cdot\frac{\partial q_{i}}{\partial x_{i}}$
must have opposite signs).
\item $\underset{(3)}{\Leftrightarrow}$ is due to robust concavity (Assumption
\ref{assu:Strong-com-subs}), and
\item $\underset{(4)}{\Leftrightarrow}$ is implied by Proposition \ref{prop:percieved-BR-vs-unbiased-BR}.\qedhere 
\end{enumerate}
\end{proof}
\begin{rem}
\label{rem:Heifetz-et-al}\citet[Lemma 1]{heifetz2007dynamic} present
a similar result to Proposition \ref{pro:alpha} in a setup of evolution
of subjective preferences (which can also be interpreted as delegation)
for \emph{symmetric two-player }games, under two strong assumptions:
(1) unique $\alpha$-equilibrium $x^{*}\left(\alpha\right)$ for each
$\alpha\in A$, which allows to define a payoff function $\tilde{\pi}_{i}\left(\alpha\right)\equiv\pi_{i}\left(x^{*}\left(\alpha\right)\right)$
for each bias profile, and (2) structural assumptions on $\tilde{\pi}_{i}\left(\alpha\right)$:
strict concavity and $\left|\frac{d^{2}\tilde{\pi}_{i}\left(\alpha\right)}{d\alpha_{1}d\alpha_{2}}\right|<\left|\frac{d^{2}\tilde{\pi}_{i}\left(\alpha\right)}{d\alpha_{1}^{2}}\right|.$
\\
Our contribution is two-fold. First, the change of mechanism from
subjective preferences / delegation to biased analytics is important,
and it induces novel predictions and policy implications. Second,
we extend \citet[Lemma 1]{heifetz2007dynamic} to arbitrary $n$-player
games and substantially relax the two strong assumptions mentioned
above.\footnote{The two strong assumptions are required for their dynamic convergence
result, namely, \citealp[Theorem 2]{heifetz2007dynamic}; we do not
present a formal dynamic convergence result in this paper.} Note, in particular, that the assumption of $\tilde{\pi}_{i}\left(\alpha\right)$
being concave does not hold in our applications, and is not required.
\end{rem}

\subsection{Pareto Dominance\label{subsec:Pareto-Dominance}}

Finally, we show that any naive analytics equilibrium of any game
with strategic complements Pareto improves over the Nash equilibrium
payoffs of the underlying game. Moreover, the converse is true for
symmetric equilibria of games with strategic substitutes: Any symmetric
naive analytics equilibrium of any game with strategic substitutes
(which admits a symmetric Nash equilibrium) is Pareto dominated by
the Nash equilibrium of the underlying game. 
\begin{defn}
Strategy profile $x$ is symmetric if $x_{i}=x_{j}$ for any pair
$i,j\in N$.
\end{defn}
\begin{prop}
\label{prop:Pareto-dominate}If $\Gamma$ satisfies Assumptions \ref{assu:Monotone-externalities}--\ref{assu:consistent-second-adaptation}
and $\left(\alpha^{*},x^{*}\right)$ is a naive analytics equilibrium:
\end{prop}
\begin{enumerate}
\item If $\frac{d\pi_{i}^{2}}{dx_{i}dx_{j}}>0$ (strategic complements),
then there exists a Nash equilibrium $x^{NE}$ s.t. for each $i\in N$:
(i) $\pi_{i}\left(x^{*}\right)>\pi_{i}\left(x_{i}^{NE}\right)$ and
(ii) $sgn\left(x_{i}^{*}-x_{i}^{NE}\right)=sgn\left(\frac{d\pi_{j}}{dx_{i}}\right)$.
\item If $\frac{d\pi_{i}^{2}}{dx_{i}dx_{j}}<0$ (strategic substitutes),
then for each Nash equilibrium $x^{NE}$ there exists a player $i\in N$
s.t. $sgn\left(x_{i}^{*}-x_{i}^{NE}\right)=-sgn\left(\frac{d\pi_{j}}{dx_{i}}\right)$.
Moreover, if $x^{NE}$ and $x^{*}$ are both symmetric profiles, then
$\pi_{i}\left(x^{*}\right)<\pi_{i}\left(x^{NE}\right)$ for each player
$i\in N$.
\end{enumerate}
\begin{proof}[Proof of part (1); see Appx. \ref{subsec:Pareto-Dominance} for the
similar proof of part (2) and for the lemmas]
~ \\
Corollary \ref{cor:sign-strategic-compl-supp} implies that $\textrm{sgn}\left(x_{i}^{*}-BR_{i}\left(x_{-i}^{*}\right)\right)=\textrm{sgn}\left(\frac{d\pi_{i}}{dx_{j}}\right)$
for each $i\in N$. A standard argument for games with strategic complements
(proven in Lemma \ref{lem:strategic-complements-standard-lemma-1})
implies that there must be a Nash equilibrium $x^{NE}$, such that
$sgn\left(x_{i}^{*}-x_{i}^{NE}\right)=sgn\left(\frac{d\pi_{j}}{dx_{i}}\right)$
for each $i\in N$. We are left with showing that $\pi_{i}\left(x^{*}\right)>\pi_{i}\left(x_{i}^{NE}\right)$.
Let $x'=\left(x_{i}^{NE},x'_{-i}\right)$ be a $\left(x_{i}^{NE},\alpha_{-i}\right)$-equilibrium.
It is simple to see (and formally proven in Lemma \ref{lem-sign-depends-only-on-alpha_j})
that $\textrm{sgn}\left(x'_{j}-BR_{j}\left(x'_{-j}\right)\right)=\textrm{sgn}\left(\frac{d\pi_{j}}{dx_{i}}\right)$
for any $j\neq i$. due to Lemma \ref{lem-sign-depends-only-on-alpha_j},
this implies that $sgn\left(x'_{j}-x_{j}^{NE}\right)=sgn\left(\frac{d\pi_{j}}{dx_{i}}\right)$
for each $j\neq i\in N$. Proposition \ref{pro:stackelberg-leader}
implies that $\pi_{i}\left(x^{*}\right)\geq\pi_{i}\left(x'\right)$
(due to $x_{i}^{*}$ being a Stackelberg-leader action with respect
to $\alpha_{-i}$). Finally, the payoff externalities imply that $\pi_{i}\left(x'\right)=\pi_{i}\left(x_{i}^{NE},x'_{-i}\right)>\pi_{i}\left(x^{NE}\right)$.
Combining the two inequalities yield that $\pi_{i}\left(x^{*}\right)>\pi_{i}\left(x_{i}^{NE}\right)$. 
\end{proof}

\subsection{Summary of Results\label{subsec:Summary-of-Results}}

Assumptions \ref{assu:Monotone-externalities}--\ref{assu:Monotone-partial-deriv}
map to eight combinations of the signs of the strategic complementarity
$\frac{d\pi_{i}^{2}}{dx_{i}dx_{j}}$, payoff externalities $\frac{d\Pi_{i}}{dx_{j}},$and
partial derivative $\frac{\partial\pi_{i}}{\partial x_{i}}$. Effectively,
these eight combinations define four unique classes of games since
relabeling strategies as their negative values (i.e., replacing $x_{i}$
with $-x_{i}$) results in essentially the same game with opposite
signs to each of three monotone derivatives (see, Fact \ref{fact:relabeling}
in Appendix \ref{subsec:Fact-on-Relabeling}). Table \ref{tab:Summary-of-Results}
summarizes our results for these four classes (and the remaining four
classes are presented in Appendix \ref{subsec:Fact-on-Relabeling}).

\begin{table}[h]
\caption{\label{tab:Summary-of-Results}Summary of Results Under Assumptions
\ref{assu:Monotone-externalities}--\ref{assu:consistent-second-adaptation}}

\medskip{}

~\\
\begin{tabular}{|>{\centering}p{3.2cm}|>{\centering}p{2.35cm}|>{\centering}p{1.1cm}|>{\centering}p{1.2cm}|>{\centering}p{1.7cm}|>{\centering}p{1.7cm}|>{\centering}p{2.4cm}|}
\hline 
~\\
Applications (Sections \ref{subsec:Price-Competition}--\ref{sec:Team-Production}) & Strategic comp./sub.\\
~\\
 $\frac{d\pi_{i}^{2}}{dx_{i}dx_{j}}$ & Payoff~\\
Extr.\\
~\\
 $\frac{d\Pi_{i}}{dx_{j}}$ & Partial derivative\\
$\frac{\partial\pi_{i}}{\partial x_{i}}$ & Perceived best replying (Prop. \ref{prop:percieved-BR-vs-unbiased-BR}) & Analytics bias\\
(Prop. \ref{pro:alpha}) & Pareto Dominance (Prop. \ref{prop:Pareto-dominate})\tabularnewline
\hline 
Price competition

w. subs. goods

(motiv. example) & \multirow{2}{2.35cm}{\textbf{\textcolor{green}{~}}\\
\textbf{\textcolor{green}{~~~~~~+}}\\
~\\
~~~Strategic complements\textbf{\textcolor{green}{{} }}} & \textbf{\textcolor{green}{~}}\\
\textbf{\textcolor{green}{+}} & \textbf{\textcolor{green}{~}}\\
\textbf{\textcolor{green}{+}} & \multirow{4}{1.7cm}{\\
~\\
~\\
\textcolor{blue}{~~}\\
\textcolor{blue}{~}\\
~$\,\,x_{i}^{*}>$\\
$BR_{i}\left(x_{-i}^{*}\right)$} & \textcolor{brown}{~}\\
\textcolor{brown}{$\alpha_{i}^{*}<1$} & ~\\
$\forall i\,x_{i}^{*}>x_{i}^{NE}$\\
NAE\tabularnewline
\cline{1-1} \cline{3-4} \cline{4-4} \cline{6-6} 
Advertising with

positive exter. \& team production &  & \textbf{~}\\
\textbf{\textcolor{green}{+}} & \textbf{\textcolor{red}{~}}\\
\textbf{\textcolor{red}{-}} &  & \textcolor{blue}{~}\\
\textcolor{blue}{$\alpha_{i}^{*}>1$}\\
 & Pareto

dominates NE.\tabularnewline
\cline{1-4} \cline{2-4} \cline{3-4} \cline{4-4} \cline{6-7} \cline{7-7} 
Price competition

with complement goods & \multirow{2}{2.35cm}{\textbf{\textcolor{green}{~}}\\
\textbf{\textcolor{green}{~~~~~~}}\textbf{\textcolor{red}{-}}\\
~\\
~~~Strategic Substitutes} & \textbf{~}\\
\textbf{\textcolor{red}{-}} & \textbf{\textcolor{green}{~}}\\
\textbf{\textcolor{green}{+}} &  & \textcolor{brown}{~}\\
\textcolor{brown}{$\alpha_{i}^{*}<1$}\\
 & ~\\
$\exists i\,x_{i}^{*}>x_{i}^{NE}$\\
Symmetric\tabularnewline
\cline{1-1} \cline{3-4} \cline{4-4} \cline{6-6} 
Advertising with

negative externalities &  & \textbf{\textcolor{red}{~}}\\
\textbf{\textcolor{red}{-}} & \textbf{\textcolor{red}{~}}\\
\textbf{\textcolor{red}{-}} &  & \textcolor{blue}{~}\\
\textcolor{blue}{$\alpha_{i}^{*}>1$}\\
 & NE Pareto dominates sym. NAE\tabularnewline
\hline 
\end{tabular}
\end{table}

\section{Applications\label{sec:Applications}}

We present three applications of our model: price competition, advertising
competition, and team production. Before analyzing the applications,
we present an auxiliary characterization result for the level of naivete
in a naive analytics equilibrium.

\subsection{Auxiliary Result: Characterization of $\alpha^{*}$}

In what follows we obtain a useful characterization for the level
of naivete in a naive analytics equilibrium. The presentation of this
result is much simplified under the assumption of unique perceived
best responses (which holds in all the applications). 
\begin{assumption}
\label{assu:well-behave-best-responses} For each player $i\in N$,
\textup{strategy} $x_{i}\in X_{i}$, and profile $\alpha_{-i}\in A^{n-1}$,
there exists a unique $\left(x_{i},\alpha_{-i}\right)$-equilibrium,
which we denote by $x_{-i}\left(x_{i},\alpha_{-i}\right)\in X_{-i}$.
\end{assumption}
Our next result shows that the magnitude of bias of player $i$ in
assessing her sensitivity of demand (i.e., $\frac{\partial q_{i}^{\alpha_{i}^{*}}}{\partial x_{i}}-\frac{\partial q_{i}\left(x^{*}\right)}{\partial x_{i}}$)
is a sum of the products of the impact of the player's strategy on
each of her opponent's strategies (i.e., ($\frac{\text{d\ensuremath{x_{j}\left(x_{i}^{*},\alpha_{j}^{*}\right)}}}{dx_{i}}$)
times the impact of the opponent's strategy on the player's demand
(i.e., $\frac{\partial q_{i}\left(x^{*}\right)}{\partial x_{j}}$).
Formally:
\begin{claim}
\label{claim:Stackelberg-derivative}Let $\left(\alpha^{*},x^{*}\right)$
be a naive analytics equilibrium of a game that satisfies Assumption
\ref{assu:well-behave-best-responses}. If $x_{i}^{*}\in Int\left(X_{i}\right)$,
$\alpha_{i}^{*}\in Int\left(A_{i}\right)$, and $x_{j}\left(x_{i},\alpha_{-i}\right)$
is differentiable at $\left(x_{i}^{*},\alpha_{-i}^{*}\right)$ for
each player $i,j$ then 
\[
\frac{\partial q_{i}^{\alpha_{i}^{*}}}{\partial x_{i}}-\frac{\partial q_{i}\left(x^{*}\right)}{\partial x_{i}}=\sum_{j\neq i}\frac{\text{d\ensuremath{x_{j}\left(x_{i}^{*},\alpha_{j}^{*}\right)}}}{dx_{i}}\cdot\frac{\partial q_{i}\left(x^{*}\right)}{\partial x_{j}}.
\]
\end{claim}
\begin{proof}
The claim is implied by substituting the biased FOC (\ref{eq:biased-FOC-SOC})
in the FOC of a Stackelberg-leader strategy (Proposition \ref{pro:stackelberg-leader}):
\[
0=\frac{d\pi_{i}\left(x_{i}^{*},q_{i}\left(x_{i}^{*},x_{-i}\left(x_{i}^{*},\alpha_{j}^{*}\right)\right)\right)}{dx_{i}}=\frac{\partial\pi_{i}}{\partial x_{i}}+\frac{\partial\pi_{i}}{\partial q_{i}}\cdot\frac{\partial q_{i}}{\partial x_{i}}+\sum_{j\neq i}\frac{\text{d\ensuremath{x_{j}\left(x_{i},\alpha_{-i}^{*}\right)}}}{dx_{i}}\cdot\frac{\partial q_{i}}{\partial x_{j}}\cdot\frac{\partial\pi_{i}}{\partial q_{i}}=
\]
\[
0=\underset{=0\,\,\,\textrm{(FOC\,\,of\,\,\eqref{eq:biased-FOC-SOC})}}{\underbrace{\frac{\partial\pi_{i}}{\partial x_{i}}+\cdot\frac{\partial\pi_{i}}{\partial q_{i}}\cdot\frac{\partial q_{i}^{\alpha_{i}^{*}}}{\partial x_{i}}\cdot\frac{\partial q_{i}}{\partial x_{i}}}}+\left(\frac{\partial\pi_{i}}{\partial q_{i}}-\frac{\partial q_{i}^{\alpha_{i}^{*}}}{\partial x_{i}}\right)\cdot\frac{\partial q_{i}}{\partial x_{i}}+\sum_{j\neq i}\frac{\text{d\ensuremath{x_{j}\left(x_{i},\alpha_{-i}^{*}\right)}}}{dx_{i}}\cdot\frac{\partial q_{i}}{\partial x_{j}}\cdot\frac{\partial\pi_{i}}{\partial q_{i}}\Rightarrow
\]
\[
0=\left(\frac{\partial\pi_{i}}{\partial q_{i}}-\frac{\partial q_{i}^{\alpha_{i}^{*}}}{\partial x_{i}}\right)\cdot\frac{\partial q_{i}}{\partial x_{i}}+\sum_{j\neq i}\frac{\text{d\ensuremath{x_{j}\left(x_{i},\alpha_{-i}^{*}\right)}}}{dx_{i}}\cdot\frac{\partial q_{i}}{\partial x_{j}}\Rightarrow\left(\frac{\partial\pi_{i}}{\partial q_{i}}-\frac{\partial q_{i}^{\alpha_{i}^{*}}}{\partial x_{i}}\right)=\sum_{j\neq i}\frac{\text{d\ensuremath{x_{j}\left(x_{i},\alpha_{-i}^{*}\right)}}}{dx_{i}}\cdot\frac{\partial q_{i}}{\partial x_{j}}.\qedhere
\]
\end{proof}

\subsection{Price Competition\label{subsec:Price-Competition}}

\paragraph{Underlying Game}

The underlying game $G_{p}=\left(N,X,q,\pi\right)$ is a price competition
with linear demand (generalizing the motivating example of Section
\ref{sec:motivation}). Each seller $i\in N$ sets a price $x_{i}\in X_{i}$,
where $X_{i}=\mathbb{R}_{+}$ if $c_{i}>0$ and $X_{i}=\left[0,\frac{a_{i}}{b_{i}}\right]$
if $c_{i}<0$ . Limiting the maximum price to $\frac{a_{i}}{b_{i}}$
is without loss of generality because setting a higher price implies
that the seller's demand cannot be positive. The demand of each firm
$i$ is:\footnote{Our results remain the same if one adapts the demand function to be
non-negative, i.e., $q_{i}\left(x\right)=\max\left(a_{i}-b_{i}x_{i}+c_{i}\cdot x_{-i},0\right)$.
We refrain from doing so, as it will require rephrasing the assumptions
of the general model (Section \ref{sec:General-Results}) in a more
cumbersome way: assuming demand to be non-negative, and allowing the
various monotone derivatives to be equal to zero when the demand is
equal to zero.}

\begin{equation}
q_{i}\left(x\right)=a_{i}-b_{i}x_{i}+c_{i}\cdot\bar{x},\,\,\,\,\,\,\,\bar{x}\equiv\sum_{j}w_{j}\cdot x_{j}\label{eq:linear-bertrand-demand}
\end{equation}
where $\bar{x}$ is a weighted mean price with weights $w_{j}>0$
and $\sum_{j}w_{j}=1$, and where $c_{i}\cdot c_{j},a_{i},b_{i}>0$,
and $\left|c_{i}\right|<b_{i}$ for each $i,j\in N$. The sign of
$c_{i}$ (which coincides with sign of $c_{-i}$) determines whether
the sold goods are substitutes ($c_{i}>0$ as in the standard differentiated
Bertrand competition) or complements ($c_{i}<0$ as in a case of two
stores that sell complementary goods, such as kitchen appliances and
cooking ingredients, or as in the case of adjacent stores that sell
unrelated goods, and a price decrease in one of the stores attracts
more customers that also visit the neighboring store). The inequality
$\left|c_{i}\right|<b_{i}$ constrains the cross-elasticity parameters
to be sufficiently small relative to the own-elasticity parameters.
If $c_{i}<0$ then we further require an additional upper bound on
the the cross-elasticity: $\sum_{j\neq i}\frac{\left|c_{j}\right|}{b_{j}}<\frac{1}{w_{i}}$
 which implies Assumption \ref{assu:consistent-second-adaptation}.

Finally, the profit of each firm is given by $\pi_{i}\left(x_{i},q_{i}\right)=q_{i}\left(x\right)\cdot x_{i}$.
This profit function corresponds to constant marginal costs, which
have been normalized to zero. Observe that game $G_{p}$ has strategic
complements if $c_{i}>0$ and strategic substitutes if $c_{i}<0$.
Lemma \ref{lem:unqiue-alpha-eq-Price-competition-game} in Appendix
\ref{subsec:Proof-of-oligopoly} shows that the price competition
game admits a unique Nash equilibrium, which we denote by $x^{NE}$. 

\paragraph{Results}

Let $\Gamma_{p}=\left(G_{p},R_{++},f_{m}\right)$ be the analytics
game with multiplicative biases $A\equiv R_{++}$, $f_{m}\left(\frac{\partial q_{i}}{\partial x_{i}},\alpha_{i}\right)=\alpha_{i}\cdot\frac{\partial q_{i}}{\partial x_{i}}$.
Our first result shows that price competition satisfies all the assumptions
of the general model, and that in any naive analytics equilibrium:
\begin{enumerate}
\item All firms hire analysts that underestimate the elasticity of demand
(a result which fit the direction of demand elasticity bias when
not controlling for price endogeneity, as in, e.g., Table 1 of \citealp{berry1994estimating}
and Table 2 of \citealp{villas1999endogeneity});
\item The prices are higher than in the Nash equilibrium; and
\item The naive analytics equilibrium Pareto dominates (resp., is Pareto
dominated by) the Nash equilibrium if the game has strategic complements
(resp., substitutes).
\end{enumerate}
\begin{prop}
\label{prop:price-satisfies-Assumptions}$\Gamma_{p}$ satisfies Assumptions
\ref{assu:Monotone-externalities}--\ref{assu:consistent-second-adaptation}.
Moreover, any NAE $\left(\alpha^{*},x^{*}\right)$ of $\Gamma_{p}$
satisfies:
\begin{enumerate}
\item Underestimation of elasticity of demand: $\alpha_{i}^{*}<1.$
\item Prices are higher than the Nash equilibrium prices: $x_{i}^{*}\left(\alpha^{*}\right)>x_{i}^{NE}.$
\item Pareto dominance relative to the Nash equilibrium:\\
 $c_{i}>0\Rightarrow\pi_{i}\left(x^{*}\right)>\pi_{i}\left(x^{NE}\right)$,
and $c_{i}<0\Rightarrow\pi_{i}\left(x^{*}\right)<\pi_{i}\left(x^{NE}\right)$.
\end{enumerate}
\end{prop}
The proof is presented in Appendix \ref{subsec:Proof-of-oligopoly}.

Our next results explicitly characterize the unique naive analytics
equilibrium in two important special cases (1) duopolies and (2) symmetric
oligopolies. This characterization will be used in analyzing the impact
of merges in Section \ref{sec:Implication-for-Merges}. Proposition
\ref{prop:price-competetion-unqiue-NAE} shows that in duopolistic
competition, both firms use the same level of naivete, which is decreasing
in the ratio between the cross-elasticity parameters and the own-elasticity
parameters. It is convenient to relabel the variables in the duopoly
case as follows: (1) $\tilde{c}_{i}=c_{i}w_{-i}$, and (2) $\tilde{b}_{i}=b_{i}-c_{i}w_{i}$.
With this relabeling the demand function is simplified to be: $q_{i}\left(x\right)=a_{i}-b_{i}x_{i}-c_{i}\bar{x}=a_{i}-\tilde{b}_{i}x_{i}+\tilde{c}_{i}x_{-i}$
\begin{prop}
\label{prop:price-competetion-unqiue-NAE}Assume that $n=2$. Then
$\Gamma_{p}$ admits a unique naive analytics equilibrium $\left(\alpha^{*},x^{*}\right)$
in which $\alpha_{1}^{*}=\alpha_{2}^{*}=\sqrt{1-\frac{\tilde{c}_{i}\tilde{c}_{-i}}{\tilde{b}_{i}\tilde{b}_{-i}}}.$
\end{prop}
\begin{proof}
Eq. \eqref{eq:biased-BR-oligopoly} in the proof of Proposition \ref{subsec:Proof-of-oligopoly}
implies that 
\[
x_{-i}\left(x_{i},\alpha_{-i}^{*}\right)=\frac{a_{-i}+\tilde{c}_{-i}x_{i}}{\left(1+\alpha_{-i}^{*}\right)\tilde{b}_{-i}}\,\Rightarrow\,\frac{dx_{-i}\left(x_{i},\alpha_{-i}^{*}\right)}{dx_{i}}=\frac{\tilde{c}_{-i}}{\left(1+\alpha_{-i}^{*}\right)\tilde{b}_{-i}}.
\]
This implies due to Claim \ref{pro:stackelberg-leader} (and the fact
that $\frac{\partial q_{i}}{\partial x_{-i}}=\tilde{c}_{i},\,\,\frac{\partial q_{i}}{\partial x_{i}}=-\tilde{b}_{i}$)
that $x_{i}^{*}$ must satisfy 
\begin{equation}
\alpha_{i}^{*}-1=\frac{\text{d\ensuremath{x_{-i}\left(x_{i},\alpha_{-i}^{*}\right)}}}{dx_{i}}\cdot\frac{\frac{\partial q_{i}}{\partial x_{-i}}}{\frac{\partial q_{i}}{\partial x_{i}}}=\frac{\tilde{c}_{i}\tilde{c}_{-i}}{-\left(1+\alpha_{-i}^{*}\right)\tilde{b}_{-i}\tilde{b}_{i}}\Leftrightarrow\left(1+\alpha_{-i}^{*}\right)\left(1-\alpha_{i}^{*}\right)=\frac{\tilde{c}_{i}\tilde{c}_{-i}}{\tilde{b}_{-i}\tilde{b}_{i}}.\label{eq:alpha-condition}
\end{equation}
 Observe that the RHS of \eqref{eq:alpha-condition} remains the same
when swapping $i$ and $-i$. This implies that $\alpha_{-i}^{*}$
and $\alpha_{i}^{*}$ must be equal, and that $\alpha_{1}^{*}=\alpha_{2}^{*}=\sqrt{1-\frac{\tilde{c}_{i}\tilde{c}_{-i}}{\tilde{b}_{-i}\tilde{b}_{i}}}$. 
\end{proof}
We note that the property that asymmetric firms must have the same
biasedness level in a naive analytics equilibrium does not hold when
there are more than two firms. 

Next we characterize the unique symmetric naive analytics equilibrium
in symmetric oligopolies, in which the demand parameters are the same
for all firms: $b_{i}=b$ and $c_{i}=c$ for each firm $i\in N$.
Proposition \ref{prop:symmetric-oligopoly-prices} shows that the
level of bias $\alpha_{i}^{*}$ depends only on the ratio $\left|\frac{b}{c}\right|$
and on $n$, and that it is increasing in $\left|\frac{b}{c}\right|$.
The dependency on $n$ is non-monotone. A monopoly is unbiased ($\alpha_{i}^{*}=1$).
By contrast, firms in duopoly have the strongest bias (i.e., $\alpha_{i}^{*}<1$
is furthest from one), and the level of biasedness decreases (i.e.,
$\alpha_{i}^{*}$ becomes closer to one) the larger the number of
competitors (converging back to 1 when $n\rightarrow\infty$). Formally,
\begin{prop}
\label{prop:symmetric-oligopoly-prices}Assume that the price competition
is symmetric. Then $\Gamma_{p}$ admits a unique symmetric naive analytics
equilibrium $\left(\alpha^{*},x^{*}\right)$, where $\alpha_{i}^{*}\left(\left|\frac{b}{c}\right|,n\right)$
depends only on $\left|\frac{b}{c}\right|$ and $n$ and satisfies
the following properties: (1) $\alpha_{i}^{*}\left(\left|\frac{b}{c}\right|,n\right)$
is increasing in $\left|\frac{b}{c}\right|$ for $n\geq2$, (2) $\alpha_{i}^{*}\left(\left|\frac{b}{c}\right|,1\right)=1>\alpha_{i}^{*}\left(\left|\frac{b}{c}\right|,2\right)$,
(3) $\alpha_{i}^{*}\left(\left|\frac{b}{c}\right|,n\right)<\alpha_{i}^{*}\left(\left|\frac{b}{c}\right|,n+1\right)<1$
for $n\geq2$, and (4) 
\[
\lim_{n\rightarrow\infty}\alpha_{i}^{*}\left(\left|\frac{b}{c}\right|,n\right)=\lim_{\frac{b}{c}\rightarrow\infty}\alpha_{i}^{*}\left(\left|\frac{b}{c}\right|,n\right)=1.
\]
\end{prop}
The proof is given in Appendix \ref{subsec:proof-of-symmetric-oligopoly}.
The intuition for $\alpha_{i}^{*}\left(\left|\frac{b}{c}\right|,n\right)$
increasing in $\left|\frac{b}{c}\right|$ is the same as in the case
of two competing firms. A monopoly has no strategic advantage from
biased estimation (and, thus, $\alpha_{i}^{*}\left(\left|\frac{b}{c}\right|,1\right)=1$).
The intuition why $\alpha_{i}^{*}\left(\left|\frac{b}{c}\right|,n\right)$
is increasing in $n$ for $n\geq2$ is that the larger the number
of competing firms, the smaller the strategic impact of a player's
price on the biased best reply of her competitors, and thus the strategic
advantage of having biased analytics decreases. 

Proposition \ref{prop:symmetric-oligopoly-prices} is illustrated
in Figure \ref{fig:Examples-of-symmetric-oligoply}. The left panel
shows the equilibrium level of bias $\alpha_{i}^{*}$ in symmetric
competition as a function of the number of competing firms $n$ for
$a=20$, $b=1$, and $c=0.7,\,0.9$. The right panel shows the (unbiased)
Nash equilibrium prices and the naive analytics equilibrium prices
in these cases. The figure shows that that the differences between
the NAE and the Nash equilibrium prices are maximal for duopolies,
and gradually decrease as $n$ increases. This suggests that a counter-factual
economic analysis of the impact on two firms merging on the equilibrium
prices, which ignores the biased analytics, would overestimate the
price increase induced by a duopoly merging to a monopoly, and it
would underestimate the price induced by all other mergers. This argument
is detailed and illustrated in Section \ref{sec:Implication-for-Merges}.
\begin{figure}[tbh]
\caption{\textbf{\label{fig:Examples-of-symmetric-oligoply}}Biases and Prices
in Symmetric Oligopolies ($a=20$, $b=1$)}
\includegraphics[width=0.49\textwidth]{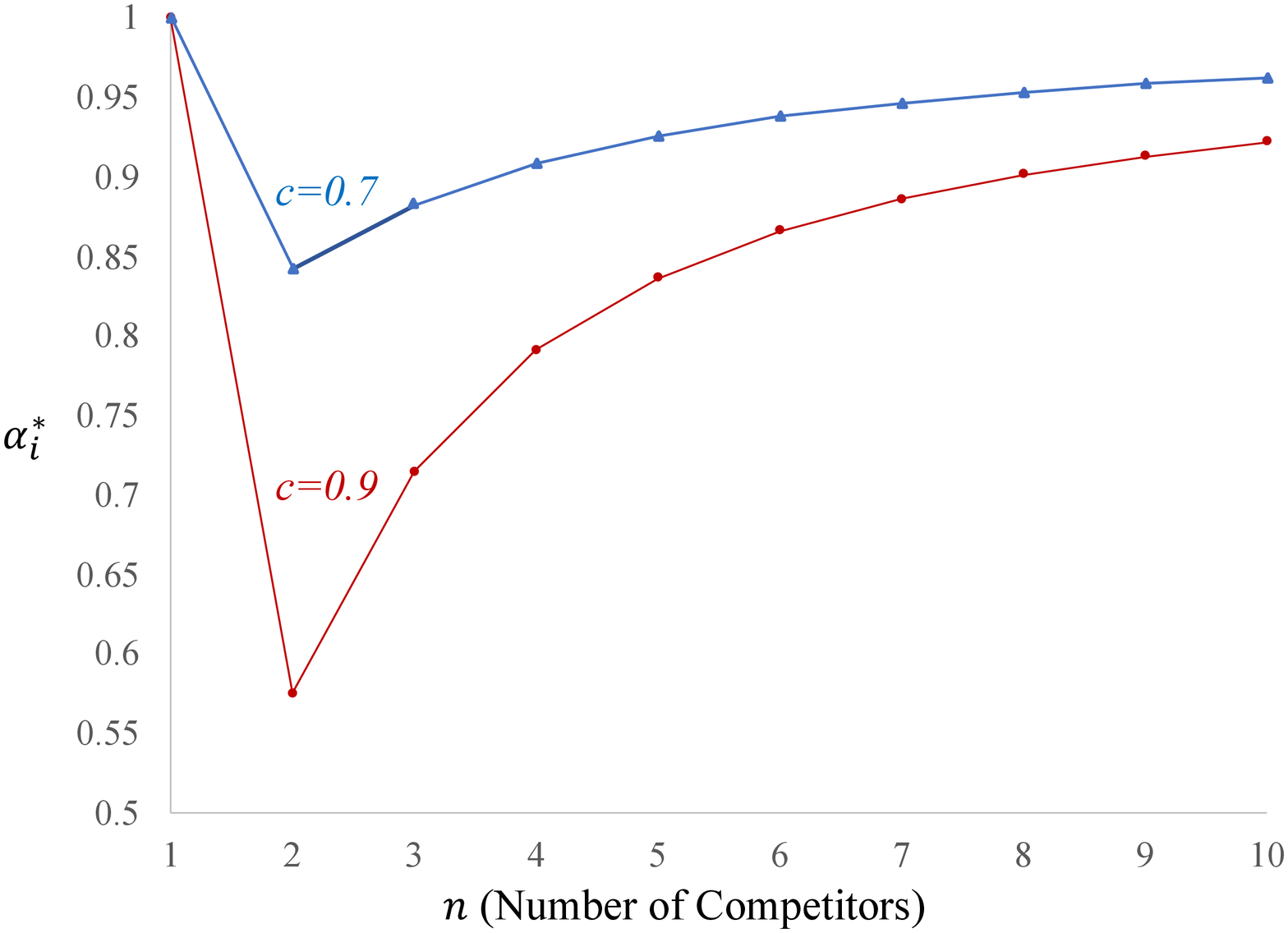}~\includegraphics[width=0.49\textwidth]{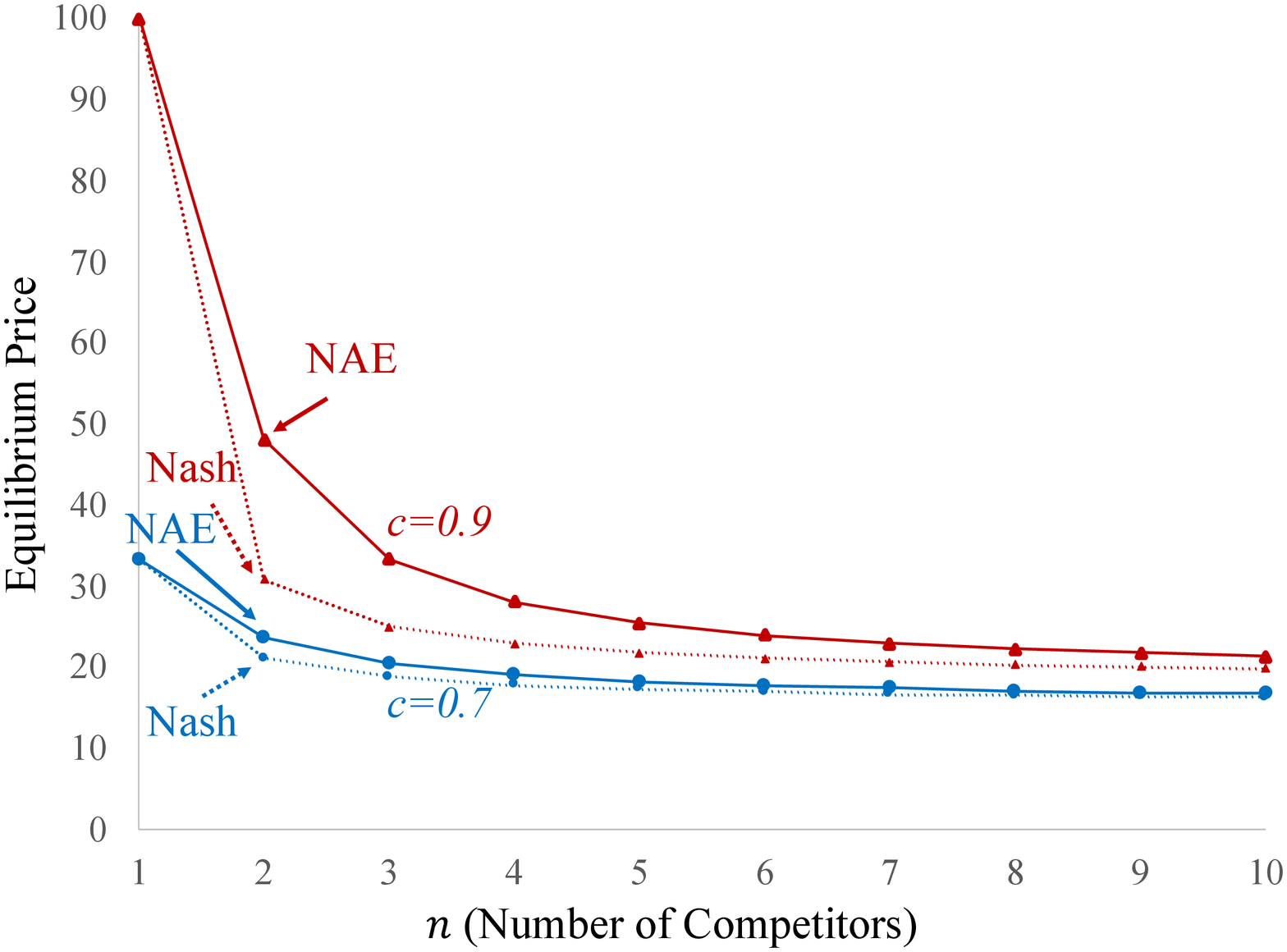}

{\footnotesize{}The left panel shows the symmetric naive analytics
equilibrium level of biasedness $\alpha_{i}^{*}$ as a function of
the number of competing firms for $c=0.7$ and $c=0.9$. The right
panel shows the (unbiased) Nash equilibrium prices and the naive analytics
equilibrium prices in each case.}{\footnotesize\par}
\end{figure}

\paragraph{Micro-Foundations for $\alpha_{i}<1$}

We conclude this application by discussing plausible mechanisms that
can induce analysts to unintentionally underestimate demand sensitivities
and elasticities:
\begin{enumerate}
\item The first mechanism has been illustrated in the motivating example
(Section \ref{sec:motivation}), and is described here again briefly.
Suppose the daily demand of each firm has a random component for which
the firm employees observe an informative signal. For example, the
employees observe the weather forecast which is correlated with the
realized demand. If the employees prefer to set lower sale prices
on days with low demand, e.g., because they have more free time to
change prices, then naive analysts who allow the employees to choose
the days for price discounts would induce a correlation between low
demand days and low prices. As a result, they will under-estimate
of the elasticity of demand. A sophisticated analyst who closely ensures
that price discounts are set uniformly at random, or accurately controls
for the weather forecast in his econometric analysis, would yield
an accurate estimation of the elasticity.
\item The second mechanism is formalized in Appx. \ref{subsec:Biased-Price-Competition},
and is presented here briefly. Employees of competing stores might
choose similar days for price discounts. One examples is setting discounts
by season (holidays) or for specific days of the week (weekends).
Another example is if the levels of their inventory is correlated
and discounts are given when the inventory level is high. This correlation
in prices would induce naive analysts who allow employees to set the
days with price discounts to under-estimate the price elasticity,
because in days with low prices the competitor is more likely to provide
a discount as well, making the response to price discounts seem weaker.
\end{enumerate}

\subsection{Advertising Competition\label{subsec:Advertising-Competition}}

Research that estimates the effectiveness of advertising often showed
that extra care is required to arrive at non-biased estimates, and
that not correcting for these biases often results in overestimating
advertising effectiveness \citep{lodish1995tv,blake2015consumer,gordon2019comparison,shapiro2019generalizable}.
In the second application of our model we analyze a duopoly competition
in advertising to understand why firms might benefit from naive analytics
that overestimates the effectiveness of advertising.

\paragraph{Underlying Game}

The underlying game $G_{a}=\left(N=\left\{ 1,2\right\} ,X,q,\pi\right)$
describes an advertising competition among two firms. Firm $i\in\left\{ 1,2\right\} $
sells a product with exogenous profit margins $p_{i}>0$ per unit
sold. The expected demand of each firm depends on the advertising
budget $x_{i}\in X_{i}$ of both firms: 
\begin{equation}
q_{i}=a_{i}+b_{i}\cdot\sqrt{x_{i}}+c_{i}\cdot\sqrt{x_{i}\cdot x_{-i}},\label{eq:demand-2}
\end{equation}
where $a_{i},b_{i},c_{i}\cdot c_{-i}>0$. When $c_{i}>0$ the feasible
budget set is unrestricted ($X_{i}=\mathbb{R}_{+}$). When $c_{i}<0$,
we restrict the maximum budget to be $M_{i}=\frac{b_{-i}}{\left|c_{-i}\right|}$
(i.e., $X_{i}=\left[0,\frac{b_{i}}{\left|c_{i}\right|}\right]$).
This restriction is without loss of generality as additional budget
above $\frac{b_{i}}{\left|c_{i}\right|}$ cannot increase the firm's
demand when $c_{i}<0$.

In this market a firm's own advertising increases demand for its own
product ($b_{i}$ is positive), but the competitor's advertising may
affect the level of the increase. When $c_{i}>0$ the total category
demand increases when the competitor advertises, which increases the
effect of a firm's own advertising. When $c_{i}<0$ the effect of
a firm's own advertising decreases in the competitor's advertising
due to competition (e.g., over the same customers). We require the
sign of $c_{i}$ and $c_{-i}$ to coincide. Positive $c_{i}$-s might
correspond to a new category of goods, in which advertising attracts
attention to the category. Negative $c_{i}$-s might correspond to
a mature category in which advertising mainly causes consumers to
switch among competing goods. The payoff of player $i$ is given
by
\begin{equation}
\pi_{i}\left(x_{i},x_{-i}\right)=q_{i}\left(x_{i},x_{-i}\right)\cdot p_{i}-x_{i}=p_{i}\cdot\left(a_{i}+b_{i}\cdot\sqrt{x_{i}}+c_{i}\cdot\sqrt{x_{i}\cdot x_{-i}}\right)-x_{i}\text{.}\label{eq:demand-1}
\end{equation}
We require the advertising externalities to be sufficiently small:
$\left|c_{i}\right|\in\left(0,\frac{1}{p_{i}}\right)$, which implies
well-behaved interior Nash equilibrium. Further, if $c_{i}<0$ we
assume $\left|c_{i}\right|<\frac{b_{i}}{b_{-i}\cdot p_{-i}}.$ Lemma
\ref{lem:unqiue-alpha-eq-Price-competition-game-1} in Appendix \ref{subsec:Proof-of-advertising-result}
proves that advertising competition admits a unique Nash equilibrium,
which we denote by $x^{NE}$.

\paragraph{Naive Analytics Equilibrium}

Let $\Gamma_{a}=\left(G_{a},R_{++},f_{m}\right)$ be the analytics
game with multiplicative biases, i.e., $f_{m}\left(\frac{\partial q_{i}}{\partial x_{i}},\alpha_{i}\right)=\alpha_{i}\cdot\frac{\partial q_{i}}{\partial x_{i}}$.
Our result shows that advertising competition satisfies all the assumptions
of the general model, and that in any naive analytics equilibrium:
\begin{enumerate}
\item The firms hire analysts that overestimate the effectiveness of advertising
(and the level of bias is the same for both firms);
\item The advertising budgets are higher than in the Nash equilibrium; and
\item The naive analytics equilibrium Pareto dominates (resp., is Pareto
dominated by) the Nash equilibrium if the game has strategic complements
(resp., substitutes).
\end{enumerate}
\begin{prop}
\label{prop:advertising-result}$\Gamma_{a}$ satisfies Assumptions
\ref{assu:Monotone-externalities}--\ref{assu:consistent-second-adaptation}.
Moreover, any NAE $\left(\alpha^{*},x^{*}\right)$ satisfies:
\begin{enumerate}
\item Symmetric overestimation of ad effectiveness: $\alpha_{1}^{*}=\alpha_{2}^{*}=\frac{2}{1+\sqrt{1-c_{1}c_{2}p_{1}p_{2}}}>1$.
\item Prices are higher than the Nash equilibrium prices: $x_{i}^{*}\left(\alpha^{*}\right)>x_{i}^{NE}$
$\forall i\in N$.
\item Pareto dominance relative to the Nash equilibrium:\\
 $c_{i}>0\Rightarrow\pi_{i}\left(x^{*}\right)>\pi_{i}\left(x^{NE}\right)$,
and $c_{i}<0\Rightarrow\pi_{i}\left(x^{*}\right)<\pi_{i}\left(x^{NE}\right)$.
\end{enumerate}
\end{prop}
The proof is presented in Appendix \ref{subsec:Proof-of-advertising-result}.

\paragraph{Micro-Foundations for $\alpha_{i}>1$}

We conclude this section by providing micro-foundational examples
of cases that would cause firms to overestimate their advertising
effectiveness:
\begin{enumerate}
\item Similarly to price competition, correlation between advertising budgets
of firms with positive externalities would cause an overestimate of
ad effectiveness. If firms choose to increase budgets during the holidays,
or just prior to weekends, they will observe an increase in demand
beyond the effects of their own advertising. This correlation will
create an overestimate of advertising elasticity.
\item When online advertising is purchased on social media platforms, such
as on Facebook, the advertiser provides the advertising platform a
budget and a target metric. The platform's algorithm then targets
consumers in order to maximize the target metric under the budget
constraint. One common such metric is sales or purchases, and a strategy
to maximize this metric is to show ads to likely buyers of the product,
or to past purchasers of the product. Under such a strategy, an analysis
that compares the purchase rates of people that have seen ads to those
that have not seen ads will overestimate the effectiveness of advertising
\citep{berman2018beyond}.
\item If firms respond to decreased demand by increasing their advertising
budgets in the next time period, and if demand is noisy, a standard
``regression to the mean'' argument implies that demand is likely
to increase in the next period regardless of the additional advertising
budget. Failing to take this into account would lead to overestimation
of the advertising effectiveness, as we formally present in Appendix.
\ref{subsec:Biased-Advertising-Effectiveness}.
\end{enumerate}

\subsection{Team Production and Overconfidence\label{sec:Team-Production}}

Thus far we have interpreted $q\left(x\right)$ as market demand and
$\alpha_{i}\neq1$ as bias due to naive analytics. We now demonstrate
that our model applies in more general settings. Specifically, we
apply the model to an underlying game of team production with strategic
complementarity. Team production is common in partnerships and other
input games (see, e.g., \citealp{holmstrom1982moral,cooper1988coordinating,heller2020promises}).
Examples include salespeople who are compensated based on the performance
of the joint sales of a team, entrepreneurs who receive a share of
the exit value of a startup, and academic co-authors who benefit from
the impact of their joint paper. It is often observed that entrepreneurial
firms are founded by teams of overconfident founders \citep{astebro2014seeking,hayward2006hubris}.
Taking this perspective, we interpret $x_{i}$ as the contribution
of each team member, and $q\left(x\right)$ as the value created by
the team. This analogy directly leads to interpreting $\alpha_{i}\neq1$
as a bias player $i$ has when evaluating their contribution to the
value created by the team, which can be seen as a measure of \emph{confidence}.
We show that in any naive analytics equilibrium all agents are overconfident
in the sense of overestimating their ability to contribute to the
team's output (i.e., having $\alpha_{i}>1$). In the case of entrepreneurship,
for example, much of the past research explained overconfidence as
necessary to overcome risk aversion and tackle uncertainty, that is,
as a response to the entrepreneurial environment which is external
to the firm.\footnote{See also \citet{heller2014overconfidence} who demonstrates that overconfidence
of entrepreneurs can help an investor in diversifying aggregate risk.} Our results provide a novel foundation for the tendency of people
(and, in particular, entrepreneurs) to be overconfident in the sense
of overestimating one's ability. We show that when skills are complementary,
overconfidence contributes to increased team efficiency, and is a
response to the internal firm environment. The results provide a novel
explanation to why investors might prefer to invest in overconfident
startup founders, why managers might prefer to hire overconfident
sales people, and why academic researchers might prefer working with
overconfident co-authors (see, \citealp{herz2020opinion}, for evidence
suggesting that co-authors over-estimate their contributions to joint
papers).

\subparagraph{Underlying Game}

We describe a team production game $G_{t}$ with strategic complements.
Consider $n$ players, each choosing how much effort $x_{i}\in X_{i}\equiv\mathbb{R}_{+}$
to exert in a joint project. The project yields a benefit $q_{i}\left(x\right)$
to each agent $i$ , where $q_{i}$ is twice continuously differentiable
in $\mathbb{R}_{++}^{2}$, strictly increasing, strictly concave,
and supermodular, i.e., $\frac{dq_{i}\left(x\right)}{dx_{j}}>0$,
$\frac{d^{2}q_{i}\left(x\right)}{dx_{i}^{2}}<0$, $\frac{d^{2}q_{i}\left(x\right)}{dx_{i}x_{j}}>0$
for any $x\in X$ and any $i,j\in N$). The payoff of each player
$i$ is equal to her project's benefit minus her effort: $\pi_{i}\left(x\right)=q_{i}\left(x\right)-x_{i}.$
Observe that the above assumptions imply that the payoff function
is strictly concave and satisfies strategic complementarity. Consequently
all Nash equilibria of the underlying game are pure, and one of them,
denoted by $x^{LNE}$, has the lowest effort levels (i.e., $x_{i}^{LNE}\leq x_{i}^{NE}$
for any Nash equilibrium $x^{NE}$). Finally, we require that the
marginal contribution of efforts are initially very high, while converging
to zero for very large efforts, which implies that players always
have bounded perceived best-replies. Specifically, we assume that
the marginal contribution of a player's effort converges to infinity
(resp., zero) when the efforts converge to zero (resp., infinity)
and the profile of efforts is symmetric:
\begin{assumption}
\label{assu:marginal-contirbution-small-large-efforts}For each player
$i\in N$ 
\[
\lim_{x_{i}\rightarrow0}\left(\left.\frac{dq_{i}\left(x\right)}{dx_{i}}\right|_{x=\left(x_{i},...,x_{i}\right)}\right)=\infty,\,\,\,\lim_{x_{i}\rightarrow\infty}\left(\left.\frac{dq_{i}\left(x\right)}{dx_{i}}\right|_{x=\left(x_{i},...,x_{i}\right)}\right)=0.
\]
 
\end{assumption}

\subparagraph{Result}

Let $\Gamma_{t}=\left(G_{t},R_{++},f_{m}\right)$ be an analytics
game induced by $G_{t}$. We show that $\Gamma_{t}$ satisfies assumptions
\ref{assu:Monotone-externalities}--\ref{assu:consistent-second-adaptation},
and as a result we can apply the results of Section \ref{sec:General-Results}
and show that in any naive analytics equilibrium: (1) all agents are
overconfident in the sense of overestimating their own influence on
the joint project, (2) exert more effort than in the lowest Nash equilibrium,
and (3) the payoffs Pareto dominates the lowest Nash equilibrium.
Formally (the proof is presented in Appx. \ref{subsec:Proof-of-Proposition-team-production}):
\begin{prop}
\label{prop:team-prodcution}$\Gamma_{a}$ satisfies Assumptions \ref{assu:Monotone-externalities}--\ref{assu:consistent-second-adaptation},
which implies that any NAE $\left(\alpha^{*},x^{*}\right)$ of $\Gamma_{a}$
satisfies for each player $i$: (1) $\alpha_{i}^{*}>1$, (2) $x_{i}^{*}>x_{i}^{LNE}$,
and (3) $\pi_{i}\left(x^{*}\right)>\pi_{i}\left(x_{i}^{LNE}\right)$.
\end{prop}

\section{Implications for Market Structure Analysis\label{sec:Implication-for-Merges}}

In this Section we demonstrate the implications of the naive analytics
equilibrium concept for analysis of market structure in oligopolies.
Data on product quantities and prices are often used to infer price-cost
margins (PCM) to be used in analysis of mergers and market power \citep{nevo2000mergers,nevo2001measuring}.
In these analyses one often assumes that firms play a differentiated
Bertrand competition, as in Section \ref{subsec:Price-Competition},
and the observed price elasticities are used to infer the unobserved
marginal costs of firms, to be used in simulation of different market
structures. 

We illustrate the implications of firms playing a naive analytics
equilibrium on the results of this analysis. Assume three firms, numbered
$i=1$, $2$ and $3$, play a symmetric differentiated Bertrand competition
with demand $q_{i}=a-bx_{i}+c\frac{x_{1}+x_{2}+x_{3}}{3}$, zero marginal
costs and $a>0$, $b>c>0$. Appendix \eqref{subsec:Proof-of-Proposition-merger}
explicitly calculates the unique pre-merger NAE $\left(x_{i}^{pre},\alpha_{i}^{pre}\right)$,
which is a special case of the symmetric oligopoly in Section \ref{subsec:Price-Competition}. 

Suppose firms $2$ and $3$ plan to merge, and an economist (e.g.,
a regulator) tries to predict the prevailing prices post-merger using
pre-merger observed prices and quantities. The economist, not knowing
that the firms are playing a naive analytics equilibrium with zero
marginal costs, assumes that pre-merger each firm set its prices $x_{i}$
to maximize the payoff $\pi_{i}^{mc}=(x_{i}-mc_{i})q_{i}$, with $mc_{i}$
being firm $i$'s (unknown) marginal cost. Post-merger, the economist
assumes firm $1$ will set prices $x_{1}^{mc}$ to maximize $\pi_{1}^{mc}$
while the merged firm $23$ will set prices $x_{2}^{mc}$ and $x_{3}^{mc}$
simultaneously to maximize the joint profit $\pi_{2}^{mc}+\pi_{3}^{mc}$. 

In reality, the firms will converge to a new naive analytics equilibrium
with prices $x_{i}^{post}$ that depend on the equilibrium biases
of the firms. We denote by $\alpha^{post}$ the long-run analytics
bias, and by $x_{i}^{post}(\alpha^{post})$ the NAE prices in the
long run. In the short run, the firms might have not adjusted their
analytics yet, setting prices $x_{i}^{post}(\alpha^{pre})$ accordingly.
We also denote by $x_{i}^{post}(mc)$ the equilibrium prices predicted
by the economist who assumes firms have the same marginal costs after
the merger, and no analytics bias. The analysis shows that the economist
will estimate positive marginal costs for the firms, and that she
will underestimate the new equilibrium prices after the merger. 
\begin{prop}
\label{prop:merger}Assume that pre-merger price competition is symmetric
and that $n=3$, $a>0$ and $b>c>0$. Assume that post-merger firms
$2$ and $3$ merge and set prices simultaneously for both goods.
Then: (1) The economist will estimate $mc_{i}>0$.  (2) The economist
will underestimate post-merger prices: $x_{i}^{post}(mc)<x_{i}^{post}(\alpha^{pre})<x_{i}^{post}x_{i}^{post}(\alpha^{post})$.
\end{prop}
The proof appears in Appendix \eqref{subsec:Proof-of-Proposition-merger}.
The prices in the pre-merger naive analytics equilibrium will be above
the Nash equilibrium prices, as shown in Proposition \ref{prop:price-satisfies-Assumptions}.
These inflated prices will cause the economist to estimate a positive
marginal cost above zero for all firms. Post merger, the firms will
increase their prices due to the decreased competition. Despite the
economist overestimating firm marginal costs, which would allow her
to predict higher prices post-merger, this bias is not enough to compensate
for underestimating the prices in a naive analytics equilibrium.

Figure \ref{fig:merger} illustrates these results for $a=20$, $b=1$
and $c\in\left(0,1\right)$. The top panels show that the estimated
marginal costs $mc_{i}$ are positive and increasing with $c$ (left),
while the analytics biases $\alpha_{i}$ are lower than 1, decrease
with $c$ and $\alpha^{post}<\alpha^{pre}$ (right). The bottom panels
illustrate how the equilibrium prices change with $c$ for player
$1$ (left) and the merged player $23$ (right). In all cases, the
prices post-merger are higher than the prices pre-merger. However,
the economist underestimates the price increase due to the merger,
while the long-run price increases are the largest.

\begin{figure}[tbh]
\caption{\textbf{\label{fig:merger}}Outcomes of Merger Analysis ($a=20$,
$b=1$)}
\includegraphics[width=0.49\textwidth]{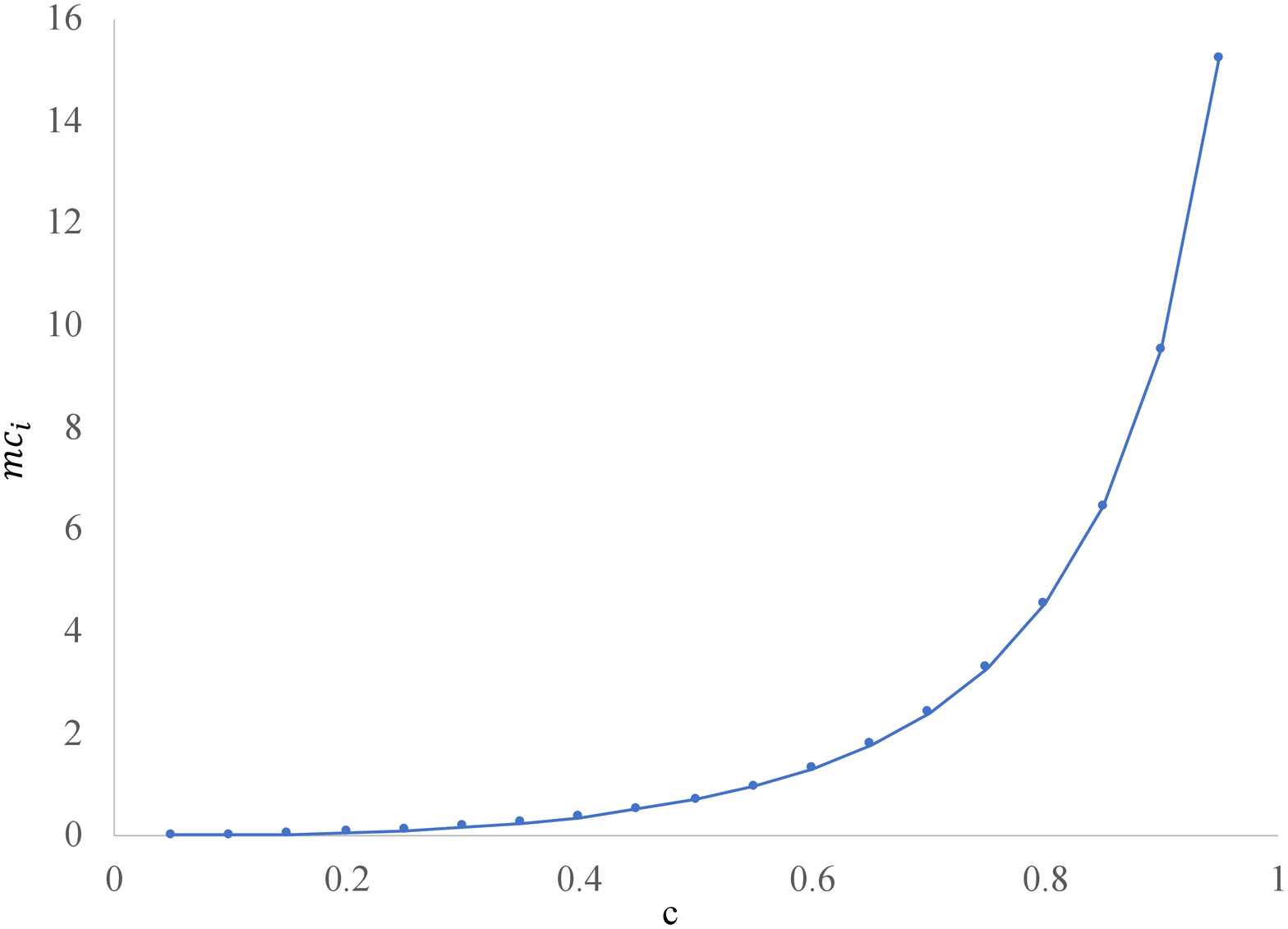}~\includegraphics[width=0.49\textwidth]{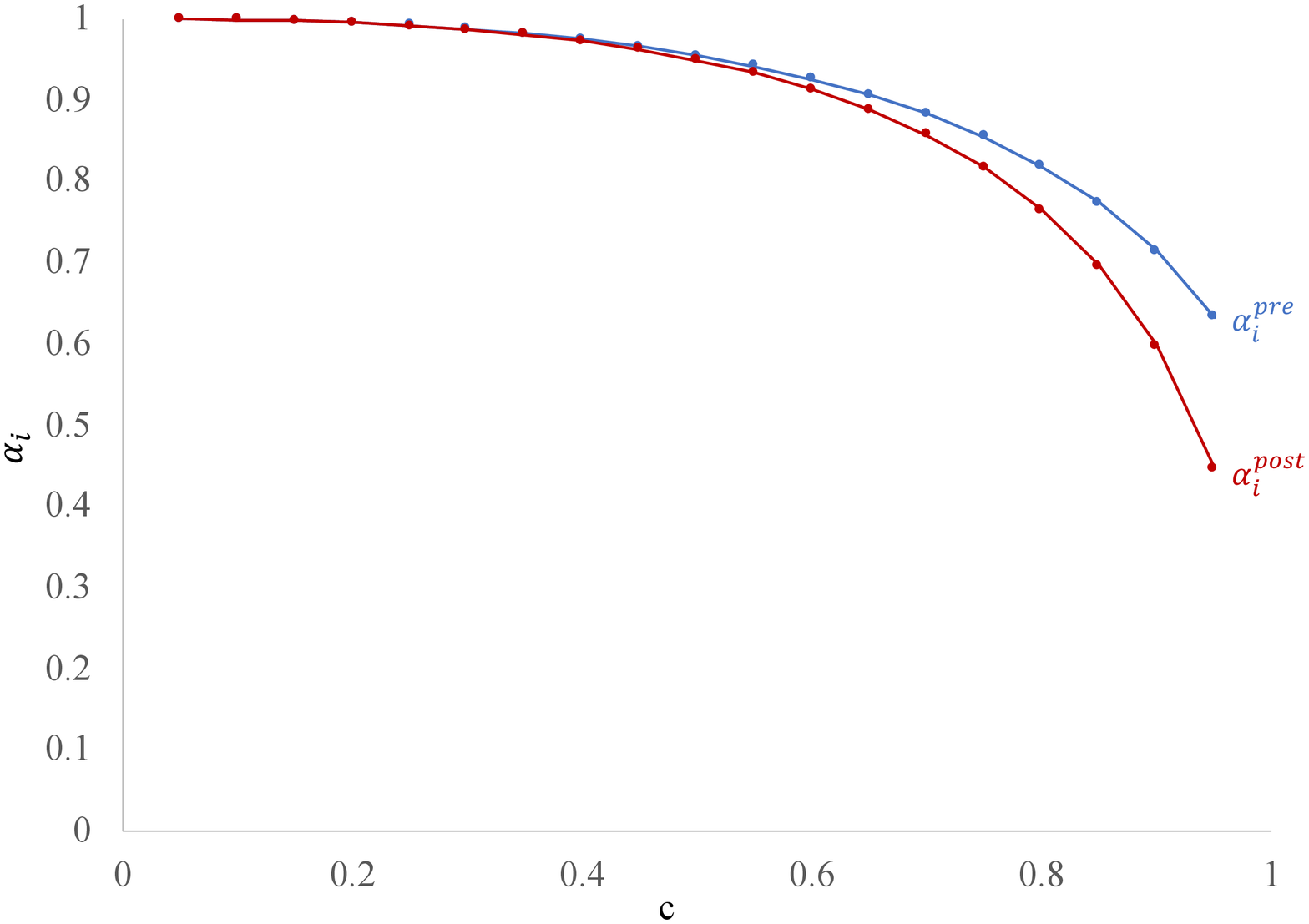}

{\footnotesize{}The left panel shows the marginal cost the economist
estimates as a function of $c$. The right panel shows the naive analytics
equilibrium biases $\alpha$ before and after the merger.}{\footnotesize\par}

\includegraphics[width=0.49\textwidth]{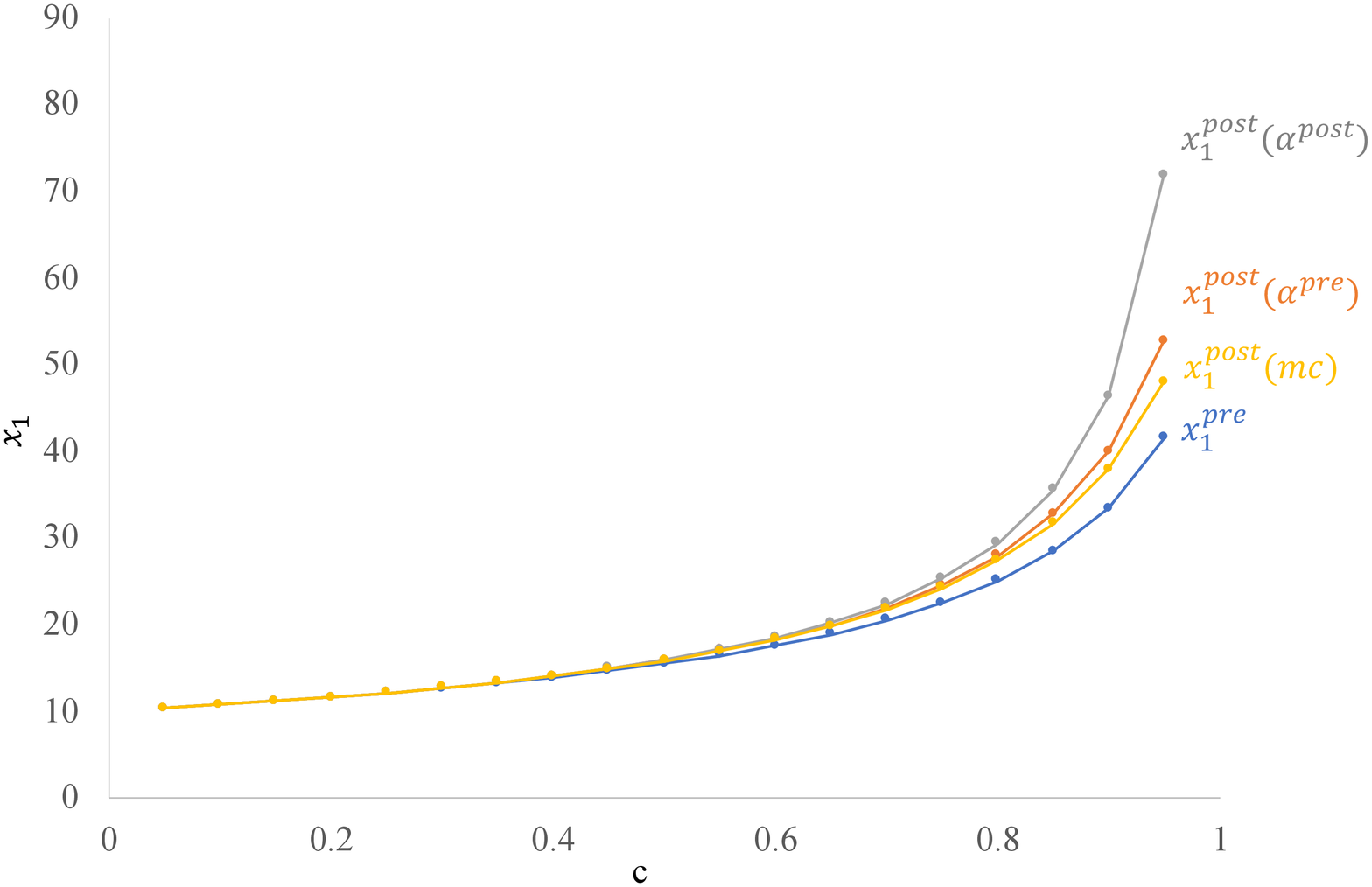}~\includegraphics[width=0.49\textwidth]{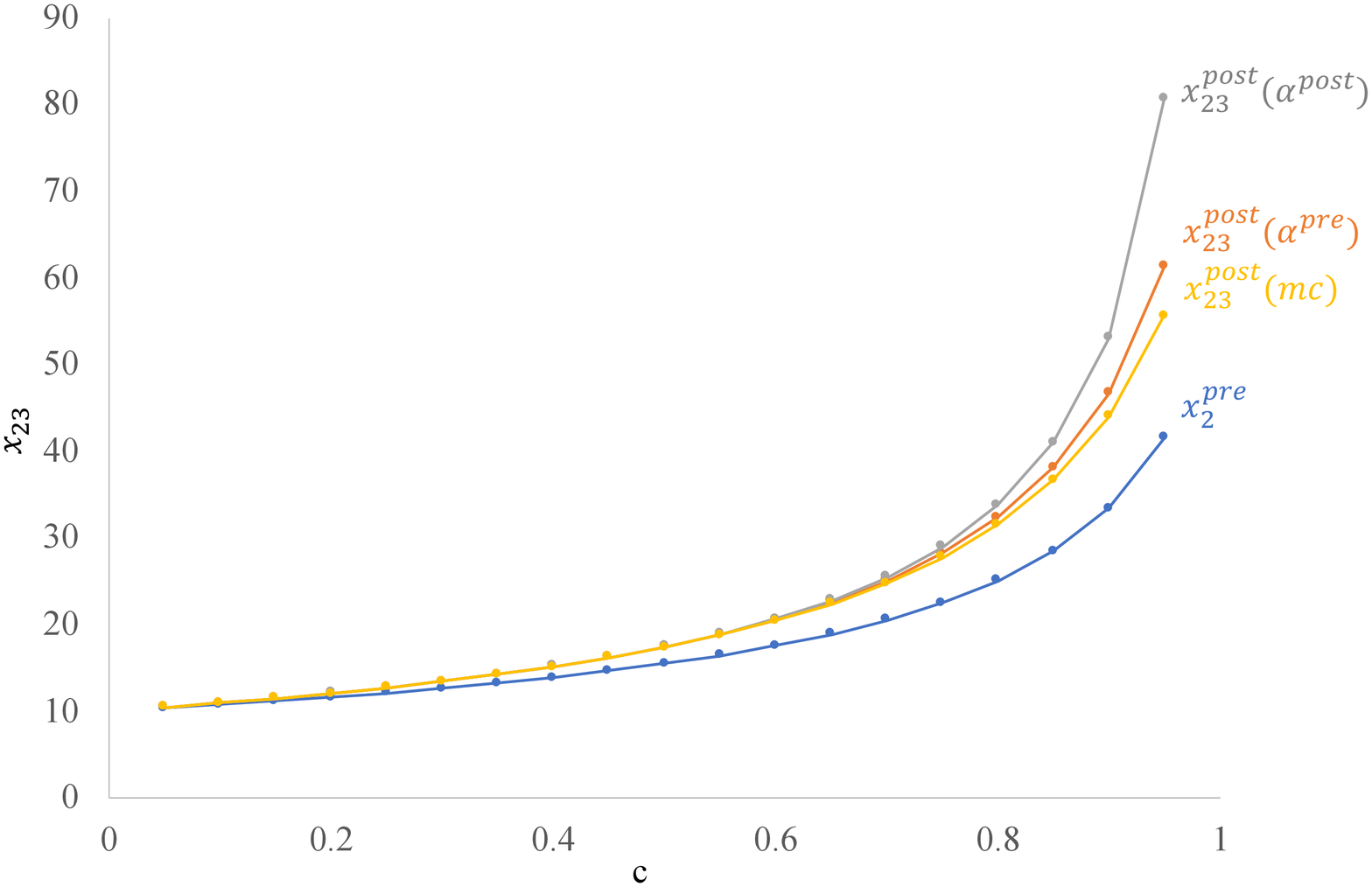} 

{\footnotesize{}The left panel shows the equilibrium prices set by
firm $1$ as a function of $c$. The right panel shows the equilibrium
prices set by firm $2$ pre-merger and the merged firm $23$ as a
function of $c$.}{\footnotesize\par}
\end{figure}

\section{Conclusion\label{sec:Conclusion}}

Naive analytics equilibrium can be used to analyze such games where
players have uncertainty about the indirect impact of their actions
on their payoffs, and allows players to use biased data analytics
to estimate this impact. This scenario is common in economic applications
such as price competition, advertising competition and team production.

The predictions of our results are consistent with commonly observed
behaviors of firms and teams. In equilibrium, players are predicted
to converge to biased estimates in the direction that causes their
opponents to respond in a beneficial manner. In pricing competition,
players are better off if they perceive consumers to be less price
elastic than they actually are, which is a possible interpretation
of observed firm pricing if they do not correct for price endogeneity
in their econometric analysis. In advertising competition, it is observed
that firms often overestimate the response to their advertising and
over-advertise, as predicted by our results. These deviations from
unbiased estimates cause deviations from the Nash equilibrium that
can be beneficial or detrimental to players. When games have strategic
complements, players will choose strategies that deviate from the
Nash equilibrium in the direction that benefits the opponents, and
their equilibrium payoffs will dominate those of the Nash equilibrium.
The converse is true for games with strategic substitutes.

The results of our analysis provide testable empirical predictions
about the direction and magnitudes of the biases. In particular, the
analysis predicts that in duopolies, both firms will have similar
level of biasedness (see, some supporting evidence in Table 2 of \citealp{villas1999endogeneity}).
Another prediction is that biasedness is strongest in duopolistic
competitions, it weakens as the number of competitors increase, and
it disappears if there are many small competitors, or in a monopolistic
market. These predictions could be tested in empirical data as well
as serve as a basis for analysis about the adoption and sophistication
of analytics in various industries. Further, our results may bring
to question some of the assumptions used in practice when performing
counterfactual analysis to estimate welfare and assess the impact
of regulatory policy. In these analyses, it is often assumed that
firms correctly perceive their economic environment and that any observed
inconsistency with this assumption may be due to unobserved factors
by researchers. However, the fact that firms misperceive the sensitivity
of demand in a naive analytics equilibrium can substantially change
the conclusions of counterfactual analysis used to assess the impact
of market structure analysis, as we have demonstrated for the merger
of competing firms in Section \ref{sec:Implication-for-Merges}. 

A second implication of our results is for research that focuses on
biases in decision making from non-causal inferential methods. The
research implicitly assumes that focusing on causality and more precise
estimates are better for firm performance, which often translates
to normative recommendation about firm practices (see, e.g., \citet{siroker2013b}
and \citet{thomke2020experimentation} on A/B testing). Our results
suggest that firms may be better off with opting for more naive heuristics,
which are indeed quite popular because they are easy to implement.
This may suggest that normative recommendations for deploying more
sophisticated analytics capabilities should be made with caution in
competitive environments.

\bibliographystyle{chicago}
\bibliography{sloppy}

\appendix
\clearpage \pagenumbering{arabic}

\section*{Online Appendices}

\section{\label{subsec:Fact-on-Relabeling}Classes with Relabeled Strategies
($x_{i}\rightarrow-x_{i}$)}

It is simple to see that relabeling strategies as their negative values
(i.e., replacing $x_{i}$ with $-x_{i}$) results in essentially the
same game with opposite signs to $\frac{d\pi_{i}}{dx_{j}}$ and $\frac{\partial\pi_{i}}{\partial x_{i}}$. 
\begin{fact}
\label{fact:relabeling}Let $\Gamma=\left(G=\left(N,X,q,\pi\right),A,f\right)$
be a naive-analytics game that satisfies Assumptions \ref{assu:Monotone-externalities}--\ref{assu:consistent-second-adaptation}.
Let $G'=\left(N,X',q',\pi'\right)$ be induced from $G$ by relabeling
the strategies $x_{i}\rightarrow-x_{i}$, i.e., $X'_{i}=-X_{i}=\left\{ x'_{i}\in\mathbb{R}|-x'_{i}\in X_{i}\right\} $,
$q'_{i}\left(x_{1},...,x_{n}\right)=q_{i}\left(-x_{1},...,-x_{n}\right)$
and $\pi'_{i}\left(x_{i},q_{i}\right)=\pi_{i}\left(-x_{i},q_{i}\right)$.
Then, $\Gamma'=\left(G',A,f\right)$ satisfies Assumptions \ref{assu:Monotone-externalities}--\ref{assu:consistent-second-adaptation},
the sign of $\frac{d\pi_{i}^{2}}{dx_{i}dx_{j}}$ remains the same,
while the signs of $\frac{d\Pi_{i}}{dx_{j}}$ and $\frac{\partial\pi_{i}}{\partial x_{i}}$
are flipped (relative to their signs in $G$), i.e., $\forall i\neq j$
\[
\frac{d\pi_{i}}{dx_{j}}>0\Leftrightarrow\frac{d\pi'_{i}}{dx_{j}}<0,\,\,\,\frac{\partial\pi_{i}}{\partial x_{i}}>0\Leftrightarrow\frac{\partial\pi'_{i}}{\partial x_{i}}<0,\,\,\,\frac{d^{2}\pi_{i}}{dx_{i}dx_{j}}>0\Leftrightarrow\frac{d^{2}\pi'_{i}}{dx_{i}dx_{j}}>0.
\]
\end{fact}
\noindent Table \ref{tab:Summary-of-Results-appx} presents the 4
classes derived from the main-text classes by relabeling the strategies.
\[
.
\]

\begin{table}[H]
\caption{\label{tab:Summary-of-Results-appx}Summary of Results Under Assumptions
\ref{assu:Monotone-externalities}--\ref{assu:consistent-second-adaptation}}

~\\
\begin{tabular}{|>{\centering}p{3.2cm}|>{\centering}p{2.35cm}|>{\centering}p{1.1cm}|>{\centering}p{1.2cm}|>{\centering}p{1.7cm}|>{\centering}p{1.7cm}|>{\centering}p{2.4cm}|}
\hline 
~\\
Applications (Sections \ref{subsec:Price-Competition}--\ref{sec:Team-Production}) & Strategic comp./sub.\\
~\\
 $\frac{d\pi_{i}^{2}}{dx_{i}dx_{j}}$ & Payoff~\\
Extr.\\
~\\
 $\frac{d\Pi_{i}}{dx_{j}}$ & Partial derivative\\
$\frac{\partial\pi_{i}}{\partial x_{i}}$ & Perceived best replying (Prop. \ref{prop:percieved-BR-vs-unbiased-BR}) & Analytics bias\\
(Prop. \ref{pro:alpha}) & Pareto Dominance (Prop. \ref{prop:Pareto-dominate})\tabularnewline
\hline 
Price competition

w. subs. goods

$x_{i}=$discount & \multirow{2}{2.35cm}{\textbf{\textcolor{green}{~}}\\
\textbf{\textcolor{green}{~~~~~~+}}\\
~\\
~~~Strategic complements} & \textbf{\textcolor{red}{~}}\\
\textbf{\textcolor{red}{-}} & \textbf{\textcolor{red}{~}}\\
\textbf{\textcolor{red}{-}} & \multirow{4}{1.7cm}{\\
~\\
~\\
\textcolor{blue}{~~}\\
\textcolor{blue}{~}\\
~$\,\,x_{i}^{*}<$\\
$BR_{i}\left(x_{-i}^{*}\right)$} & \textcolor{brown}{~}\\
\textcolor{brown}{$\alpha_{i}^{*}<1$} & ~\\
$\forall i\,x_{i}^{*}<x_{i}^{NE}$\\
NAE\tabularnewline
\cline{1-1} \cline{3-4} \cline{4-4} \cline{6-6} 
Advertising with

positive exter. $x_{i}=$budget cut &  & \textbf{~}\\
\textbf{\textcolor{red}{-}} & \textbf{\textcolor{green}{~}}\\
\textbf{\textcolor{green}{+}} &  & \textcolor{blue}{~}\\
\textcolor{blue}{$\alpha_{i}^{*}>1$}\\
 & Pareto

dominates NE.\tabularnewline
\cline{1-4} \cline{2-4} \cline{3-4} \cline{4-4} \cline{6-7} \cline{7-7} 
Price competition

w. compl. goods $x_{i}=$discount & \multirow{2}{2.35cm}{\textbf{\textcolor{green}{~}}\\
\textbf{\textcolor{green}{~~~~~~}}\textbf{\textcolor{red}{-}}\\
~\\
~~~Strategic Substitutes} & \textbf{~}\\
\textbf{\textcolor{green}{+}} & \textbf{\textcolor{red}{~}}\\
\textbf{\textcolor{red}{-}} &  & \textcolor{brown}{~}\\
\textcolor{brown}{$\alpha_{i}^{*}<1$}\\
 & ~\\
$\exists i\,x_{i}^{*}<x_{i}^{NE}$\\
Symmetric\tabularnewline
\cline{1-1} \cline{3-4} \cline{4-4} \cline{6-6} 
Advertising with

neg. externalities $x_{i}=$budget cut &  & \textbf{\textcolor{green}{~}}\\
\textbf{\textcolor{green}{+}} & \textbf{\textcolor{green}{~}}\\
\textbf{\textcolor{green}{+}} &  & \textcolor{blue}{~}\\
\textcolor{blue}{$\alpha_{i}^{*}>1$}\\
 & NE Pareto dominates sym. NAE\tabularnewline
\hline 
\end{tabular}
\end{table}

\section{\label{subsec:Example-why-Assumption-6-is-nessecary}Example why
Assumption \ref{assu:consistent-second-adaptation} is Necessary}

Example \ref{exa:example-3-player-clockwise-1} illustrates why Assumption
\ref{assu:consistent-second-adaptation} (consistent\textbf{ }secondary
adaptation) is necessary for Proposition \ref{prop:percieved-BR-vs-unbiased-BR}
(and, thus, for the remaining results of Section \ref{sec:General-Results}).
\begin{example}
\label{exa:example-3-player-clockwise-1} Consider a 3-player game
in which the firms are located in a circle, the demand of each firm
is decreasing in all prices, and it mainly depends on its own price
and on its clockwise neighbor's price. That is the demand functions
are given by (for some $0<\epsilon<<1)$: 
\[
q_{1}=120-x_{1}-x_{2}-\epsilon x_{3},\,\,\,q_{2}=120-x_{2}-x_{3}-\epsilon x_{1},\,\,\,q_{3}=120-x_{3}-x_{1}-\epsilon x_{2},
\]
and the profit function of each firm is $\pi_{i}=q_{i}x_{i}$. Observe
that the game has strategic substitutes ($\frac{d\pi_{i}^{2}}{dx_{i}dx_{j}}<0$)
and negative externalities ($\frac{d\pi_{i}}{dx_{j}}<0$), which implies
that the strategic complementarity and payoff externalities are in
the same direction. In what follows we show that $x_{i}^{*}<BR_{i}\left(x_{-i}^{*}\right)$
in any naive analytics equilibrium, contradicting Proposition \ref{prop:percieved-BR-vs-unbiased-BR}.
Assume to the contrary that $x_{1}^{*}>BR_{1}\left(x_{-1}^{*}\right)$,
which implies that Player $1$ gains from unilaterally decreasing
her price towards the unbiased best reply. This decrease initially
induces both competitors to increase their prices, but player's 3
response is much stronger. This new higher price of player 3 induces
player 2 to readjust her price by decreasing it below her initial
NAE price, which increases the payoff of player 1. This implies that
player 1 gains from the change in her opponents' strategies. Thus,
$x_{1}^{*}$ cannot be a Stackelberg-leader strategy, which contradicts
$\left(\alpha^{*},x^{*}\right)$ being a naive analytics equilibrium.
\end{example}

\section{Microfoundations for biased estimation\label{sec:Microfoundations-for-biased}}

\subsection{\label{subsec:Biased-Price-Competition}Bias $\alpha_{i}<1$ in Price
Competition }

Suppose the analysts hired by each of the two firms decide to experiment
with prices to find the price elasticity of demand by alternating
between a high price ($p_{H}$) and a low price ($p_{L}$), setting
a low price (discount) $\mu_{L}$-share of the time. The experiment
can be characterized by a level of sloppiness $\gamma_{i}\in[0,1]$.
In a fraction $\gamma_{i}$ of the time, the analyst doesn't monitor
the firm's employees and does not carefully supervise that the employees
choose the discount times uniformly at random. Hence, it is possible,
for example, that the firm's employees will implement discounts on
days of low demand, possibly due to the employees having more free
time in these days to deal with posting different prices. In the rest
of the time ($1-\gamma_{i}$ fraction), the analyst verifies that
the prices are set randomly. Consequently, when either analyst sets
prices uniformly at random, there will not be correlation between
the firm's prices. This happens $1-\gamma_{1}\gamma_{2}$ fraction
of the time. In the remaining $\gamma_{1}\gamma_{2}$ fraction of
the time, there might be correlation between the firm's prices, which
we denote by $\rho$. The joint distribution of prices conditional
on the correlation $\rho$ and the fractions $\gamma_{1},\gamma_{2}$
is described in Table \ref{tb:joint}.

\begin{table}[h]
\begin{centering}
\begin{tabular}{|c|c|c|}
\hline 
 & $p_{L}$ & $p_{H}$\tabularnewline
\hline 
\hline 
$p_{L}$ & $\mu_{LL}=\mu_{L}^{2}+\mu_{L}(1-\mu_{L})\gamma_{1}\gamma_{2}\rho$ & $\mu_{LH}=\mu_{L}(1-\mu_{L})(1-\gamma_{1}\gamma_{2}\rho)$\tabularnewline
\hline 
$p_{H}$ & $\mu_{HL}=\mu_{L}(1-\mu_{L})(1-\gamma_{1}\gamma_{2}\rho)$ & $\mu_{HH}=(1-\mu_{L})^{2}+\mu_{L}(1-\mu_{L})\gamma_{1}\gamma_{2}\rho$\tabularnewline
\hline 
\end{tabular}
\par\end{centering}
\caption{Joint distribution of prices with correlation $\rho$ and the fractions
$\gamma_{1},\gamma_{2}$}
\label{tb:joint}
\end{table}
When calculating the price elasticity of demand to decide how to change
prices, the analyst calculates:
\begin{equation}
\eta_{i}=-\frac{\frac{\Delta Q_{i}}{\overline{Q_{i}}}}{\frac{\Delta P_{i}}{\overline{P_{i}}}}\label{eq:elasticity}
\end{equation}
Where $\Delta Q_{i}$ is the difference in average demand between
high priced and low priced periods, $\overline{Q_{i}}$ is the average
realized demand, $\Delta P_{i}=p_{H}-p_{L}$ is the difference in
price between high and low price periods, and $\overline{P_{i}}=\mu_{L}p_{L}+(1-\mu_{L})p_{H}$
is the average price set by the firm.

The demand observed by firm $i$ when setting price $p_{i}$ and when
its competitor sets a price $p_{-i}$ is $Q_{i}(p_{i},p_{-i})=a_{i}-b_{i}p_{i}+c_{i}p_{-i}$. 

Using the joint probabilities in Table \ref{tb:joint}, we find that
$\overline{Q_{i}}=a-(b-c)\left(\mu_{L}p_{L}+\left(1-\mu_{L}\right)p_{H}\right)$
and $\Delta Q_{i}=-\left(p_{H}-p_{L}\right)\left(b-c\gamma_{1}\gamma_{2}\rho\right)$.

Plugging into \eqref{eq:elasticity}, firm $i$ will estimate its
price elasticity as:

\begin{equation}
\eta_{i}=\frac{\left(b-c\gamma_{1}\gamma_{2}\rho\right)\left(\mu_{L}p_{L}+p_{H}\left(1-\mu_{L}\right)\right)}{a-(b-c)\left(\mu_{L}p_{L}+p_{H}\left(1-\mu_{L}\right)\right)},
\end{equation}
while the true elasticity is $\eta_{i}^{T}=\frac{b\left(\mu_{L}p_{L}+p_{H}\left(1-\mu_{L}\right)\right)}{a-(b+c)\left(\mu_{L}p_{L}+p_{H}\left(1-\mu_{L}\right)\right)}$.
Hence the analyst will estimate the firm's price elasticity as being
lower than $\eta_{i}^{T}$ when $c>0$ and $\rho>0$.

\subsection{\label{subsec:Biased-Advertising-Effectiveness}Bias $\alpha_{i}>1$
in Advertising Competition}

Assume the firm's sales at time $t$ behave according to the linear
model $sales_{t}=\mu+x_{t}+\epsilon_{t}$ where $\mu$ is the average
sales, $x_{t}$ is the level of advertising, that can be $x_{L}$
or $x_{H}$ with $x_{H}>x_{L}\ge0$, and $\epsilon_{t}$ is demand
shock which is distributed i.i.d $\mathscr{\mathcal{N}}(0,1)$. In
this model, the true effect of advertising, $\frac{d(sales_{t})}{dx_{t}}$
equals 1.

The firm has a sales target $\mu$ and its advertising strategy is
to increase advertising to $x_{t+1}=x_{H}$ if sales fall below $\mu$
at time $t$, i.e., if $sales_{t}<\mu$, and otherwise set $x_{t+1}=x_{L}$.

To estimate the effect of advertising, the firm looks at the difference
in sales when advertising is increased or decreased (otherwise the
change cannot be attributed to changes in advertising) and takes the
average to calculate
\begin{equation}
\frac{\mathbb{E}[\Delta sales]}{\Delta x}=\frac{\frac{\mathbb{E}[sales_{t+1}-sales_{t}|x_{t+1}=x_{H},x_{t}=x_{L}]}{x_{H}-x_{L}}+\frac{\mathbb{E}[sales_{t+1}-sales_{t}|x_{t+1}=x_{L},x_{t}=x_{H}]}{x_{L}-x_{H}}}{2}\label{eq:ad-effect}
\end{equation}

More sophisticated approaches can take into account a weighted average
of these estimates and also take into account the baseline sales when
advertising does not change.

The left-hand part of the summand equals: 
\begin{eqnarray*}
\frac{\mathbb{E}[sales_{t+1}-sales_{t}|x_{t+1}=x_{H},x_{t}=x_{L}]}{x_{H}-x_{L}} & = & \frac{\mu+x_{H}+\mathbb{E}[\epsilon_{t+1}]-(\mu+x_{L}+\mathbb{E}[\epsilon_{t}|sales_{t}<\mu])}{x_{H}-x_{L}}\\
 & = & =\frac{\mu+x_{H}-\left(\mu+x_{L}-\frac{\phi(-x_{L})}{\Phi(-x_{L})}\right)}{x_{H}-x_{L}}=1+\frac{\frac{\phi(-x_{L})}{\Phi(-x_{L})}}{x_{H}-x_{L}}>1
\end{eqnarray*}

where $\phi(\cdot)$ is the standard Normal pdf and $\Phi(\cdot)$
its cdf. The right-hand part equals:

\begin{eqnarray*}
\frac{\mathbb{E}[sales_{t+1}-sales_{t}|x_{t+1}=x_{L},x_{t}=x_{H}]}{x_{L}-x_{H}} & = & \frac{\mu+x_{L}+\mathbb{E}[\epsilon_{t+1}]-(\mu+x_{H}+\mathbb{E}[\epsilon_{t}|sales_{t}\ge\mu])}{x_{L}-x_{H}}\\
 & = & \frac{\mu+x_{L}-\left(\mu+x_{H}-\frac{\phi(-x_{H})}{1-\Phi(-x_{H})}\right)}{x_{L}-x_{H}}=1-\frac{\frac{\phi(x_{H})}{\Phi(x_{H})}}{x_{H}-x_{L}}<1
\end{eqnarray*}

Because $\frac{\phi(x)}{\Phi(x)}$ is decreasing in $x$, the sum
in the numerator of \eqref{eq:ad-effect} is larger than 2, which
results in the firm overestimating the effectiveness of its advertising
to be more than 1.

\section{\label{sec:Technical-Proofs}Technical Proofs}

\subsection{\label{subsec:Proof-of-Stackelberg}Proof of Prop. \ref{pro:stackelberg-leader}
(Unilateral Sophistication)}

Assume to the contrary that $\left(\alpha^{*},x^{*}\right)$ is
a naive analytics equilibrium and $x_{i}^{*}\notin X_{i}^{SL}\left(\alpha_{-i}^{*}\right)$.
The fact that $x_{i}^{*}\notin X_{i}^{SL}\left(\alpha_{-i}^{*}\right)$
implies that there exist $x'_{i}\in X_{i}$ and $\left(x'_{i},\alpha_{-i}^{*}\right)$-equilibrium
$x'$ such that $\pi_{i}\left(x'\right)>\pi_{i}\left(x^{*}\right)$.
Let $\alpha'_{i}\in A$ be a bias satisfying $\frac{d\pi_{i}^{\alpha'_{i}}\left(x'\right)}{dx_{i}}=0$
and $\frac{d^{2}\pi_{i}^{\alpha'_{i}}\left(x'\right)}{dx_{i}^{2}}<0$
(such $\alpha'{}_{i}$ exists due to part (2) of Definition \ref{def:SL-strategy}).
Then $x'$ is $\left(\alpha'_{i},\alpha_{-i}^{*}\right)$-equilibrium
that satisfies $\pi_{i}\left(x'\right)>\pi_{i}\left(x^{*}\right)$,
which contradicts $\left(\alpha^{*},x^{*}\right)$ being a naive analytics
equilibrium.

Next assume to the contrary that $x^{*}$ is an $\alpha^{*}$-equilibrium,
$x_{i}^{*}\in X_{i}^{SL}\left(\alpha_{-i}^{*}\right)$ for each player
$i$ , and $\left(\alpha^{*},x^{*}\right)$ is not a naive analytics
equilibrium. The last assumption implies that there is $i\in N$,
bias $\alpha'_{i}$ and $\left(\alpha'_{i},\alpha_{-i}^{*}\right)$-equilibrium
$x'$ s.t. $\pi_{i}\left(x'\right)>\pi_{i}\left(x^{*}\right)$. Observe
that $x'$ is an $\left(x'_{i},\alpha_{-i}^{*}\right)$-equilibrium.
The fact that $\pi_{i}\left(x'\right)>\pi_{i}\left(x^{*}\right)$
contradicts the assumption that $x_{i}\left(\alpha^{*}\right)\in X_{i}^{SL}\left(\alpha_{-i}\right)$.

\subsection{Proposition \ref{prop:Pareto-dominate} (Strategic Complementarity)}

\begin{lem}
\label{lem:strategic-complements-standard-lemma-1}Let $\Gamma$ be
a game Assumptions \ref{assu:Monotone-externalities}--\ref{assu:Strong-com-subs}.
Let $x^{*}$ be a strategy profile satisfying $\textrm{sgn}\left(x_{i}^{*}-BR_{i}\left(x_{-i}^{*}\right)\right)\in\left\{ \textrm{sgn}\left(\frac{d\pi_{i}}{dx_{j}}\right),0\right\} $
for each $i\in N$ and $\textrm{sgn}\left(x_{k}^{*}-BR_{i}\left(x_{-k}^{*}\right)\right)\neq0$
for some $k\in N$. hen there exists a Nash equilibrium $x^{NE}$,
such that $sgn\left(x_{i}^{*}-x_{i}^{NE}\right)=sgn\left(\frac{d\pi_{j}}{dx_{i}}\right)$
for each $i\in N$.
\end{lem}
\begin{proof}
We begin by showing a slightly weaker property, namely, that there
exists a Nash equilibrium $x^{NE}$, such that $sgn\left(x_{i}^{*}-x_{i}^{NE}\right)=\left\{ sgn\left(\frac{d\pi_{j}}{dx_{i}}\right),0\right\} $
for each $i\in N$. Assume to the contrary that there exists player
$j$ for which $sgn\left(x_{j}^{*}-x_{j}^{NE}\right)=-sgn\left(\frac{d\pi_{j}}{dx_{i}}\right)$
for every Nash equilibrium $x^{NE}$. Consider an auxiliary game $G^{R}$
similar to $G$ except that each player $i$ is restricted to choose
a strategy $x_{i}$ satisfying $sgn\left(x_{i}^{*}-x_{i}\right)\in\left\{ \textrm{sgn}\left(\frac{d\pi_{i}}{dx_{j}}\right),0\right\} $.
Due to the concavity (Assumption \ref{assu:Strong-concavity}), the
game $G^{R}$ admits a pure Nash equilibrium, which we denote by $x^{RE}$.
The profile $x^{RE}$ cannot be a Nash equilibrium of the original
game $G$ because $sgn\left(x_{j}^{*}-x_{j}^{RE}\right)\in\left\{ \textrm{sgn}\left(\frac{d\pi_{i}}{dx_{j}}\right),0\right\} $,
while $sgn\left(x_{j}^{*}-x_{j}^{NE}\right)=-sgn\left(\frac{d\pi_{j}}{dx_{i}}\right)$
for every Nash equilibrium $x^{NE}$. This implies that there exists
player $i$ for which $x_{i}^{RE}=x_{i}^{*}$ and $\textrm{sgn}\left(x_{i}^{*}-BR_{i}\left(x_{-i}^{RE}\right)\right)=-\textrm{sgn}\left(\frac{d\pi_{i}}{dx_{j}}\right)$,
which contradicts $\textrm{sgn}\left(x_{i}^{*}-BR_{i}\left(x_{-i}^{*}\right)\right),sgn\left(BR_{i}\left(x_{-i}^{*}\right)-BR_{i}\left(x_{-i}^{RE}\right)\right)\in\left\{ \textrm{sgn}\left(\frac{d\pi_{i}}{dx_{j}}\right),0\right\} $
(where the latter inclusion is implied by the strategic complimentary
and the fact that $sgn\left(x_{-i}^{*}-x_{-i}^{RE}\right)\in\left\{ \textrm{sgn}\left(\frac{d\pi_{i}}{dx_{j}}\right),0\right\} $).

Next, we prove the slightly stronger property that $sgn\left(x_{i}^{*}-x_{i}^{NE}\right)=sgn\left(\frac{d\pi_{j}}{dx_{i}}\right)$
for each $i\in N$. Assume to the contrary that there exists player
$j$ for which $sgn\left(x_{j}^{*}-x_{j}^{NE}\right)=\left\{ -\textrm{sgn}\left(\frac{d\pi_{i}}{dx_{j}}\right),0\right\} $
for every Nash equilibrium $x^{NE}$. Due to argument presented above
$sgn\left(x_{i}^{*}-x_{i}^{NE}\right)=\left\{ sgn\left(\frac{d\pi_{j}}{dx_{i}}\right),0\right\} $
for each $i\in N$. This implies that $x_{j}^{*}=x_{j}^{NE}$ and
\[
sgn\left(BR_{j}\left(x_{-j}^{*}\right)-BR_{j}\left(x_{-j}^{NE}\right)\right)=sgn\left(BR_{j}\left(x_{-j}^{*}\right)-x_{j}^{NE}\right)=\left\{ sgn\left(\frac{d\pi_{j}}{dx_{i}}\right),0\right\} 
\]
due to the strategic complements. Finally, the fact that $sgn\left(x_{j}^{*}-BR_{i}\left(x_{-j}^{*}\right)\right)=sgn\left(\frac{d\pi_{j}}{dx_{i}}\right)$
implies that $sgn\left(x_{j}^{*}-x_{j}^{NE}\right)=sgn\left(\frac{d\pi_{j}}{dx_{i}}\right)$
and we get a contradiction.
\end{proof}
\begin{lem}
\label{lem-sign-depends-only-on-alpha_j}Let $\Gamma$ be a game satisfying
Assumptions \ref{assu:Monotone-externalities}--\ref{assu:consistent-second-adaptation}.
Let $x$ be an $\alpha$-equilibrium and let $x'$ be an $\left(\alpha_{j},\alpha'_{-j}\right)$-equilibrium.
Then $\textrm{sgn}\left(x_{j}-BR_{j}\left(x_{-j}\right)\right)=\textrm{sgn}\left(x'_{j}-BR_{j}\left(x'_{-j}\right)\right)$.
\end{lem}
\begin{proof}
The fact that x (resp., $x'$) is an $\alpha$-equilibrium (resp.,
$\left(\alpha_{j},\alpha'_{-j}\right)$-equilibrium) implies that
$\frac{d\pi_{j}^{\alpha_{j}}\left(x\right)}{dx_{j}}=0$ (resp., $\frac{d\pi_{J}^{\alpha_{j}}\left(x'\right)}{dx_{j}}=0$).
This, in turn, implies due to the definition of $\frac{dq_{j}^{\alpha_{j}}}{dx_{j}}$
and Assumptions \ref{assu:Monotone-partial-deriv}--\ref{assu:Strong-concavity}
that $\alpha_{j}\cdot\frac{d\pi_{j}\left(x\right)}{dx_{j}}\cdot\frac{\partial\pi_{j}\left(x\right)}{\partial x_{j}}>0$
(resp., $\alpha_{j}\cdot\frac{d\pi_{j}\left(x'\right)}{dx_{j}}\cdot\frac{\partial\pi_{j}\left(x'\right)}{\partial x_{j}}>0$).
Combining these inequalities imply that $\textrm{sgn}\left(\frac{d\pi_{j}\left(x\right)}{dx_{j}}\right)=\textrm{sgn}\left(\frac{d\pi_{j}\left(x'\right)}{dx_{j}}\right)$
due to Assumption \ref{assu:Monotone-partial-deriv}, which implies
that $\textrm{sgn}\left(x_{j}-BR_{j}\left(x_{-j}\right)\right)=\textrm{sgn}\left(x'_{j}-BR_{j}\left(x'_{-j}\right)\right)$
due to Assumption \ref{assu:Strong-concavity}.
\end{proof}
We now prove part (2) of Prop. \ref{prop:Pareto-dominate} for the
case of strategic substitutes (part (1) is proven in the main text
and the preceding lemmas). Corollary \ref{cor:sign-strategic-compl-supp}
implies that $\textrm{sgn}\left(x_{i}^{*}-BR_{i}\left(x_{-i}^{*}\right)\right)=-\textrm{sgn}\left(\frac{d\pi_{i}}{dx_{j}}\right)$
for each $i\in N$. We now show that that for every Nash equilibrium
$x^{NE}$ there exists player $i\in N$, such that $sgn\left(x_{i}^{*}-x_{i}^{NE}\right)=-sgn\left(\frac{d\pi_{j}}{dx_{i}}\right)$.
Fix any player $j\in N$. Assume to the contrary that $sgn\left(x_{i}^{*}-x_{i}^{NE}\right)=sgn\left(\frac{d\pi_{j}}{dx_{i}}\right)$
for each player $i$. This implies (due to robust strategic substitutes)
that $\textrm{sgn}\left(BR_{i}\left(x_{-j}^{*}\right)-BR_{j}\left(x_{-j}^{NE}\right)\right)=-\textrm{sgn}\left(\frac{d\pi_{i}}{dx_{j}}\right)$,
which, in turn, implies that $-\textrm{sgn}\left(\frac{d\pi_{i}}{dx_{j}}\right)=\textrm{sgn}\left(x_{j}^{*}-BR_{j}\left(x_{-j}^{NE}\right)\right)=\textrm{sgn}\left(x_{j}^{*}-x_{j}^{NE}\right)$,
and we get a contradiction. Next assume that $x^{NE}$ and $x^{*}$
are both symmetric profiles. The symmetry of the profiles and the
above argument implies that $sgn\left(x_{i}^{*}-x_{i}^{NE}\right)=-sgn\left(\frac{d\pi_{j}}{dx_{i}}\right)$
for each player $i\in N$. This implies that $\pi_{i}\left(x^{*}\right)<\pi_{i}\left(x_{i}^{*},x_{-i}^{NE}\right)\leq\pi_{i}\left(x^{NE}\right)$,
where the first (resp., second) inequality is due to monotone externalities
(resp., $x^{NE}$ being a Nash equilibrium).

\subsection{\label{subsec:Proof-of-oligopoly}{\normalsize{}Proof of Proposition
\ref{prop:price-satisfies-Assumptions}} {\normalsize{}(Price Competition
$\Rightarrow$ Assumptions \ref{assu:Monotone-externalities}--\ref{assu:consistent-second-adaptation})}}

The following three lemmas will be helpful for the proof Proposition
\ref{prop:price-satisfies-Assumptions}.
\begin{lem}
\label{lem:unqiue-alpha-eq-Price-competition-game}Price competition
$\Gamma_{P}$ admits a unique $\alpha$-equilibrium for any $\alpha\in\mathbb{R}_{++}$.
\end{lem}
\begin{proof}
The robust concavity of the payoff function (proved below) implies
that any $\alpha$-equilibrium is fully characterized by the FOC

\begin{equation}
0=\frac{d\pi_{i}^{\alpha_{i}}}{dx_{i}}=q_{i}+\alpha_{i}\left(b_{i}-c_{i}w_{i}\right)x_{i}=a_{i}-\left(b_{i}+\alpha_{i}\left(b_{i}-c_{i}w_{i}\right)\right)x_{i}+c_{i}\bar{x}\Rightarrow x_{i}=\frac{a_{i}+c_{i}\bar{x}}{b_{i}+\alpha_{i}\left(b_{i}-c_{i}w_{i}\right)},\label{eq:x_i_NE-1}
\end{equation}
Multiplying each $i$-th equation by $w_{i}$ and summing up the $n$
equations yields 
\begin{equation}
\bar{x}=\sum_{i\in N}\frac{w_{i}\left(a_{i}+c_{i}\bar{x}\right)}{b_{i}+\alpha_{i}\left(b_{i}-c_{i}w_{i}\right)},\label{eq:x-mean-equation-alpha}
\end{equation}
which is a linear one-variable equation that yields a unique solution
$\bar{x}^{\alpha},$which induces a unique $\alpha$-equilibrium $x^{\alpha}$
by substituting $\bar{x}=x^{\alpha}$ in (\ref{eq:x_i_NE-1}). When
substituting $\alpha=\overrightarrow{1}$, this implies, the existence
of a unique Nash equilibrium, which we denote by $x^{NE}$.
\end{proof}
\begin{lem}
\label{lem-bounded-BR-pos-B_i}Assume that $BR_{i}^{\alpha_{i}}\left(x_{-i}\right)=A_{i}^{\alpha_{i}}+B_{i}^{\alpha_{i}}\sum_{j\neq i}w_{j}x_{j}$,
$A_{i}^{\alpha_{i}}>0$, $B_{i}^{\alpha_{i}}\in\left(0,1\right)$
for each player $i\in N$ and each $\alpha_{i}\in A$. Let $\epsilon>0$,
$\underline{x}_{i}=A_{i}^{\alpha_{i}}$, and $\overline{x}_{i}=\frac{\max_{i\in N}A_{i}^{\alpha_{i}}+\epsilon}{1-\max_{i\in N}B_{i}^{\alpha_{i}}}$.
Then $BR_{i}^{\alpha}\left(\underline{x}_{-i}^{\alpha}\right),BR_{i}^{\alpha_{i}}\left(\overline{x}_{-i}^{\alpha}\right)\in\left(\underline{x}_{i}^{\alpha},\overline{x}_{i}^{\alpha}\right)$
for each $i\in N$.
\end{lem}
\begin{proof}
It is immediate that $BR_{i}^{\alpha_{i}}\left(\underline{x}_{-i}\right),BR_{i}^{\alpha_{i}}\left(\overline{x}_{-i}\right)>A_{i}^{\alpha_{i}}=\underline{x}_{i}^{\alpha}.$
Next observe that 
\[
BR_{i}^{\alpha_{i}}\left(\underline{x}_{-i}\right),BR_{i}^{\alpha_{i}}\left(\overline{x}_{-i}\right)\leq A_{i}^{\alpha_{i}}+B_{i}^{\alpha_{i}}\left(\frac{\max_{i\in N}A_{i}^{\alpha_{i}}+\epsilon}{1-\max_{i\in N}B_{i}^{\alpha_{i}}}\right)\leq\max_{i\in N}A_{i}^{\alpha_{i}}+B_{i}^{\alpha_{i}}\left(\frac{\max_{i\in N}A_{i}^{\alpha_{i}}+\epsilon}{1-\max_{i\in N}B_{i}^{\alpha_{i}}}\right)
\]
\[
\frac{\max_{i\in N}A_{i}^{\alpha_{i}}\left(1-\max_{i\in N}B_{i}^{\alpha_{i}}\right)+B_{i}^{\alpha_{i}}\max_{i\in N}A_{i}^{\alpha_{i}}+B_{i}^{\alpha_{i}}\epsilon}{1-\max_{i\in N}B_{i}^{\alpha_{i}}}<\frac{\max_{i\in N}A_{i}^{\alpha_{i}}+\epsilon}{1-\max_{i\in N}B_{i}^{\alpha_{i}}}=\overline{x}_{i}.\qedhere
\]
\end{proof}
\begin{lem}
\label{lem-bounded-BR-neg-B_i}Assume that $BR_{i}^{\alpha_{i}}\left(x_{-i}\right)=A_{i}^{\alpha_{i}}-B_{i}^{\alpha_{i}}\sum_{j\neq i}w_{j}x_{j}$,
$A_{i}^{\alpha_{i}}>0$, $B_{i}^{\alpha_{i}}\in\left(0,1\right)$
for each player $i\in N$ and each $\alpha_{i}\in A$. Let $\underline{x}_{i}^{\alpha_{i}}=\min_{i\in N}\left(A_{i}^{\alpha_{i}}\left(1-B_{i}^{\alpha_{i}}\right)\right)$
and $\overline{x}_{i}^{\alpha_{i}}=A_{i}^{\alpha_{i}}$. Then $BR_{i}^{\alpha}\left(\underline{x}_{-i}^{\alpha}\right),BR_{i}^{\alpha_{i}}\left(\overline{x}_{-i}^{\alpha}\right)\in\left(\underline{x}_{i}^{\alpha},\overline{x}_{i}^{\alpha}\right)$
for each $i\in N$.
\end{lem}
\begin{proof}
It is immediate that $BR_{i}^{\alpha_{i}}\left(\underline{x}_{-i}\right),BR_{i}^{\alpha_{i}}\left(\overline{x}_{-i}\right)<A_{i}^{\alpha_{i}}=\overline{x}_{i}^{\alpha_{i}}.$
Next observe that 
\[
BR_{i}^{\alpha_{i}}\left(\underline{x}_{-i}\right),BR_{i}^{\alpha_{i}}\left(\overline{x}_{-i}\right)\geq A_{i}^{\alpha_{i}}-B_{i}^{\alpha_{i}}A_{i}^{\alpha_{i}}=\left(1-B_{i}^{\alpha_{i}}\right)A_{i}^{\alpha_{i}}\geq\min_{i\in N}\left(A_{i}^{\alpha_{i}}\left(1-B_{i}^{\alpha_{i}}\right)\right)=\underline{x}_{i}^{\alpha_{i}}.
\]
\end{proof}
We begin by showing that $\Gamma_{p}$ satisfies Assumptions \ref{assu:Strong-concavity}--\ref{assu:consistent-second-adaptation}.
\begin{enumerate}
\item Assumption \ref{assu:Monotone-externalities}: $\frac{d\pi_{i}}{dx_{j}}=\frac{d\pi_{i}}{dq_{i}}\frac{dq_{i}}{dx_{j}}=x_{i}\frac{dq_{i}}{dx_{j}}=x_{i}w_{j}c_{i}\,\Rightarrow$
$\textrm{sgn}\left(\frac{d\pi_{i}}{dx_{j}}\right)=\textrm{sgn}\left(c_{i}\right).$
\item Assumption \ref{assu:Monotone-partial-deriv}: (1) $\frac{\partial\pi_{i}}{\partial x_{i}}=q_{i}>0$,
and (2) $\frac{\partial\pi_{i}}{\partial q_{i}}\frac{\partial q_{i}}{\partial x_{i}}=-x_{i}\cdot\left(b_{i}-w_{i}c_{i}\right)<0$.
\item Assumption \ref{assu:Strong-concavity-prime}:  (which implies Assumption
\ref{assu:Strong-concavity} of robust concavity): (1) $\frac{d}{dx_{i}}\left(\frac{\partial\pi_{i}}{\partial x_{i}}\right)=\frac{dq_{i}}{dx_{i}}=-\left(b_{i}-w_{i}c_{i}\right)<0$,
and (2) $\frac{d}{dx_{i}}\left(\frac{\partial\pi_{i}}{\partial q_{i}}\cdot\frac{\partial q_{i}}{\partial x_{i}}\right)=\frac{d}{dx_{i}}\left(-x_{i}\left(b_{i}-w_{i}c_{i}\right)\right)=-\left(b_{i}-w_{i}c_{i}\right)<0$.
\item Assumption \ref{assu:Strong-com-subs-prime}:  (1) $\frac{d}{dx_{j}}\left(\frac{\partial\pi_{i}}{\partial x_{i}}\right)=\frac{d}{dx_{j}}\left(q_{i}\right)=w_{j}c_{i}$,
(2) $\frac{d}{dx_{j}}\left(\frac{\partial\pi_{i}}{\partial q_{i}}\frac{\partial q_{i}}{\partial x_{i}}\right)=\frac{d}{dx_{j}}\left(x_{i}\cdot\left(w_{i}c_{i}-b\right)\right)=0$.
\item Assumption \ref{assu:bounded-percieved-BR} (bounded perceived best
replies): For each player $i\in N$ and $\alpha_{i}>0$:
\[
0=\frac{d\pi_{i}^{\alpha}\left(x\right)}{dx_{i}}\Leftrightarrow0=\frac{\partial\pi_{i}}{\partial x_{i}}+\alpha_{i}\frac{\partial\pi_{i}}{\partial q_{i}}\cdot\frac{\partial q_{i}}{\partial x_{i}}=q_{i}-\alpha_{i}x_{i}\left(b_{i}-w_{i}c_{i}\right)=a_{i}-b_{i}x_{i}+c_{i}\bar{x}-\alpha_{i}x_{i}\left(b_{i}-w_{i}c_{i}\right)
\]
\begin{equation}
\Leftrightarrow\left(1+\alpha_{i}\right)\left(b_{i}-w_{i}c_{i}\right)x_{i}=a_{i}+c_{i}\sum_{j\neq i}w_{j}x_{j}\Leftrightarrow BR_{i}^{\alpha_{i}}\left(x_{-i}\right)=\frac{a_{i}+c_{i}\sum_{j\neq i}w_{j}x_{j}}{\left(1+\alpha_{i}\right)\left(b_{i}-w_{i}c_{i}\right)}.\label{eq:biased-BR-oligopoly}
\end{equation}
If $c_{i}>0$, then let $\epsilon>0$,
\[
\underline{x}_{i}^{\alpha}=\frac{a_{i}}{\left(1+\alpha_{i}\right)\left(b_{i}-w_{i}c_{i}\right)},\,\,\,\overline{x}_{i}^{\alpha}=\frac{\max_{i\in N}\left(\frac{a_{i}}{\left(1+\alpha_{i}\right)\left(b_{i}-w_{i}c_{i}\right)}\right)+\epsilon}{1-\max_{i\in N}\left(\frac{c_{i}}{\left(1+\alpha_{i}\right)\left(b_{i}-w_{i}c_{i}\right)}\right)},
\]
and Lemma \ref{lem-bounded-BR-pos-B_i} implies that Assumption \ref{assu:bounded-percieved-BR}
is satisfied. If $c_{i}<0$, then let 
\[
\underline{x}_{i}^{\alpha}=\min_{i\in N}\left(\frac{a_{i}}{\left(1+\alpha_{i}\right)\left(b_{i}-w_{i}c_{i}\right)}\left(1-\frac{\left|c_{i}\right|}{\left(1+\alpha_{i}\right)\left(b_{i}-w_{i}c_{i}\right)}\right)\right),\,\,\,\overline{x}_{i}^{\alpha}=\frac{a_{i}}{\left(1+\alpha_{i}\right)\left(b_{i}-w_{i}c_{i}\right)},
\]
and Lemma \ref{lem-bounded-BR-neg-B_i} implies that Assumption \ref{assu:bounded-percieved-BR}
is satisfied.
\item Assumption \ref{assu:consistent-second-adaptation} (consistent secondary
adaptation): The assumption is trivial if $c_{i}>0$. Assume that
$c_{i}<0$. Let $\alpha\in A^{n}$. $i\in N$ and $\hat{x}_{i}\in X_{i}$.
Let $x$ be an $\alpha$-equilibrium. The fact that $x$ is an $\alpha$-equilibrium
implies that for each $j\in N$: 
\[
BR_{j}^{\alpha_{j}}\left(\hat{x}_{i},x_{-i}\right)=\frac{a_{j}+c_{j}\sum_{k\neq j}w_{k}x_{k}+c_{j}w_{i}\left(\hat{x}_{i}-x_{i}\right)}{\left(1+\alpha_{j}\right)\left(b_{j}-w_{j}c_{j}\right)}=x_{j}+\frac{c_{j}w_{i}\left(\hat{x}_{i}-x_{i}\right)}{\left(1+\alpha_{j}\right)\left(b_{j}-w_{j}c_{j}\right)}.
\]
The secondary adaptation is in the same direction as the original
adaptation iff 
\[
sgn\left(\sum_{j\neq i}BR_{j}^{\alpha_{j}}\left(\hat{x}_{i},x_{-i}\right)+\hat{x}_{i}-\sum_{j}x_{j}\right)=sgn\left(\hat{x}_{i}-x_{i}\right)\Leftrightarrow
\]
\[
\left|\sum_{j\neq i}\left(BR_{j}^{\alpha_{j}}\left(\hat{x}_{i},x_{-i}\right)-x_{j}\right)\right|<\left|\left(\hat{x}_{i}-x_{i}\right)\right|\Leftrightarrow\left|w_{i}\sum_{j\neq i}\frac{c_{j}\left(\hat{x}_{i}-x_{i}\right)}{\left(1+\alpha_{j}\right)\left(b_{j}-w_{j}c_{j}\right)}\right|<\left|\left(\hat{x}_{i}-x_{i}\right)\right|
\]
\[
\Leftrightarrow w_{i}\sum_{j\neq i}\frac{\left|c_{j}\right|}{\left(1+\alpha_{j}\right)\left(b_{j}-w_{j}c_{j}\right)}<\frac{1}{w_{i}},
\]
which always hold due to the assumption that $\sum_{j\neq i}\frac{\left|c_{j}\right|}{b_{j}+w_{j}\left|c_{j}\right|}<\frac{1}{w_{i}}.$
\end{enumerate}
Proposition \ref{pro:alpha} implies that $\alpha_{i}^{*}<1$ (Part
(1)). Parts (2) and (3) immediately implied from Proposition \ref{prop:Pareto-dominate}
if $c_{i}>0$. We are left with proving parts (2) and (3) when $c_{i}<0$.
Eq. (\ref{eq:x-mean-equation-alpha}) implies that $\bar{x}$ is the
unique solution to 
\[
0=f\left(\bar{x},\alpha\right)\equiv\sum_{i\in N}\frac{w_{i}\left(a_{i}+c_{i}\bar{x}\right)}{b_{i}+\alpha_{i}\left(b_{i}-c_{i}w_{i}\right)}-\bar{x}.
\]
Observe that $f\left(\bar{x},\alpha\right)$ is decreasing in both
parameters, which implies that increasing $\alpha_{i}$ decreases
the unique $\bar{x}$ satisfying $f\left(\bar{x},\alpha\right)=0$.
Thus, $\bar{x}\left(\alpha\right)$ is decreasing in each $\alpha_{i}$,
which implies (due to Eq. (\ref{eq:x_i_NE-1})) that each $x_{j}$
is decreasing in each $\alpha_{i}$. The fact that $\alpha_{i}^{*}<1$
for each $i$ implies Part (2): $x_{i}^{*}>x_{i}^{NE}$. Finally Part
(3) is implied by 
\[
\pi_{i}\left(x^{*}\right)<\pi_{i}\left(x_{i}^{*},x_{-i}^{NE}\right)<\pi_{i}\left(x^{NE}\right),
\]
 where the first inequality is due to the positive externalities and
$x_{i}^{*}>x_{i}^{NE}$ and the second inequality is implied by $x_{i}^{NE}$
being the unique best reply to $x_{-i}^{NE}$.

\subsection{\label{subsec:proof-of-symmetric-oligopoly} Proof of Proposition
\ref{prop:symmetric-oligopoly-prices} (Symmetric Oligopoly)}

An analogous argument to Lemma \ref{lem:unqiue-alpha-eq-Price-competition-game}
implies that exists a unique $\left(x_{i},\alpha_{-i}\right)$-equilibrium
for each $x_{i}\in X_{i}$and $\alpha_{-i}\in A^{n-1}$. Let $x$
be an $\left(x_{i},\alpha_{-i}\right)$-equilibrium in which all firms
have the same level of biasedness (i.e., $\alpha_{j}-\alpha_{k}$
for each $j,k\neq i$). Eq. (\ref{eq:biased-BR-oligopoly}) implies
that for each $j\neq i$: 

\[
x_{j}=BR_{j}^{\alpha_{j}}\left(x_{-j}\right)=\frac{a_{j}+c_{j}\sum_{k\neq j}\frac{1}{n}x_{k}}{\left(1+\alpha_{j}\right)\left(b_{j}-\frac{1}{n}c_{j}\right)}=\frac{a+\tilde{c}\sum_{k\neq j}x_{k}}{\left(1+\alpha_{j}\right)\tilde{b}},
\]
where we simplified the notation by having $\tilde{c}=\frac{c_{j}}{n}$
and $\tilde{b}=b_{j}-\frac{c_{j}}{n}$. By symmetry, $x_{j}=x_{k}$
for each $j,k\neq i$. This implies that 
\[
x_{j}=\frac{a+\tilde{c}\left(\left(n-2\right)x_{j}+x_{i}\right)}{\left(1+\alpha_{j}\right)\tilde{b}}\,\Leftrightarrow\,x_{j}=\frac{a+\tilde{c}x_{i}}{\left(1+\alpha_{j}\right)\tilde{b}-\tilde{c}\left(n-2\right)}.
\]

This implies that
\[
x_{j}\left(x_{i},\alpha_{-i}^{*}\right)=\frac{a+\tilde{c}x_{i}}{\left(1+\alpha_{j}^{*}\right)\tilde{b}-\tilde{c}\left(n-2\right)}\,\Rightarrow\,\frac{dx_{j}\left(x_{i},\alpha_{-i}^{*}\right)}{dx_{i}}=\frac{1}{\left(1+\alpha_{j}^{*}\right)\frac{\tilde{b}}{\tilde{c}}-\left(n-2\right)},
\]
which, in turn, implies due to Claim \ref{pro:stackelberg-leader}
that $x_{j}^{*}$ must satisfy 
\begin{equation}
\alpha_{j}^{*}-1=\sum_{j\neq i}\frac{\text{d\ensuremath{x_{j}\left(x_{i},\alpha_{-i}^{*}\right)}}}{dx_{i}}\cdot\frac{\frac{\partial q_{i}}{\partial x_{j}}}{\frac{\partial q_{i}}{\partial x_{i}}}=\sum_{j\neq i}\frac{1}{\left(1+\alpha_{j}^{*}\right)\frac{\tilde{b}}{\tilde{c}}-\left(n-2\right)}\cdot\frac{\tilde{c}}{-\tilde{b}}\label{eq:alpha-condition-1}
\end{equation}
\[
=-\frac{\left(n-1\right)\frac{\tilde{c}}{\tilde{b}}}{\left(1+\alpha_{j}^{*}\right)\frac{\tilde{b}}{\tilde{c}}-\left(n-2\right)}\Rightarrow\left(1-\left(\alpha_{j}^{*}\right)^{2}\right)\frac{\tilde{b}}{\tilde{c}}-\left(1-\alpha_{j}^{*}\right)\left(n-2\right)=\left(n-1\right)\frac{\tilde{c}}{\tilde{b}}
\]
\[
\Rightarrow\frac{\tilde{b}}{\tilde{c}}\left(\alpha_{j}^{*}\right)^{2}-\alpha_{j}^{*}\left(n-2\right)+\left(n-1\right)\frac{\tilde{c}}{\tilde{b}}+n-2-\frac{\tilde{b}}{\tilde{c}}=0
\]
\[
\Rightarrow\alpha_{j}^{*}=\frac{n-2\pm\sqrt{\left(n-2\right)^{2}-4\left(n-1\right)-4\frac{\tilde{b}}{\tilde{c}}\left(n-2-\frac{\tilde{b}}{\tilde{c}}\right)}}{2\frac{\tilde{b}}{\tilde{c}}}
\]
\[
\Rightarrow\alpha_{j}^{*}\left(\left|\frac{\tilde{b}}{\tilde{c}}\right|,n\right)=\frac{n-2+\sqrt{\left(n-2\right)^{2}-4\left(n-1\right)-4\left|\frac{\tilde{b}}{\tilde{c}}\right|\left(n-2-\left|\frac{\tilde{b}}{\tilde{c}}\right|\right)}}{2\left|\frac{\tilde{b}}{\tilde{c}}\right|}.
\]
.It is straightforward to verify that $\alpha_{j}^{*}\left(\left|\frac{\tilde{b}}{\tilde{c}}\right|,n\right)$
is increasing in both parameters and that $\lim_{n\rightarrow\infty}\alpha_{i}^{*}\left(\left|\frac{\tilde{b}}{\tilde{c}}\right|,n\right)=\lim_{\left|\frac{b}{c}\right|\rightarrow\infty}\alpha_{i}^{*}\left(\left|\frac{\tilde{b}}{\tilde{c}}\right|,n\right)=1$.
The fact that $\left|\frac{\tilde{b}}{\tilde{c}}\right|=\left|\frac{b-\frac{c}{n}}{\frac{c}{n}}\right|=n\left|\frac{b}{c}\right|-1$
is an increasing function of $n$ and $\left|\frac{b}{c}\right|$
implies that $\alpha_{i}^{*}\left(\left|\frac{b}{c}\right|,n\right)$
is increasing in both parameters and that $\lim_{n\rightarrow\infty}\alpha_{i}^{*}\left(\left|\frac{b}{c}\right|,n\right)=\lim_{\left|\frac{b}{c}\right|\rightarrow\infty}\alpha_{i}^{*}\left(\left|\frac{b}{c}\right|,n\right)=1$
  Further observe that (\ref{eq:biased-BR-oligopoly}) implies
that the symmetric NAE and NE prices are 
\begin{equation}
x_{j}^{*}=\frac{a}{\left(1+\alpha\right)b-c\left(1+\frac{\alpha}{n}\right)},\,\,x_{j}^{NE}=\frac{a}{2b-\left(1+\frac{1}{n}\right)c}.\label{eq:symmetric-NAE-prices}
\end{equation}

\subsection{\label{subsec:Proof-of-advertising-result}Proof of Prop. \ref{prop:advertising-result}
(Advertising Competition)}
\begin{lem}
\label{lem:unqiue-alpha-eq-Price-competition-game-1}Game $\Gamma_{a}$
admits a unique $\alpha$-equilibrium for any $\alpha\in\mathbb{R}_{++}$.
\end{lem}
\begin{proof}
The robust concavity of the payoff function (proved below) implies
that any $\alpha$-equilibrium is fully characterized by the FOC

\begin{equation}
0=\frac{d\pi_{i}^{\alpha_{i}}}{dx_{i}}=p_{i}\alpha_{i}\frac{b_{i}+c_{i}\sqrt{x_{-i}}}{2\sqrt{x_{i}}}-1\Rightarrow\sqrt{x_{i}}=\frac{p_{i}\alpha_{i}\left(b_{i}+c_{i}\sqrt{x_{-i}}\right)}{2},\label{eq:x_i_NE-1-1}
\end{equation}
Substituting $\sqrt{x_{-i}}=\frac{p_{-i}\alpha_{-i}\left(b_{-i}+c_{-i}\sqrt{x_{i}}\right)}{2}$
yields the following unique solution: 
\begin{equation}
2\sqrt{x_{i}}=p_{i}\alpha_{i}\left(b_{i}+\frac{c_{i}p_{-i}\alpha_{-i}\left(b_{-i}+c_{-i}\sqrt{x_{i}}\right)}{2}\right)\Rightarrow x_{i}\left(\alpha\right)=\left(\frac{p_{i}\alpha_{i}\left(2b_{i}+c_{i}p_{-i}\alpha_{-i}b_{-i}\right)}{4-p_{i}p_{-i}\alpha_{i}\alpha_{-i}c_{i}c_{-i}}\right)^{2}.\label{eq:x_i_alpha-advertsiing}
\end{equation}
 When substituting $\alpha=\overrightarrow{1}$, this implies, the
existence of a unique Nash equilibrium: $x_{i}^{NE}=\left(\frac{p_{i}\left(2b_{i}+c_{i}p_{-i}b_{-i}\right)}{4-p_{i}p_{-i}c_{i}c_{-i}}\right)^{2}.$
\end{proof}

We begin by showing that $\Gamma_{a}$ satisfies Assumptions \ref{assu:Monotone-externalities}--\ref{assu:consistent-second-adaptation}:
\begin{enumerate}
\item Assumption \ref{assu:Monotone-externalities}: $\frac{d\pi_{i}}{dx_{-i}}=\frac{d\pi_{i}}{dq_{i}}\frac{dq_{i}}{dx_{-i}}=p_{i}\frac{dq_{i}}{dx_{-i}}=p_{i}\frac{c_{i}\sqrt{x_{i}}}{2\sqrt{x_{-i}}}\,\Rightarrow$
$\textrm{sgn}\left(\frac{d\pi_{i}}{dx_{-i}}\right)=\textrm{sgn}\left(c_{i}\right).$
\item Assumption \ref{assu:Monotone-partial-deriv}: (1) $\frac{\partial\pi_{i}}{\partial x_{i}}=-1<0$,
and (2) $\frac{\partial\pi_{i}}{\partial q_{i}}\frac{\partial q_{i}}{\partial x_{i}}=p_{i}\left(\frac{b_{i}+c_{i}\sqrt{x_{-i}}}{2\sqrt{x_{i}}}\right)>0$,
where the inequality is implied by the assumption that either $c_{i}>0$
or $x_{-i}\leq M_{-i}=\frac{b_{i}}{\left|c_{i}\right|}$.
\item Assumption \ref{assu:Strong-concavity-prime} (which implies Assumption
\ref{assu:Strong-concavity} of robust concavity): (1) $\frac{d}{dx_{i}}\left(\frac{\partial\pi_{i}}{\partial x_{i}}\right)=0$,
and (2) $\frac{d}{dx_{i}}\left(\frac{\partial\pi_{i}}{\partial q_{i}}\cdot\frac{\partial q_{i}}{\partial x_{i}}\right)=\frac{d}{dx_{i}}\left(p_{i}\left(\frac{b_{i}+c_{i}\sqrt{x_{-i}}}{2\sqrt{x_{i}}}\right)\right)<0$.
\item Assumption \ref{assu:Strong-com-subs-prime}: (1) $\frac{d}{dx_{-i}}\left(\frac{\partial\pi_{i}}{\partial x_{i}}\right)=\frac{d}{dx_{-i}}\left(-1\right)=0$,
(2) $\frac{d}{dx_{-i}}\left(\frac{b_{i}+c_{i}\sqrt{x_{-i}}}{2\sqrt{x_{i}}}\right)=\frac{c_{i}}{4\sqrt{x_{i}}\sqrt{x_{-i}}}$.
\item Assumption \ref{assu:bounded-percieved-BR}: For each player $i\in N$
and $\alpha_{i}>0$:
\[
0=\frac{d\pi_{i}^{\alpha_{i}}}{dx_{i}}=p_{i}\alpha_{i}\frac{b_{i}+c_{i}\sqrt{x_{-i}}}{2\sqrt{x_{i}}}-1\Rightarrow\sqrt{BR_{i}^{\alpha_{i}}\left(x_{-i}\right)}=\frac{p_{i}\alpha_{i}\left(b_{i}+c_{i}\sqrt{x_{-i}}\right)}{2}
\]
If $c_{i}>0$, then an analogous argument to Lemma \ref{lem-bounded-BR-pos-B_i}
(where $\sqrt{BR_{i}^{\alpha_{i}}\left(x_{-i}\right)}$ and $\sqrt{x_{-i}}$
replaces $BR_{i}^{\alpha_{i}}\left(x_{-i}\right)$ and $x_{-i}$)
implies that Assumption \ref{assu:bounded-percieved-BR} is satisfied
with respect to any $\epsilon>0$ and
\[
\underline{x}_{i}^{\alpha}=\left(\frac{p_{i}\alpha_{i}b_{i}}{2}\right)^{2},\,\,\,\overline{x}_{i}^{\alpha}=\left(\frac{\max_{i\in N}\left(\frac{p_{i}\alpha_{i}b_{i}}{2}\right)+\epsilon}{1-\max_{i\in N}\left(\frac{p_{i}\alpha_{i}c_{i}}{2}\right)}\right)^{2},
\]
If $c_{i}<0$, then Lemma \ref{lem-bounded-BR-pos-B_i} implies that
Assumption \ref{assu:bounded-percieved-BR} is satisfied with respect
to. 
\[
\underline{x}_{i}^{\alpha}=\left(\min_{i\in N}\left(\frac{p_{i}\alpha_{i}b_{i}}{2}\left(1-\frac{p_{i}\alpha_{i}c_{i}}{2}\right)\right)\right)^{2},\,\,\,\overline{x}_{i}^{\alpha}=\left(\frac{p_{i}\alpha_{i}b_{i}}{2}\right)^{2}.
\]
\item Assumption \ref{assu:consistent-second-adaptation} (consistent secondary
adaptation): The assumption is trivial due to having two players.
\end{enumerate}
Next we prove part (1) (which implies parts (2-3) due to Proposition
\ref{prop:Pareto-dominate}). Taking the derivative of Eq. (\ref{eq:x_i_NE-1-1})
implies that $\frac{dx_{-i}\left(x_{i},\alpha_{-i}^{*}\right)}{dx_{i}}=\frac{\sqrt{x_{-i}}}{\sqrt{x_{i}}}\alpha_{-i}^{*}\frac{p_{-i}}{2}c_{-i}$.
It is immediate that $\frac{\partial q_{i}}{\partial x_{-i}}=\frac{c_{i}\sqrt{x_{i}}}{2\sqrt{x_{-i}}}$
and $\frac{\partial q_{i}}{\partial x_{-i}}=\frac{1}{2\alpha_{i}^{*}\frac{p_{i}}{2}}$.
Thus, Claim \ref{claim:Stackelberg-derivative} implies that $x_{i}^{*}$
must satisfy 
\begin{equation}
\alpha_{i}^{*}-1=\frac{\sqrt{x_{-i}}}{\sqrt{x_{i}}}\alpha_{-i}^{*}\frac{p_{-i}}{2}c_{-i}\frac{c_{i}\sqrt{x_{i}}2\alpha_{i}^{*}\frac{p_{i}}{2}}{2\sqrt{x_{-i}}}=\alpha_{1}^{*}\alpha_{2}^{*}\frac{p_{1}p_{2}c_{1}c_{2}}{4}.\label{eq:alpha-condition-2}
\end{equation}
 Observe that the RHS of (\ref{eq:alpha-condition-2}) remains the
same when swapping $i$ and $j$. This implies that $\alpha_{j}^{*}$
and $\alpha_{i}^{*}$ must be equal. The resulting one-variable quadratic
equation has two solutions: $\alpha_{1}^{*}=\alpha_{2}^{*}=\frac{2}{1+\sqrt{1-c_{1}c_{2}p_{1}p_{2}}}\in\left(1,2\right)$
and $\hat{\alpha}_{1}=\hat{\alpha}_{2}=\frac{2}{1-\sqrt{1-c_{1}c_{2}p_{1}p_{2}}}>2$,
and it is easy to verify that only the first solution $\alpha_{1}^{*}=\alpha_{2}^{*}$
satisfies the SOC. 

\subsection{\label{subsec:Proof-of-Proposition-team-production}Proof of Proposition
(Team Production)}
\begin{enumerate}
\item Assumption \ref{assu:Monotone-externalities} (monotone externalities):
$\frac{d\pi_{i}}{dx_{-i}}=\frac{d\pi_{i}}{dq_{i}}\frac{dq_{i}}{dx_{-i}}=\frac{dq_{i}}{dx_{-i}}>0$.
\item Assumption \ref{assu:Monotone-partial-deriv}: (1) $\frac{\partial\pi_{i}}{\partial x_{i}}=-1<0$,
and (2) $\frac{\partial\pi_{i}}{\partial q_{i}}\frac{\partial q_{i}}{\partial x_{i}}=\frac{\partial q_{i}}{\partial x_{i}}>0$.
\item Assumption 3' (which implies Assumption \ref{assu:Strong-concavity}
of robust concavity): (1) $\frac{d}{dx_{i}}\left(\frac{\partial\pi_{i}}{\partial x_{i}}\right)=0$,
and (2) $\frac{d}{dx_{i}}\left(\frac{\partial\pi_{i}}{\partial q_{i}}\cdot\frac{\partial q_{i}}{\partial x_{i}}\right)=\frac{d}{dx_{i}}\left(\frac{\partial q_{i}}{\partial x_{i}}\right)=\frac{d^{2}q_{i}}{dx_{i}^{2}}<0$.
\item Assumption 4': (1) $\frac{d}{dx_{j}}\left(\frac{\partial\pi_{i}}{\partial x_{i}}\right)=\frac{d}{dx_{j}}\left(-1\right)=0$,
(2) $\frac{d}{dx_{j}}\left(\frac{\partial\pi_{i}}{\partial q_{i}}\frac{\partial q_{i}}{\partial x_{i}}\right)=\frac{d}{dx_{j}}\left(\frac{\partial q_{i}}{\partial x_{i}}\right)=\frac{d^{2}q_{i}}{dx_{i}dx_{j}}>0$.
\item Assumption \ref{assu:bounded-percieved-BR} (bounded perceived best
replies): Fix bias profile $\alpha\in A^{n}$. For each player $i\in N$,
let $z_{i}$ be the true demand sensitivity that is perceived as equal
to one by the biased player $i$, i.e., $\frac{\partial q_{i}^{\alpha_{i}}}{\partial x_{i}}\left(z_{i}\right)=1$.
Due to Assumption \ref{assu:marginal-contirbution-small-large-efforts},
there exist symmetric profiles $\underline{x}_{\alpha}<\bar{x}_{\alpha}$
such that for each player $i\in N$: 
\[
\frac{dq_{i}\left(\overline{x}_{\alpha}\right)}{dx_{i}}<z_{i}<\frac{dq_{i}\left(\underline{x}_{\alpha}\right)}{dx_{i}}\,\,\Rightarrow\,\,\frac{dq_{i}^{\alpha_{i}}\left(\overline{x}_{\alpha}\right)}{dx_{i}}<1<\frac{dq_{i}^{\alpha_{i}}\left(\underline{x}_{\alpha}\right)}{dx_{i}}
\]
\[
\Rightarrow\,\,\frac{d\pi_{i}^{\alpha_{i}}\left(\overline{x}_{\alpha}\right)}{dx_{i}}<0<\frac{d\pi_{i}^{\alpha_{i}}\left(\underline{x}_{\alpha}\right)}{dx_{i}}\,\,\Rightarrow\,\,BR_{i}^{\alpha_{i}}\left(\underline{x}_{\alpha}\right),BR_{i}^{\alpha_{i}}\left(\overline{x}_{\alpha}\right)\in\left(\underline{x}_{\alpha},\overline{x}_{\alpha}\right).
\]
\item Assumption \ref{assu:consistent-second-adaptation} is trivial due
to having strategic complements. 
\end{enumerate}

\subsection{\label{subsec:Proof-of-Proposition-merger}Proof of Proposition \ref{prop:merger}
(Market Structure Analysis)}

Appendix \ref{subsec:proof-of-symmetric-oligopoly} shows that in
an NAE the equilibrium biases pre-merger are 
\[
\alpha_{i}^{pre}\left(\left|\frac{\tilde{b}}{\tilde{c}}\right|,n\right)=\frac{n-2+\sqrt{\left(n-2\right)^{2}-4\left(n-1\right)-4\left|\frac{\tilde{b}}{\tilde{c}}\right|\left(n-2-\left|\frac{\tilde{b}}{\tilde{c}}\right|\right)}}{2\left|\frac{\tilde{b}}{\tilde{c}}\right|}=\frac{c+\sqrt{c^{2}-36bc+36b^{2}}}{6b-2c}
\]
 for $n=3$, $\tilde{b}=b-\frac{c}{n}$ and $\tilde{c}=\frac{c}{n}$,
resulting in prices $x_{i}^{pre}=\frac{a}{\left(1+\alpha^{*}\right)b-c\left(1+\frac{\alpha^{*}}{n}\right)}=\frac{3a}{3b(1+\alpha^{*})-c(3+\alpha^{*})}=\frac{6a}{\sqrt{36b^{2}-36bc+c^{2}}+6b-5c}$.
Standard analysis shows that when the firm payoffs are $\pi_{i}^{mc}=(x_{i}^{mc}-mc_{i})q_{i}$
then $BR_{i}(x_{-i}^{mc})=\frac{a+\tilde{b\cdot}mc_{i}+\tilde{c}\sum_{j\neq i}x_{j}^{mc}}{2\tilde{b}}$.
In a symmetric Nash equilibrium, and assuming all firms have the same
marginal costs, we can sum up the conditions $x_{i}^{mc}=BR_{i}(x_{-i}^{mc})$
to receive 
\[
\bar{x}^{mc}=\frac{a+\tilde{b\cdot}mc_{i}+\tilde{c}(n-1)\bar{x}^{mc}}{2\tilde{b}}\Rightarrow\bar{x}^{mc}=\frac{a+\tilde{b}\cdot mc}{2\tilde{b}-(n-1)\tilde{c}}.
\]
By solving $\bar{x}^{pre}=\bar{x}^{mc}$, the economist estimates
the marginal cost as 
\[
mc_{i}=\frac{3a\left(6b-3c-\sqrt{36b^{2}-36bc+c^{2}}\right)}{(3b-c)\left(\sqrt{36b^{2}-36bc+c^{2}}+6b-5c\right)}
\]
 which is always positive when $b>|c|$ and $a>0.$

When firms $2$ and $3$ merge, because the demand for goods $2$
and $3$ is symmetric, we can assume that they will set the same price
$x_{23}^{post}$ for both goods. The joint firm's true payoff will
be $\pi_{23}^{post}=x_{23}^{post}(q_{2}+q_{3})$, while firm $1$
remains with the same payoff, yielding weights $w_{1}=1/3$ and $w_{23}=2/3$
in \eqref{eq:linear-bertrand-demand}. Proposition \eqref{prop:price-competetion-unqiue-NAE}
shows that for a duopoly, the long run unique NAE yields biases $\alpha_{1}^{post}=\alpha_{23}^{post}=\sqrt{1-\frac{2c^{2}}{(3b-c)(3b-2c)}}$.
Comparing to $\alpha^{pre}$, we find that $\alpha^{pre}>\alpha^{post}\iff c>0$.
Applying the implicit function theorem to \eqref{eq:x-mean-equation-alpha},
let $g=\sum_{i\in N}\frac{w_{i}\left(a_{i}+c_{i}\bar{x}\right)}{b_{i}+\alpha\left(b_{i}-c_{i}w_{i}\right)}-\bar{x}$,
then $\frac{d\bar{x}}{d\alpha}=-\frac{\frac{\partial g}{\partial\alpha}}{\frac{\partial g}{d\bar{x}}}$.
Because $b_{i}>w_{i}c_{i}$, $a_{i}>0$, and $\bar{x}\ge0$, then
$-\frac{\partial g}{d\alpha}=\sum_{i\in N}\frac{\left(b_{i}-c_{i}w_{i}\right)w_{i}\left(a_{i}+c_{i}\bar{x}\right)}{\left(b_{i}+\alpha\left(b_{i}-c_{i}w_{i}\right)\right)^{2}}>0$,
while 
\[
\frac{\partial g}{d\bar{x}}=-\frac{(1+\alpha)^{2}\left(b_{1}w_{23}(b_{23}-c_{23})+b_{23}w_{1}(b_{1}-c_{1})\right)+\alpha(2+\alpha)w_{1}w_{23}c_{1}c_{23}}{(b_{1}+(b_{1}-c_{1}w_{1})\alpha)(b_{23}+(b_{23}-c_{23}w_{23})\alpha)}<0.
\]
 Hence, $\alpha^{pre}>\alpha^{post}$ implies $\bar{x}^{post}(\alpha^{pre})<\bar{x}^{post}(\alpha^{post})$
and by Equation \eqref{eq:x_i_NE-1}, $x_{i}^{post}(\alpha^{post})>x_{i}^{post}(\alpha^{pre})$.
When the economist predicts the post-merger prices by assuming the
merged firm's payoff is $\pi_{23}^{mc}=(x_{23}^{mc}-mc_{23})(q_{2}+q_{3})$
while firm $1$ has the same payoff function as before the merger,
the resulting equilibrium prices are 
\[
x_{1}^{mc,post}=\frac{a\left(c^{2}\left(\sqrt{36b^{2}-36bc+c^{2}}-5c\right)-144b^{2}c+108b^{3}+54bc^{2}\right)}{(3b-c)\left(6b^{2}-6bc+c^{2}\right)\left(\sqrt{36b^{2}-36bc+c^{2}}+6b-5c\right)},
\]
\[
x_{23}^{mc,post}=\frac{a\left(c\left(\sqrt{36b^{2}-36bc+c^{2}}+7c\right)+36b^{2}-36bc\right)}{\left(6b^{2}-6bc+c^{2}\right)\left(\sqrt{36b^{2}-36bc+c^{2}}+6b-5c\right)}.
\]
The difference $x_{1}^{post}(\alpha^{pre})-x_{1}^{mc,post}=${\scriptsize{}
\[
\frac{a(3b-2c)\left(594b^{2}c^{2}+\left(126b^{2}c-108b^{3}-48bc^{2}+7c^{3}\right)\sqrt{36b^{2}-36bc+c^{2}}-1080b^{3}c+648b^{4}-138bc^{3}+13c^{4}\right)}{(3b-c)\left(6b^{2}-6bc+c^{2}\right)\left(18b^{2}-18bc+5c^{2}\right)\left(\sqrt{36b^{2}-36bc+c^{2}}+6b-5c\right)},
\]
} which is positive when $b>c>0$ as confirmed with Mathematica software
(code in supplementary Appendix). . Similarly, the difference, $x_{23}^{post}(\alpha^{pre})-x_{23}^{mc,post}=${\scriptsize{}
\[
\frac{2ac\left(\left(17c^{3}-144b^{2}c+108b^{3}+15bc^{2}\right)-\left(18b^{2}-15bc+c^{2}\right)\sqrt{36b^{2}-36bc+c^{2}}\right)}{\left(6b^{2}-6bc+c^{2}\right)\left(\sqrt{36b^{2}-36bc+c^{2}}+6b-5c\right)\left((7c-12b)\sqrt{36b^{2}-36bc+c^{2}}-72b^{2}+90bc-23c^{2}\right)}
\]
} is also positive as confirmed with Mathematica software (code in
supplementary Appendix).
\end{document}